\documentclass[10pt, a4paper]{article}

\usepackage[german,english]{babel} % English and German language 
\usepackage{booktabs} % Horizontal rules in tables 
\usepackage{comment} % Necessary to comment several paragraphs at once
\usepackage[utf8]{inputenc} % Required for international characters
\usepackage[T1]{fontenc} % Required for output font encoding for international characters
\usepackage{caption, soul}
\usepackage{authblk}

\usepackage{bigints,bbm,slashed,mathtools,amssymb,amsmath,amsfonts,amsthm} % Math packages which might be useful for equations
\usepackage{physics, upgreek}
\usepackage[mathscr]{euscript}
\usepackage{cancel}
\usepackage{url}
\usepackage{enumerate}
\usepackage{tikz} % For tikz figures (to draw arrow diagrams, see a guide how to use them)
\usepackage{tikz-cd}
\usepackage{pgfplots}
\pgfplotsset{compat=1.16}
\usetikzlibrary{babel}
\usetikzlibrary{positioning,arrows} % Adding libraries for arrows
\usetikzlibrary{decorations.pathreplacing, decorations.markings} % Adding libraries for decorations and paths
\usetikzlibrary{patterns}
\usepackage{tikzsymbols} % For amazing symbols ;) https://mirror.hmc.edu/ctan/graphics/pgf/contrib/tikzsymbols/tikzsymbols.pdf 
\usepackage{blindtext} % To add some blind text in your paper

\usepackage{geometry} % Necessary package for defining margins

\usepackage{setspace} % Required for spacing
\usepackage[T1]{fontenc} % Output font encoding for international characters
\usepackage[utf8]{inputenc} % Required for inputting international characters

\usepackage{XCharter} % Use the XCharter font

\usepackage{fancyhdr} % Package is needed to define header and footer
\usepackage[hang,flushmargin]{footmisc}

\usepackage[square, numbers]{natbib} % Use natbib for citing

\usepackage{har2nat} % Allows to use harvard package with natbib https://mirror.reismil.ch/CTAN/macros/latex/contrib/har2nat/har2nat.pdf
\usepackage[colorlinks=true, citecolor=blue]{hyperref}

\numberwithin{equation}{section}
\theoremstyle{definition}
\newtheorem{definition}{Definition}[section]
\theoremstyle{plain}
\newtheorem{theorem}{Theorem}[section]
\newtheorem{proposition}[theorem]{Proposition}
\newtheorem{lemma}[theorem]{Lemma}
\newtheorem{corollary}[theorem]{Corollary}
\newtheorem{conjecture}[theorem]{Conjecture}

\newtheorem{innercustomgeneric}{\customgenericname}
\providecommand{\customgenericname}{}
\newcommand{\newcustomtheorem}[2]{%
  \newenvironment{#1}[1]
  {%
   \renewcommand\customgenericname{#2}%
   \renewcommand\theinnercustomgeneric{##1}%
   \innercustomgeneric
  }
  {\endinnercustomgeneric}
}

\newcustomtheorem{customthm}{Theorem}
\newcustomtheorem{customlemma}{Lemma}
\newcustomtheorem{customcorollary}{Corollary}

\theoremstyle{remark}
\newtheorem*{remark}{Remark}

\newcommand{\R}{\mathbb{R}}

\newcommand{\N}{\mathbb{N}}

\DeclareFontFamily{U}{mathx}{}
\DeclareFontShape{U}{mathx}{m}{n}{<-> mathx10}{}
\DeclareSymbolFont{mathx}{U}{mathx}{m}{n}
\DeclareMathAccent{\widehat}{0}{mathx}{"70}
\DeclareMathAccent{\widecheck}{0}{mathx}{"71}

\title{BKL bounces outside homogeneity: \\ Einstein--Maxwell--scalar field in surface symmetry}
\author{Warren Li}
\affil{\small Princeton University, Department of Mathematics, Fine Hall, Washington Road, Princeton, NJ 08544, USA}
\begin{document}

\maketitle

\begin{abstract}
    We study the phenomenon of bounces, as predicted by Belinski, Khalatnikov and Lifshitz (BKL) in the study of singularities arising from Einstein's equations, as an instability mechanism within the setting of the (inhomogeneous) Einstein--Maxwell--scalar field system in surface symmetry. This article can be viewed as a companion to our other article \cite{MeGowdyPaper}, where we study bounces for the Einstein vacuum equations in Gowdy symmetry. That is, we show many features of such bounces generalize to the matter model described, albeit in a different symmetry class. The articles may be read independently.

    In analogy to \cite{MeGowdyPaper}, we describe a wide class of inhomogeneous initial data which permit formation of a spacelike singularity, but such that the dynamics towards different spatial points at the singularity are well--described by independent nonlinear ODEs reminiscent of BKL bounces. A major ingredient is the proof of so-called Asymptotically Velocity Term Dominated behaviour even in the presence of such bounces, though one difference from \cite{MeGowdyPaper} is that our model does not permit the existence of so-called ``spikes''. % Solutions arising from such data are stable in the sense that the singularity (and related curvature blowup) , but unstable in the sense that the quantitative asymptotics near singularity have sensitive dependence on the nonlinearity in the ODEs. %, and that such dynamics may be appropriately localised to timelike curves approaching the singularity.

    One particular application is the study of (past) instability of certain generalized Kasner spacetimes with no electromagnetic field present. Perturbations of such spacetimes are such that the singularity persists, but for perturbations with electromagnetism turned on the intermediate dynamics -- between data and the singularity -- features up to one BKL-like bounce. This is in analogy with the instability of polarized Gowdy spacetimes due to non-polarized perturbations in \cite{MeGowdyPaper}.
\end{abstract}

\setcounter{tocdepth}{2}
\tableofcontents
%auto-ignore
\section{Introduction} \label{sec:intro}

\subsection{The Einstein--Maxwell--scalar field system} \label{sub:intro_einstein}

We study the structure and (in)stability properties of spacelike singularities arising within the context of the Einstein--Maxwell--scalar field system in surface symmetry. The study of spacelike singularities has a long tradition in the context of the Einstein vacuum equations, particularly in Gowdy symmetry, and we refer the reader to our companion paper \cite{MeGowdyPaper} whose results are somewhat in parallel. See also the ``dictionary'' in Section~\ref{sub:intro_gowdy}. In the present paper, we show that many of the phenomena associated to spacelike singularities in vacuum are also present in the study of the Einstein--Maxwell--scalar field equations with unknowns $(\mathcal{M}, \mathbf{g}, \phi, \mathbf{F})$.

Here $\mathbf{g}$ is a Lorentzian metric solving the Einstein equations on $\mathcal{M}^{1+3}$, sourced by two matter fields: a scalar field $\phi: \mathcal{M} \to \R$ solving the usual wave equation and an electromagnetic field given by the $2$-form $\mathbf{F} \in \Omega^2(\mathcal{M})$ solving the source free Maxwell equations. If $\mathbf{D}$ is the Levi-Civita connection associated to $\mathbf{g}$, with $\mathbf{Ric}[\mathbf{g}]$ and $\mathbf{R}[\mathbf{g}]$ its associated Ricci and scalar curvature, then the evolution equations are:
\begin{gather} 
    \mathbf{Ric}_{\mu\nu}[\mathbf{g}] - \frac{1}{2} \mathbf{R}[\mathbf{g}] \mathbf{g}_{\mu\nu} = 2 \, \mathbf{T}_{\mu\nu}[\mathbf{\Phi}], \label{eq:einstein}  \\[0.3em]
    \label{eq:wave}
    \square_{\mathbf{g}} \phi = (\mathbf{g}^{-1})^{\mu\nu} \mathbf{D}_{\mu} \mathbf{D}_{\nu} \phi = 0, \\[0.3em] \label {eq:maxwell}
    d\mathbf{F} = 0, \qquad \mathbf{D}^{\mu} \mathbf{F}_{\mu\nu} = 0.
\end{gather}

The quantity $\mathbf{T}[\mathbf{\Phi}]$ on the right hand side of the Einstein equations \eqref{eq:einstein} is the energy-momentum tensor associated to the matter fields $\mathbf{\Phi}$, and in our case is given by $\mathbf{T}_{\mu\nu} [\mathbf{\Phi}] = \mathbf{T}_{\mu\nu} [\phi] + \mathbf{T}_{\mu\nu} [\mathbf{F}]$, where
\begin{gather} \label{eq:wave_em}
    \mathbf{T}_{\mu\nu} [\phi] = \mathbf{D}_{\mu} \phi \, \mathbf{D}_{\nu} \phi - \frac{1}{2} \mathbf{D}_{\rho} \phi \, \mathbf{D}^{\rho} \phi \, \mathbf{g}_{\mu \nu}, \\[0.3em] \label{eq:maxwell_em}
    \mathbf{T}_{\mu\nu} [\mathbf{F}] = \mathbf{F}_{\mu}^{\phantom{\mu}\rho} \, \mathbf{F}_{\nu \rho} - \frac{1}{4} \mathbf{F}_{\rho \sigma} \, \mathbf{F}^{\rho \sigma} \, \mathbf{g}_{\mu\nu},
\end{gather}
represent the energy-momenta of scalar and electromagnetic matter respectively.

The Einstein--Maxwell--scalar field system \eqref{eq:einstein}--\eqref{eq:maxwell}, particularly in spherical symmetry, has long been used as a model system to study properties of spacetimes arising in General Relativity. See for instance \cite{dr_priceslaw, LukOh1} regarding the black hole exterior region of asymptotically flat spacetimes and Price's law, and \cite{dafermos03, LukOh2} considering the black hole interiors of such spacetimes. 
Together these works provide a proof of the Strong Cosmic Censorship conjecture\footnote{See \cite{penrose_cc} for an introduction to the Strong Cosmic Censorship conjecture, or \cite{Christodoulou_cc} for a more modern account.} for the spherically symmetric Einstein--Maxwell--scalar field system; the inextendibility of spacetime is due to the presence of a (weakly) singular null future boundary.

The Einstein--Maxwell--scalar field model above is admittedly a somewhat restrictive model. For instance, rigidity of the electromagnetic field in spherical symmetry means that for $\mathbf{F} \neq 0$, $(\mathcal{M}, \mathbf{g})$ cannot have a regular center of symmetry and thus precludes \emph{one-ended} gravitational collapse. Thus the works above and the Penrose diagrams in Figure~\ref{fig:twoended} must be \emph{two-ended}. A generalization of this model where $\mathbf{F}$ has nontrivial dynamics and thereby allows the study of one-ended complete initial data is the Einstein--Maxwell--\emph{charged} scalar field model, where $\mathbf{F}$ is nonlinearly sourced by $\phi$. Here there is also partial progress on Strong Cosmic Censorship and singular null boundaries due to Van de Moortel and Kehle--Van de Moortel \cite{Moortel18, MoortelChristoph, MoortelChristoph2}.

In any case such \emph{weak null singularities} are not the focus of the present paper. Instead, we consider \emph{spacelike singularities}. %In spherical (or surface) symmetry, spacelike singularities appear as a singular past or future boundary with spacelike character in a Penrose diagram. 
Such singularities are well studied in the Einstein--scalar field system in spherical symmetry (i.e.~with $\mathbf{F} = 0$), see for instance \cite{Christodoulou_formation, Christodoulou_BV, Dafermos_trapped} which show that spherically symmetric solutions to the Einstein--scalar field system containing a trapped surface typically terminate in a spacelike singularity. In fact Christodoulou showed that such solutions arising from \emph{generic} initial data either disperse\footnote{Dispersion here means having a Penrose diagram akin to that of Minkowski space.} or have the Penrose diagram of Figure~\ref{fig:christodoulou_collapse}, providing a complete resolution of the Weak and Strong Cosmic Censorship Conjectures in this model \cite{Christodoulou_wcc, Christodoulou_cc}.

\begin{figure}[ht] 
    \centering
    \begin{tikzpicture}[scale=0.7]
        \node (s) at (0,-0.5) [circle, draw, inner sep=0.5mm, fill=black] {};
        \node [below left=0.2mm of s] {$p_0$};
        \node (i0) at (8, 0) [circle, draw, inner sep=0.5mm] {};
        \node [below right=0.2mm of i0] {$i^0$};
        \node (i+) at (3, 5) [circle, inner sep=0.5mm, draw] {};
        \node [above right=0.2mm of i+] {$i^+$};
        \node (ss) at (0, 4.5) [circle, inner sep=0.5mm, draw] {};
        \node [above left=0.2mm of ss] {$b_0$};

        \path[fill=lightgray, opacity=0.5] (s) .. controls (2.7, 0.4) and (5.3, -0.5) .. (8, 0)
            -- (3, 5) .. controls (2, 5.5) and (1, 4.0) .. (0, 4.5) 
            .. controls (0.1, 2) .. (s);
        \path[fill=lightgray] (i+) .. controls (1.5, 4) and (0.5, 3.8)
            .. (0, 4.5)
            .. controls (1, 4.0) and (2, 5.5) .. (i+);

        \draw [thick] (s) .. controls (2.7, 0.4) and (5.3, -0.5) ..  (i0)
            node (sigma) [midway, below] {$\Sigma$};
        \draw [dashed] (i0) -- (i+) node (nullinf) [midway, above right] {$\mathcal{I}^+$};
        \draw [thick] (ss) .. controls (0.1, 2) .. (s) node [midway, left] {$\Gamma$};
        \draw (i+) -- (0.1, 2.1) node [midway, below right] {$\mathcal{H}^+$};
        \draw (i+) .. controls (1.5, 4) and (0.5, 3.8) .. (0, 4.5)
            node [midway, below] {$\mathcal{A}$};
        \draw [dashed] (i+) .. controls (2, 5.5) and (1, 4.0) .. (ss)
            node [midway, above] {$\mathcal{S}$};
    \end{tikzpicture}

    \captionsetup{justification = centering}
    \caption{Penrose diagram representation of a gravitational collapse solution to the Einstein-scalar field equations. This solution possesses a complete future null infinity $\mathcal{I}^+$ as well as a black hole region bounded to the past by the \textit{event horizon} $\mathcal{H}^+$. The black hole interior contains an \textit{apparent horizon} $\mathcal{A}$ and a (darkly shaded) \textit{trapped region} in which $\partial_v r < 0$ and which culminates at a spacelike singularity $\mathcal{S} = \{ r = 0 \}$.}% connecting $i^+$ with the first singularity at the center, denoted $b_0$.}
    \label{fig:christodoulou_collapse}
\end{figure}
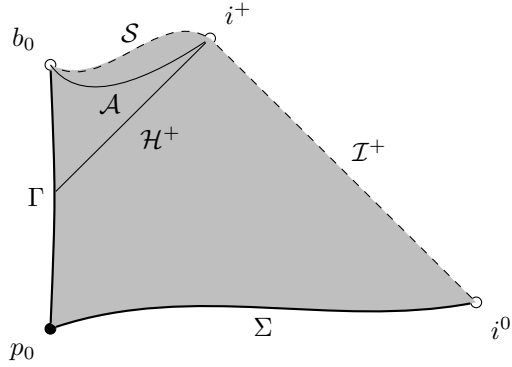

Returning to the Einstein--Maxwell--scalar field system, spacelike singularities remain relevant for spacetimes arising from the evolution of ``large'' spherically symmetric initial data, where one expects the maximal globally hyperbolic development to have a future boundary consisting of portions which are null and portions which are spacelike, see Figure~\ref{fig:twoended}. 

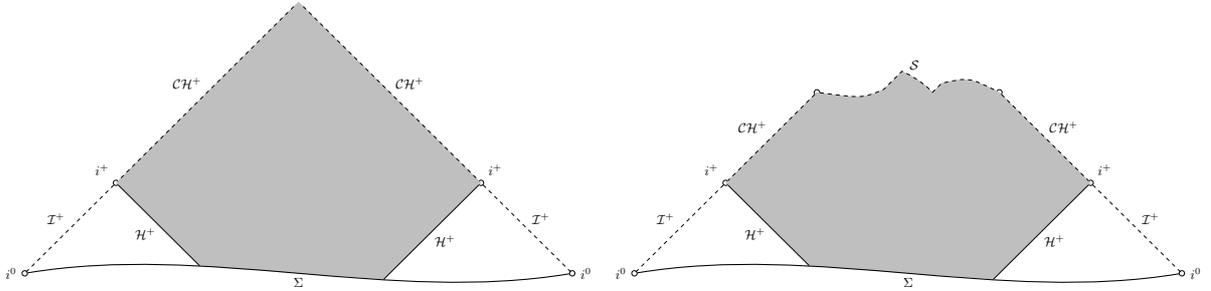
\begin{figure}[ht]
    \centering
    \begin{minipage}{.45\textwidth}
        \scalebox{0.5}{
    \begin{tikzpicture}[scale=0.8]
        %\node (s) at (0, -6) [circle, draw, inner sep=0.5mm, fill=black] {};
        \node (il) at (-6, 0) [circle, draw, inner sep=0.5mm] {};
        \node (ir) at (+6, 0) [circle, draw, inner sep=0.5mm] {};
        \node (e) at (0, 6) {};
        \node (i0l) at (-9, -3) [circle, draw, inner sep=0.5mm] {};
        \node (i0r) at (+9, -3) [circle, draw, inner sep=0.5mm] {};

        \node [above right=0.2mm of ir] {$i^+$}; 
        \node [above left=0.2mm of il] {$i^+$};
        \node [right=0.2mm of i0r] {$i^0$}; 
        \node [left=0.2mm of i0l] {$i^0$};
        \path[fill=lightgray, opacity=0.5] (0, 6) -- (-6, 0) -- (-3.25, -2.75)
            .. controls (0, -3) .. (2.8, -3.2) -- (6, 0) --  (0, 6);
        \path[fill=lightgray] (0, 6) -- (-6, 0) 
            .. controls (-4.5, -1.5) and (-1.6, -1.6) .. (-1.5, -2.85)
            .. controls (0, -3) .. (1.45, -3.1)
            .. controls (1.6, -2) and (4.5, -1.5) .. (6, 0) --  (0, 6);

        \draw (-3.25, -2.75) -- (il) node [midway, below left=-0.5mm] {$\mathcal{H}^+$};
        \draw (2.8, -3.2) -- (ir) node [midway, below right=-0.5mm] {$\mathcal{H}^+$};
        \draw [dashed] (e) -- (il) node [midway, above left] {$\mathcal{CH}^+$};
        \draw [dashed] (e) -- (ir) node [midway, above right] {$\mathcal{CH}^+$};
        \draw [dashed] (il) -- (-9, -3) node [midway, above left] {$\mathcal{I}^+$};
        \draw [dashed] (ir) -- (+9, -3) node [midway, above right] {$\mathcal{I}^+$};
        \draw [thick] (i0l) .. controls (-3, -2) and (3, -4)  .. (i0r)
            node [midway, below] {$\Sigma$};
    \end{tikzpicture}}
\end{minipage} \hspace{10pt}
\begin{minipage}{.45\textwidth}
    \scalebox{0.5}{
    \begin{tikzpicture}[scale=0.8]
        %\node (s) at (0, -6) [circle, draw, inner sep=0.5mm, fill=black] {};
        \node (il) at (-6, 0) [circle, draw, inner sep=0.5mm] {};
        \node (ir) at (+6, 0) [circle, draw, inner sep=0.5mm] {};
        \node (e) at (0, 6) {};
        \node (i0l) at (-9, -3) [circle, draw, inner sep=0.5mm] {};
        \node (i0r) at (+9, -3) [circle, draw, inner sep=0.5mm] {};
        \node (sl) at (-3, 3) [circle, draw, inner sep=0.5mm] {};
        \node (sr) at (+3, 3) [circle, draw, inner sep=0.5mm] {};

        \node [above right=0.2mm of ir] {$i^+$}; 
        \node [above left=0.2mm of il] {$i^+$};
        \node [right=0.2mm of i0r] {$i^0$}; 
        \node [left=0.2mm of i0l] {$i^0$};
        \path[fill=lightgray, opacity=0.5] (-6, 0) -- (-3, 3)
            .. controls (-1.5,2.8) .. (-0.8, 3.1) -- (-0.2, 3.7)
            .. controls (0.2, 3.5) .. (0.6, 3.2)
            -- (0.8, 3.0) -- (1.1, 3.3)
            .. controls (2, 3.5) .. (+3, +3) -- (6, 0) -- (2.8, -3.2)
            .. controls (0, -3) .. (-3.25, -2.75) -- (-6, 0);
        \path[fill=lightgray] (-6, 0) 
            .. controls (-4.5, -1.5) and (-1.6, -1.6) .. (-1.5, -2.85)
            .. controls (0, -3) .. (1.45, -3.1)
            .. controls (1.6, -2) and (4.5, -1.5) .. (6, 0)
            -- (+3, +3) .. controls (2, 3.5) .. (1.1, 3.3)
            -- (0.8, 3.0) -- (0.6, 3.2)
            .. controls (0.2, 3.5) .. (-0.2, 3.7)
            -- (-0.8, 3.1) .. controls (-1.5, 2.8) .. (-3, 3)
            -- (-6, 0);

        \draw (-3.25, -2.75) -- (il) node [midway, below left=-0.5mm] {$\mathcal{H}^+$};
        \draw (2.8, -3.2) -- (ir) node [midway, below right=-0.5mm] {$\mathcal{H}^+$};
        \draw [dashed] (sl) -- (il) node [midway, above left] {$\mathcal{CH}^+$};
        \draw [dashed] (sr) -- (ir) node [midway, above right] {$\mathcal{CH}^+$};
        \draw [dashed] (il) -- (-9, -3) node [midway, above left] {$\mathcal{I}^+$};
        \draw [dashed] (ir) -- (+9, -3) node [midway, above right] {$\mathcal{I}^+$};
        \draw [dashed] (sl) .. controls (-1.5,2.8) .. (-0.8, 3.1) -- (-0.2, 3.7)
            .. controls (0.2, 3.5) .. (0.6, 3.2)
            node [midway, above=0.8mm] {$\mathcal{S}$}
            -- (0.8, 3.0) -- (1.1, 3.3)
            .. controls (2, 3.5) .. (sr);
        \draw [thick] (i0l) .. controls (-3, -2) and (3, -4)  .. (i0r)
            node [midway, below] {$\Sigma$};
    \end{tikzpicture}}
\end{minipage}
\captionsetup{justification=centering}
\caption{Penrose diagrams representing two-ended spherically symmetric gravitational collapse spacetimes for the Einstein--Maxwell--scalar field model. The shaded regions correspond to the black hole interiors; note that in some cases (e.g.~perturbations of Reissner-Nordstr\"om, see \cite{LukOh1, LukOh2}), the interior terminates in a (weakly) singular null boundary, while for ``large data'' there may still be a portion of the boundary given by $\mathcal{S} = \{ r = 0 \}$. A priori, $\mathcal{S}$ could contain both spacelike and null pieces.}  
\label{fig:twoended}
\end{figure}

In this article we consider not just spherically symmetric but more general \emph{surface symmetric} spacetimes, where $(\mathcal{M}^{1+3}, \mathbf{g})$ can be written as a warped product of a $1+1$-dimensional Lorentzian manifold $\mathcal{Q}^{1+1}$, and a complete Riemannian surface $\Sigma^2$ of constant sectional curvature $\kappa \in \{-1, 0, +1\}$:
\[
    \mathcal{M} = \mathcal{Q} \times \Sigma, \quad \mathbf{g} = g_{\mathcal{Q}} + r^2 d \sigma_{\Sigma}.
\]
Here $g_{\mathcal{Q}}$ is a Lorentzian metric on $\mathcal{Q}$ while $r: \mathcal{Q} \to \R_{>0}$ is a positive real-valued function on $\mathcal{Q}$ called the \emph{area-radius function}\footnote{Indeed, when $\Sigma$ is compact this is a characteristic length scale associated to each copy of $\Sigma$ in the spacetime.}. 

In the (local) study of spacelike singularities, one may restrict attention to regions where the area-radius function $r$ is a \emph{timelike} coordinate. For instance, one considers only the trapped region of the Penrose diagram in Figure~\ref{fig:christodoulou_collapse}. In such trapped regions, we choose to write the spacetime metric in the following gauge:
\begin{equation} \label{eq:surface_sym_intro}
    \mathbf{g} = - e^{2 \mu} dr^2 + e^{2 \lambda} d x^2 + r^2 d \sigma_{\Sigma}.
\end{equation}
Level sets of the coordinate $x$ are orthogonal to hypersurfaces of constant $r$, while $\mu, \lambda$ are real-valued functions of $r$ and $x$. With a suitable coordinate representation for $\phi$ and $\mathbf{F}$, the PDE system corresponding to the Einstein--Maxwell--scalar field system \eqref{eq:einstein}--\eqref{eq:maxwell} in this gauge is recorded in Section~\ref{sub:emsfss}.

The gauge \eqref{eq:surface_sym_intro} has previously been considered for the surface symmetric Einstein--Vlasov--scalar field system in \cite{TegankongGlobal, TegankongRendall, TegankongNoutcheguemeRendall}, following work related to the surface symmetric Einstein--Vlasov system in \cite{ReinEinsteinVlasov}. To the author's knowledge, the gauge \eqref{eq:surface_sym} has not previously been used for other matter systems, including Maxwell; thus we choose to include a more thorough description of the gauge in Sections~\ref{sub:eqns_sss} and \ref{sub:emsfss}. We also record a local well-posedness result and a continuation criterion in Section~\ref{sub:ivp}, differing slightly from that given in \cite{TegankongNoutcheguemeRendall}.

In the sequel we shall assume that $\mathcal{Q}$ is diffeomorphic to $I \times \mathbb{S}^1$, where $I \subset (0, + \infty)$ is an interval and we periodise the $x$-variable to lie in $\mathbb{S}^1$, exactly as in \cite{TegankongRendall}. This is purely for convenience; upon localising to suitable causal subdomains one could also consider settings where $x$ lies in an interval rather than $\mathbb{S}^1$.

\subsection{Spacelike singularities in surface symmetry} \label{sub:intro_sing}

Our goal is to provide an in-depth description of the behaviour of solutions to the Einstein--Maxwell--scalar field system \eqref{eq:einstein}--\eqref{eq:maxwell} in the vicinity of a spacelike singularity. With the metric $\mathbf{g}$ written in the gauge \eqref{eq:surface_sym_intro}, the spacelike singularity corresponds to the boundary portion $\{ r = 0 \} \subset \partial \mathcal{Q}$. 

A priori, it is unclear that a solution arising from reasonable initial data even reaches $\{ r = 0 \}$. For example, the (subextremal) Reissner-Nordstr\"om black hole interior written in the gauge \eqref{eq:surface_sym_intro} is:
\[
    \mathbf{g}_{RN} = - \left( - 1 + \frac{2M}{r} - \frac{Q^2}{r^2} \right)^{-1} dr^2 + \left( - 1 + \frac{2M}{r} - \frac{Q^2}{r^2} \right) \, dx^2 + r^2 d \sigma_{\mathbb{S}^2},
\]
for constants $M$ and $Q$ obeying $0 < |Q| < M$. Hence the solution, at least in this gauge, is only well-defined for $r_- < r < r_+$, where $r_{\pm} = M \pm \sqrt{M^2 - Q^2}$. It is well-known that the boundary piece $\{ r = r_- \}$ is a null (and completely regular) Cauchy horizon as opposed to a spacelike singularity.

Nonetheless, we shall restrict attention to scenarios where $\{ r = 0 \}$ is reached. Indeed, our first main result, Theorem~\ref{thm:global_rough}, shows global existence towards $\{ r = 0 \}$ for suitably chosen initial data. See also previous results of \cite{Tchapnda, WeaverT2, RendallPlaneSym, TegankongRendall} concerning the Einstein--Vlasov and Einstein--scalar field system. We also make an analogy to our previous result \cite{Me_Kasner}, which similarly considers the Einstein--Maxwell--scalar field system, albeit in a different gauge known as the \emph{double-null gauge}.

Given existence up to $\{ r = 0 \}$, one should next ask how solutions behave near $r = 0$. For instance, one could consider the asymptotic behaviour of the metric and the matter quantities. This is also partially answered in \cite{ReinEinsteinVlasov, TegankongRendall, Me_Kasner} for the matter models they consider, where \emph{under certain conditions} it is shown that the metric behaves in a Kasner--like fashion (see Section~\ref{sub:intro_bkl} for a discussion of Kasner--like singularities), and that the scalar field $\phi$ has the following asymptotics as $r \to 0$:
\begin{equation} \label{eq:phi_intro}
    \phi(r, x) = \Psi (x) \log r + \Xi (x) + o(1),
\end{equation}
where the coefficients $\Psi(x)$ and $\Xi(x)$ are allowed to depend on the spatial coordinate $x \in \mathbb{S}^1$. Furthermore, the \emph{generalized Kasner exponents} of the near-singularity spacetime are:
\begin{equation} \label{eq:kasnerlike_intro}
    p_1(x) = \frac{\Psi(x)^2 - 1}{\Psi(x)^2 + 3}, \quad p_2(x) = p_3(x) = \frac{2}{\Psi(x)^2 + 3}.
\end{equation}

We return to the ``\emph{under certain conditions}'' alluded to above, at least in the context of the Einstein--Maxwell--scalar field system. We consider two different cases: for the pure Einstein--scalar field system with $\mathbf{F} = 0$ and $\kappa \in \{ 0, +1 \}$, there are no additional conditions; global existence towards $\{ r = 0 \}$ and the asymptotics \eqref{eq:phi_intro}--\eqref{eq:kasnerlike_intro} will always apply. Furthermore, the function $\Psi(x)$ is permitted to take any value. 

When either $\mathbf{F} \neq 0$ or $\kappa = -1$, two additional conditions are required:
\begin{itemize}
    \item
        \hypertarget{assump1}{Assumption 1}: Global existence towards $\{ r = 0 \}$ with $e^{2\mu} = - g (\nabla r, \nabla r) \to 0$ uniformly as $r \downarrow 0$.
    \item
        \hypertarget{assump2}{Assumption 2}: If $\mathbf{F} \neq 0$, then assume the \emph{subcriticality condition}, namely that for some $\alpha > 0$ and all $r$ sufficiently small, one has $|r \partial_r \mu| (r, x) \geq \frac{1}{2}( 2 + \alpha)$. %If $\mathbf{F} \neq 0$, then also assume the \emph{subcriticality} condition that for some $\alpha > 0$ and $r$ sufficiently small, $|r \partial_r \phi| (r, x) > 1 + \alpha$. This is morally equivalent to $\Psi(x) \geq 1 + \alpha$ in \eqref{eq:phi_
\end{itemize}
As we see later, Assumption 2 implies that near $\{ r = 0 \}$ one has $\Psi(x)^2 \geq 1 + \alpha$ in \eqref{eq:phi_intro}, see e.g.~\eqref{eq:mu_asymp} below.
The reason that one does not require these additional assumptions when $\mathbf{F} = 0$ and $\kappa \in \{ 0, +1 \}$ is that global existence and an appropriately modified subcriticality condition follow immediately from the equations in Section~\ref{sub:emsfss} when restricted to such spacetimes.

In any case, if one satisfies these assumptions then asymptotics of the form \eqref{eq:phi_intro}--\eqref{eq:kasnerlike_intro} were found in \cite[Theorem 3.1 and Theorem 3.2]{Me_Kasner}, though note \cite{Me_Kasner} deals only with the spherically symmetric case $\kappa = +1$ and uses the double null gauge in place of \eqref{eq:surface_sym_intro}. One drawback of the double null gauge is that null coordinates tend to degenerate towards $\{ r = 0 \}$, and one only has $C^{1, \alpha}$ regularity for say $\Psi(x)$. In our current gauge \eqref{eq:surface_sym}, the asymptotics \eqref{eq:phi_intro}--\eqref{eq:kasnerlike_intro} can be given to any desired degree of regularity.

\setcounter{theorem}{-1}
\begin{theorem} \label{thm:asymp_smooth}
    Let $(\mathcal{M}, \mathbf{g}, \phi, \mathbf{F})$ be a surface symmetric solution to the Eintein--Maxwell--scalar field system \eqref{eq:einstein}--\eqref{eq:maxwell}, given in the gauge \eqref{eq:surface_sym_intro}, arising from regular initial data $(\phi, \mu, \lambda, \partial_r \phi, \partial_r \mu, \partial_r \lambda)|_{r = r_0} \in (C^{k+1})^3 \times (C^k)^3$ for some $k \geq 1$. Suppose that either $\mathbf{F} = 0$ and $\kappa \in \{0, + 1\}$, or that Assumptions 1 and 2 above hold. Then there exist $C^{k-1}$ functions $\Psi(x)$, $\Xi(x)$, $M(x)$ and $\Lambda(x)$ on $\mathbb{S}^1$ such that:
    \begin{gather}
        \phi(r, x) = \Psi(x) \log r + \Xi (x) + \mathrm{Err}_{\phi}(r, x), \label{eq:phi_asymp} \\[0.4em]
        \mu(r, x) = \tfrac{1}{2} (\Psi(x)^2 + 1) \log r + M (x) + \mathrm{Err}_{\mu}(r, x), \label{eq:mu_asymp} \\[0.4em]
        \lambda(r, x) = \tfrac{1}{2} (\Psi(x)^2 - 1) \log r + \Lambda(x) + \mathrm{Err}_{\lambda}(r, x), \label{eq:lambda_asymp}
    \end{gather}
    where the error terms tend to $0$ in the $C^{k-1}$ norm uniformly as $r \to 0$, i.e.~they obey $\| \mathrm{Err}_{*}(r, \cdot) \|_{C^{k-1}_x} \to 0$ where $*$ is any of $\phi, \mu, \lambda$. Furthermore, $M(x)$ obeys the asymptotic momentum constraint equation
    \begin{equation}
        d M (x) = \Psi(x) \cdot d \,\Xi (x).
    \end{equation}
\end{theorem}

By performing a change of coordinates to change the gauge \eqref{eq:surface_sym_intro} into a Kasner-like form (see Section~\ref{sub:intro_bkl}), this confirms the asymptotics \eqref{eq:phi_intro}--\eqref{eq:kasnerlike_intro}. To the extent of the author's knowledge, Theorem~\ref{thm:asymp_smooth} is new, but is not the focus of the present article; we therefore defer its proof to Appendix~\ref{app:asymp}.
Instead the starting point of our article is to address the following related questions:
\begin{quote}
    Can one characterise a wide class of initial Cauchy data at $r = r_0$, including data which are not at first glance subcritical in the sense of Assumption 2 (e.g.~data at $r = r_0$ such that $r \partial_r \mu|_{r = r_0} < \frac{1}{2}$ \underline{and} $\mathbf{F} \neq 0$), so that the corresponding solutions of the Einstein--Maxwell--scalar field system feature global existence towards $\{ r = 0 \}$, as well as the asymptotics \eqref{eq:phi_intro}--\eqref{eq:kasnerlike_intro}? 

    Further, for all data in this class can one fully understand the \emph{intermediate dynamics} of the spacetime between the regular Cauchy data at $r = r_0$ and the singularity at $r = 0$?
\end{quote}
    
Note that at $r = 0$ itself one expects that $\frac{1}{2}(\Psi^2 + 1) = \lim_{r \to 0} r \partial_r \mu(r, x) \geq \frac{1}{2}$, hence in answering these questions one expects that in some cases there must be a nonlinear transition between $r = r_0$ and $r = 0$ which nevertheless permits global existence. Later, we identify this nonlinear mechanism as a ``bounce'' such as that identified by Belinski, Khalatnikov and Lifshitz in \cite{bkl71, bk73}, see Section~\ref{sub:intro_bkl}.

Our main results provide both a description of such a class of initial data (Theorem~\ref{thm:global_rough}) as well as a description of the intermediate (nonlinear) dynamics (Theorem~\ref{thm:bounce_rough}), thereby answering the above questions in the affirmative. This class will moreover be open (in the $C^{\infty}$ topology on initial data sets), thus our main results will represent a qualitative stability result in the sense that in the regime considered, spacelike singularity formation, and related phenomena such as curvature blow-up, are stable to small perturbations. %This is the content of our global existence result Theorem~\ref{thm:global_rough}.

In the process of answering these questions we uncover an interesting corollary of Theorem~\ref{thm:bounce_rough} that we interpret as an \emph{instability result} -- see Corollary~\ref{cor:bounce} for a precise statement. To describe this, first consider a surface symmetric spacetime $(\mathcal{M}, \mathbf{g}_0, \phi_0)$ solving the Einstein--scalar field equations with $\kappa \in \{ 0, +1 \}$. Since this spacetime has $\mathbf{F}_0 = 0$, the asymptotic quantity $\Psi_0(x)$ associated to $\phi_0$ by \eqref{eq:phi_intro} is permitted to take any value and \underline{does not necessarily} obey $\Psi_0(x)^2 \geq 1$ for all $x \in \mathbb{S}^1$. 

Now consider a perturbation of $(\mathcal{M}, \mathbf{g}_0, \phi_0, \mathbf{F}_0 = 0)$ to $(\mathcal{M}, \mathbf{g}_1, \phi_1, \mathbf{F}_1 \neq 0)$. The instability result then says that though one has global existence towards $\{ r = 0 \}$ in $(\mathcal{M}, \mathbf{g}_1, \phi_1, \mathbf{F}_1)$, the new $\Psi_1(x)$ associated to $\phi_1(x)$ via \eqref{eq:phi_intro} will necessarily satisfy $\Psi_1(x)^2 \geq 1$. In fact we will show that $\Psi_1(x) \approx \max \{ \Psi_0(x), \Psi_0(x)^{-1} \}$, quantifying the ``bounce'' instability precisely.

Corollary~\ref{cor:bounce} requires neither the unperturbed spacetime $(\mathcal{M}, \mathbf{g}_0, \phi_0)$, nor the perturbation, to be spatially homogeneous, and together with the results of our companion article \cite{MeGowdyPaper} can be considered the first rigorous evidence of BKL-type bounces \emph{outside of homogeneity}. See Section~\ref{sub:intro_gowdy} for a more thorough comparison to \cite{MeGowdyPaper}. (BKL Bounces are much better understood in the spatially homogeneous setting, where the dynamics reduce to a system of finite dimensional ODEs \cite{Weaver_bianchi, RingstromBianchi, LiebscherRendallTchapnda, BeguinDutilleul}.)

\subsection{Our main theorems in rough form} \label{sub:intro_thm}

Our main theorems concern surface symmetric solutions of the Einstein--Maxwell--scalar field system \eqref{eq:einstein}--\eqref{eq:maxwell} as described in Section~\ref{sub:intro_einstein}. With metric $\mathbf{g}$ given in the gauge \eqref{eq:surface_sym_intro}, the Einstein--Maxwell--scalar field system reduces to a system of wave and transport equations for three scalar functions $\phi(r, x)$, $\mu(r, x)$ and $\lambda(r, x)$. This system is given in Section~\ref{sub:emsfss}.

\sloppy Local existence (see Proposition~\ref{prop:lwp}) for this system states that for suitably regular initial data $\phi_D, \mu_D, \lambda_D, \dot{\phi}_D, \dot{\mu}_D, \dot{\lambda}_D: \mathbb{S}^1 \to \R$ and some $r_0 > 0$, there exists a maximal interval $I \subset (0, + \infty)$ containing $r_0$ for which $(\phi, \mu, \lambda)$ solves the Einstein--Maxwell--scalar field system in $I \times \mathbb{S}^1$ and such that
\[
    (\phi, \mu, \lambda, r \partial_r \phi, r \partial_r \mu, r \partial_r \lambda) |_{r = r_0} = (\phi_D, \mu_D, \lambda_D, \dot{\phi}_D, \dot{\mu}_D, \dot{\lambda}_D).
\]

We have two main theorems. The first, Theorem~\ref{thm:global_rough}, which we refer to as a \emph{global existence theorem}, characterizes a class of initial data $(\phi_D, \mu_D, \lambda_D, \dot{\phi}_D, \dot{\mu}_D, \dot{\lambda}_D)$ and $r_0 > 0$ for the Einstein--Maxwell--scalar field system such that the maximal interval $I$ extends all the way to zero i.e.~$I = (0, T)$ for $T \in (0, + \infty]$. The second, Theorem~\ref{thm:bounce_rough}, which we refer to as the \emph{bounce theorem}, characterises the (nonlinear) dynamics in spacetimes arising from such initial data, including a relationship between the initial data at $r = r_0$ and the eventual asymptotics of $(\phi, \mu, \lambda)$ towards $r = 0$.

To define this class of initial data, we introduce two real-valued parameters $\upeta > 2$ and $\upzeta > 0$ are real-valued parameters. Our class of initial data is chosen to satisfy three conditions:
\begin{itemize}
    \item
        (\emph{Weak subcriticality}) The functions $\dot{\phi}_D$ and $\mu_D$ satisfy the following, for all $x \in \mathbb{S}^1$:
        \begin{equation} \label{eq:weak_subcriticality}
            \upeta^{-1} \leq \dot{\phi}_D(x) \leq \upeta, \qquad \frac{Q^2}{r_0^2} e^{2 \mu_D}(x) \leq 1.
        \end{equation}

    \item
        (\emph{Energy boundedness}) For some $N \in \N$ depending on $\upeta$ as well as another small constant $\upgamma > 0$, we have control of $L^2$ control of up to $N$ derivatives of the initial data:
        \begin{equation} \label{eq:energy_boundedness}
            \frac{1}{2} \sum_{f \in \{\phi, \mu, \lambda\}} \sum_{K = 0}^N \int_{\mathbb{S}^1} \left( (\partial_x^K \dot{f}_D)^2 + r_0^2 e^{2(\mu_D - \lambda_D)} (\partial_x^{K+1} f_D)^2 + r_0^{2\upgamma} (\partial_x^K f_D)^2 \right) \, dx \leq \upzeta.
        \end{equation}
        
    \item
        (\emph{Closeness to singularity}) The ``initial time'' $r_0$ is chosen to satisfy $0 < r_0 < r_*$ for some $r_*$ depending on $\upeta$ and $\upzeta$. Furthermore we have
        \begin{equation} \label{eq:closeness_singularity}
            e^{2 \mu_D} \leq \upzeta r_0, \qquad e^{2(\mu_D - \lambda_D)} \leq \upzeta r_0.
        \end{equation}
\end{itemize}
Note that taking the union of all such initial data as $\upeta > 2$ and $\upzeta > 0$ vary, one characterises the class of initial data to which our results apply as a subset of the ``moduli space of initial data'' for the surface symmetric Einstein--Maxwell--scalar field equations which is open in the $C^{\infty}$ topology.

We defer further discussion of these conditions (both regarding why they are necessary and the relevance of initial data satisfying them) to after the statement of the theorems. The first theorem, corresponding to global existence towards $r = 0$, is as follows:

\begin{theorem}[Global existence, rough version] \label{thm:global_rough}
    For some $\upeta > 2$ and $\upzeta > 0$, let $(\phi_D, \mu_D, \lambda_D, \dot{\phi}_D, \dot{\mu}_D, \dot{\lambda}_D)$ be initial data at $r_0 > 0$ satisfying \eqref{eq:weak_subcriticality}--\eqref{eq:closeness_singularity} and $r_0 < r_* = r_*(\upeta, \upzeta)$. Then the corresponding solution $(\phi, \lambda, \mu)$ to the surface symmetric Einstein--Maxwell--scalar field system of Section~\ref{sub:emsfss} is such that the maximal interval of existence $I$ extends all the way to $r = 0$ i.e.~$(0, r_0) \subset I$.
%
%    Furthermore, there exist $L^{\infty}$ bounds for low-order quantities including $r \partial_r \phi$, while energies $\mathcal{E}^{(K)}(r)$ corresponding to $L^2$ norms of up to $K + 1$ derivatives of $(\phi, \mu, \lambda)$ at a slice of constant $r$ blow up at worst at an inverse polynomial rate $r^{-q}$
\end{theorem}

Our second main theorem concerns the dynamical behaviour of particular quantities along causal curves. To fix notation, let $\gamma: I \to \mathcal{Q} = I \times \mathbb{S}^1$ be any \emph{causal curve} parameterised by its $r$-coordinate. Using the gauge for $\mathbf{g}$ in \eqref{eq:surface_sym_intro}, this means we let $\gamma(r) = (r, x(r))$ with $\left| \frac{dx}{dr} \right| \leq e^{\mu - \lambda}$. Then define:
\[
    \mathscr{P}_{\gamma}(r) \coloneqq r \partial_r \phi (\gamma(r)), \qquad \mathscr{Q}_{\gamma}(r) \coloneqq \frac{Q^2}{r^2} e^{2\mu} (\gamma(r)).
\]
It is important that our theorem holds \underline{uniformly} in the choice of causal curve $\gamma$, with the ODEs describing bounces being independent for different choices of $\gamma$ with distinct endpoints, up to error terms which reflect \emph{AVTD behaviour}, see Section~\ref{sub:intro_bkl}. This reflects that the result is \emph{spatially inhomogeneous}.

\begin{theorem}[Bounces, rough version]  \label{thm:bounce_rough}
    Let $(\phi, \mu, \lambda)$ be a solution to the surface symmetric Einstein--Maxwell--scalar field system of Section~\ref{sub:emsfss}, arising from initial data as in Theorem~\ref{thm:global_rough}. Then for $\gamma(r)$ any causal curve as above, the quantities $\mathscr{P}_{\gamma}(r)$ and $\mathscr{Q}_{\gamma}(r)$ obey the ODEs:
    \[
        r \frac{d}{dr} \mathscr{P}_{\gamma} = - \mathscr{P}_{\gamma} \mathscr{Q}_{\gamma} + \mathscr{E}_{\mathscr{P}}, \qquad r \frac{d}{dr} \mathscr{Q}_{\gamma} = \mathscr{Q}_{\gamma} ( \mathscr{P}_{\gamma}^2 - 1 - \mathscr{Q}_{\gamma} - \mathscr{E}_{\mathscr{Q}} ),
    \]
    where the error terms vanish quickly as $r \to 0$. Furthermore:
    \begin{enumerate}[(i)]
        \item
            If $Q = 0$, then $\mathscr{P}_{\gamma}(r)$ converges to some $\mathscr{P}_{\gamma, \infty}$ as $r \to 0$, satisfying $\mathscr{P}_{\gamma, \infty} \approx \mathscr{P}_{\gamma}(r_0)$.
            
        \item
            If $Q \neq 0$, then $\mathscr{Q}_{\gamma}(r)$ converges to $0$ as $r \to 0$, while $\mathscr{P}_{\gamma}(r)$ converges to some $\mathscr{P}_{\gamma, \infty}$ as $r \to 0$, necessarily satisfying $\mathscr{P}_{\gamma, \infty} \geq 1$ and in the particular case where $\mathscr{Q}_{\gamma}(r_0)$ is small,
            \[
                \mathscr{P}_{\gamma, \infty} \approx \max \{ \mathscr{P}_{\gamma}(r_0), \mathscr{P}_{\gamma}(r_0)^{-1} \} + O( \mathscr{Q}_{\gamma}(r_0)).
            \]
    \end{enumerate}
\end{theorem}

The final statement in Theorem~\ref{thm:bounce_rough} can be interpreted as an instability for $\mathscr{P}_{\gamma}$ along causal curves $\gamma$, in the following sense: imagine one had initial data as in Theorem~\ref{thm:global_rough}, chosen so that for some  $x \in \mathbb{S}^1$, $\dot{\phi}_D(x) = r \partial_r \phi(r_0, x)$ obeys $\upeta^{-1} < \dot{\phi}_D(x) < 1$ and $\frac{Q^2}{r_0^2} e^{2\mu_D}(x)$ small but nonzero. Then though existence towards $r = 0$ follows from Theorem~\ref{thm:global_rough}, the behaviour of $\mathscr{P}_{\gamma}(r) = r \partial_r \phi(\gamma(r))$ on any causal curve through $(r_0, x)$ is such that as $r \to 0$, $\mathscr{P}_{\gamma}(r)$ eventually transitions to something close to $\mathscr{P}_{\gamma}(r_0)^{-1}$. %Further, for $\upeta$ chosen large, this transition, or ``bounce'', can be made very large!

It is of particular interest to apply this instability in the context of \emph{perturbing} spacetimes where $Q = 0$ and thus by (i) one $\mathscr{P}_{\gamma, \infty} \approx r \partial_r \phi(r_0, x) < 1$ is allowed. By considering the specific case where $\gamma$ is a constant $x$-curve, then in the language of Theorem~\ref{thm:asymp_smooth}, this means $\Psi(x) < 1$ is allowed. 
Now consider perturbations such that charge is turned on i.e.~$Q$ is set to be nonzero (but small). The the resulting spacetime still exists up to $r = 0$, but the \emph{asymptotics change}; the new value of $\mathcal{P}_{\gamma, \infty}$ will now be such that $\mathcal{P}_{\gamma, \infty} \approx r \partial_r \phi(r_0, x)^{-1}$. Using familiar langauge that means $\Psi(x) \geq 1$ for the perturbed spacetime, and this demonstrates that the quantitative asymptotics of the original uncharged spacetime are themselves unstable. See Figure~\ref{fig:bounce_instability} and Corollary~\ref{cor:bounce}.

\begin{figure}[ht] 
    \centering

    \begin{tikzpicture}[scale = 0.8, every text node part/.style={align=center}]
        \coordinate (LLD) at (-9, 0);
        \coordinate (LLM) at (-9, 3);
        \coordinate (LLU) at (-9, 5);
        \coordinate (LRD) at (-2, 0);
        \coordinate (LRM) at (-2, 3);
        \coordinate (LRU) at (-2, 5);
    
        \begin{scope}[decoration=
            {markings, mark=at position 0.6 with {\arrow{>}}}
            ]
            \path[fill=lightgray, opacity=0.5] (LLU) -- (LLD) -- (LRD) -- (LRU) -- (LLU);
    
            \draw[decorate, decoration={brace, raise=5pt, amplitude=5pt}] (-9, 5.1) -- node[midway, above=10pt] {$x \in \mathbb{S}^1$} (-2, 5.1);
            \draw[dashed, postaction={decorate}] (LLU) -- node[midway, left] {$r$}  (LLD);  
            \draw[dashed, postaction={decorate}] (LRU) -- (LRD);  
            \draw[thick, dashed] (LLD) -- node[midway, below] {$\Psi(x) \approx {\color{blue}\dot{\phi}_D(x)}$, \\ $1/4 < \Psi(x) < 1/2$} (LRD);
            \draw[thick] (LLM) -- node[midway, above] {${\color{blue}\dot{\phi}_D(x)} \coloneqq r \partial_r \phi(r_0, x)$, \\ $1/4 < \dot{\phi}_D(x) < 1/2$} (LRM);

            \node (LC) at (-5.5, 1.5) {$Q = 0$};
        \end{scope}
    
        \coordinate (RRD) at (+9, 0);
        \coordinate (RRM) at (+9, 3);
        \coordinate (RRU) at (+9, 5);
        \coordinate (RLD) at (+2, 0);
        \coordinate (RLM) at (+2, 3);
        \coordinate (RLU) at (+2, 5);
    
        \begin{scope}[decoration=
            {markings, mark=at position 0.6 with {\arrow{>}}}
            ]

            \path[fill=lightgray, opacity=0.5] (RRU) -- (RRD) -- (RLD) -- (RLU) -- (RRU);
    
            \draw[decorate, decoration={brace, raise=5pt, amplitude=5pt}] (2, 5.1) -- node[midway, above=10pt] {$x \in \mathbb{S}^1$} (9, 5.1);
            \draw[dashed, postaction={decorate}] (RRU) -- node[midway, right] {$r$} (RRD);  
            \draw[dashed, postaction={decorate}] (RLU) -- (RLD);  
            \draw[thick, dashed] (RRD) -- node[midway, below] {$\Psi(x) \approx {\color{blue}\dot{\phi}_D(x)^{-1}}$, \\ $2 < \Psi(x) < 4$} (RLD);
            \draw[thick] (RRM) -- node[midway, above] {${\color{blue}\dot{\phi}_D(x)} \coloneqq r \partial_r \phi(r_0, x)$, \\ $1/4 < \dot{\phi}_D(x) < 1/2$} (RLM);

            \node (RC) at (5.5, 1.5) {$Q = \varepsilon$};
        \end{scope}

        \node at (0, 3.25) {\small data at $r = r_0$};
        \node at (0, -0.25) {\small asymptotics at $r = 0$};

        \draw[very thick, ->] (-1, 1.5) -- node[midway, above] {\small perturbation} (1, 1.5);
    \end{tikzpicture}

    \captionsetup{justification = centering}
    \caption{Example of a bounce instability arising from a charged perturbation of an uncharged spacetime. In the unperturbed spacetime (left) $Q = 0$ and the asymptotic quantity $\Psi(x)$ is close to the (inhomogeneous) data {\color{blue}$\dot{\phi}_D(x)$}, while in the charged spacetime (right) with $Q = \varepsilon$ the  bounce of Theorem~\ref{thm:bounce_rough} applied to constant $x$-curves $\gamma$ means that $\Psi(x)$ is instead close to {\color{blue}$\dot{\phi}_D(x)^{-1}$}. A similar diagram appears in \cite[Figure 1]{MeGowdyPaper}}
    \label{fig:bounce_instability}
\end{figure}
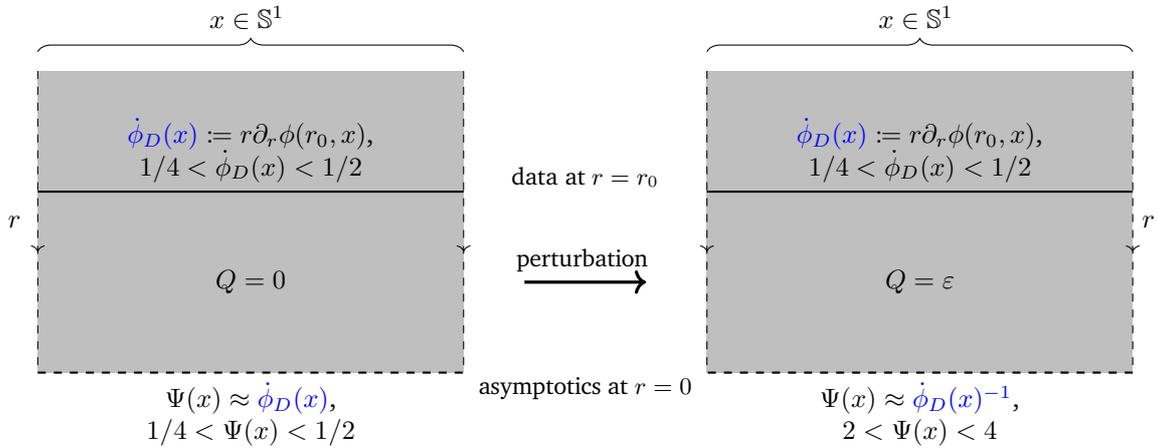

We return now to our conditions on initial data. The first condition, weak subcriticality, is engineered so that we can include spacetimes where the nonlinear instability above can occur, but also such that this nonlinear behaviour is not too wild. For instance, $r \partial_r \phi(r_0, x)$ is bounded below by $\upeta^{-1}$ else the transition from $r \partial_r \phi(r_0, x)$ to $r \partial_r \phi(r_0. x)^{-1}$ would be too drastic. In the proof, this means that $\mathscr{P}_{\gamma}$ and $\mathscr{Q}_{\gamma}$ as defined in Theorem~\ref{thm:bounce_rough} remain bounded.

The second condition of energy boundedness is natural as our proof will rely on energy estimates at top order. %The number of derivatives required, $N$, will depend on $\upeta$, due to the fact that our energy energy estimates degenerate towards $r = 0$, and that one only closes $L^{\infty}$ estimates at low-order by using an interpolation method similar to that of \cite{FournodavlosRodnianskiSpeck, groeniger2023formation}. As $\upeta$ increases, the degeneration is worse and therefore more derivatives are required to get the correct interpolated estimate. 
We also comment upon the form of the energy in \eqref{eq:energy_boundedness}. It is important that this energy is compatible with the asymptotics of Theorem~\ref{thm:asymp_smooth}, namely that for $(\phi, \lambda, \mu)$ as in that theorem
\[
    \frac{1}{2} \int_{\mathbb{S}^1} \left( (r \partial_r \partial_x^K f)^2 + r^2 e^{2(\mu - \lambda)} (\partial_x^{K+1} f)^2 + r^{2\upgamma} (\partial_x^K f)^2 \right) \, dx \to \frac{1}{2} \| \partial_x^K \Psi \|_{L^2(\mathbb{S}^1)}^2, \quad \text{ as } r \to 0
\]
Similar convergence holds for $\lambda$ and $\mu$. Thus there exists $\upzeta$ such that for $r_0 > 0$ sufficiently small the spacetimes of Theorem~\ref{thm:asymp_smooth} are compatible with \eqref{eq:energy_boundedness}. One the other hand, since $\upzeta$ can be chosen large, our results are not confined to near-homogeneous spacetimes. When $\upzeta$ is large, however, our final condition, closeness to singularity, is required to prevent spatial inhomogeneity from playing a major role in the dynamics. Note that $r_*$ and $\upzeta$ play a similar role to that of $\zeta_1$ and $\zeta_0$ in \cite[Theorem 12]{groeniger2023formation}.

\subsection{Relationship to BKL and the Kasner map} \label{sub:intro_bkl}

We relate our results to the physics and mathematical literature regarding spacelike singularities; this section is abridged from a similar in discussion in our companion paper \cite{MeGowdyPaper}, though modified to include discussion of the scalar field $\phi$ and electromagnetic field $\mathbf{F}$. %Celebrated ``singularity theorems'' of Penrose \cite{penrose} and Hawking \cite{hawking67} yield that ``singularities'' are a robust prediction of Einstein's equations \eqref{eq:einstein}. However, these results do not prove singularity formation in the sense of blow-up, but instead (causal) geodesic incompleteness. In particular, they say little about dynamical features of spacetimes near their incomplete boundaries.

In particular, we outline the approach of Belinski, Khalatnikov and Lifshitz (often abbreviated to BKL) in their heuristic investigation of near-singularity solutions to Einstein's equation \cite{kl63, bkl71, bk73, bk77, bkl82}. They propose the following ansatz for solutions to the Einstein--scalar field system near singularity: for a spacetime $\mathcal{M} = (0, T) \times \mathbf{\Sigma}^3$ with singularity located at $\{0\} \times \mathbf{\Sigma}$, they suggest the following leading order expansion for $\mathbf{g}$ and $\phi$:
\begin{equation} \label{eq:bkl}
    \mathbf{g} = - d \tau^2 + \sum_{I=1}^3 \tau^{2 p_I(x)} \mathbf{\omega}^I(x) \otimes \mathbf{\omega}^I(x) + \cdots, \qquad \phi = p_{\phi}(x) \log \tau + \cdots.
\end{equation}
Here $p_I(x), p_{\phi}(x)$ are functions on $\mathbf{\Sigma}$ known as \textit{generalized Kasner exponents}, and $\{\mathbf{\omega}^I(x)\}$ is a frame of $1$-forms on $\mathbf{\Sigma}$. By inserting these into \eqref{eq:einstein}--\eqref{eq:wave} and considering a formal power series in $t$, BKL determine that the generalised Kasner exponents must satisfy the \textit{generalized Kasner relations}:
\begin{equation} \label{eq:kasner_relation}
    \sum_{I = 1}^3 p_I(x) = 1, \quad \sum_{I = 1}^3 p_I^2(x) + 2 p_{\phi}^2(x) = 1.
\end{equation}

A further consistency check using the Einstein equations suggests that for the ansatz \eqref{eq:bkl} to remain valid up to $t = 0$, either one has $p_I > 0$ for all $I = 1,2,3$, or that whenever $p_I < 0$ is negative the associated one-form $\omega^I$ must be integrable in the sense of Frobenius, i.e.~$\omega^I \wedge d \omega^I = 0$. In vacuum (where $\phi = 0$ and thus $p_{\phi} = 0$), the Kasner relations \eqref{eq:kasner_relation} prohibits the $p_I$ all being positive, but once $\phi \neq 0$ then positivity of the $p_I$ is possible.

Next BKL \cite{bkl71} give heuristics explaining what happens if the above consistency check is violated. They first assume \emph{Asymptotically Velocity Term Dominated (AVTD)} behaviour, meaning that `spatial derivatives' occuring in the Einstein equations \eqref{eq:einstein} are subdominant in comparison to $\partial_{\tau}$-derivatives. Quantitatively, this means that the dynamics in the causal futures of distinct points $(0, p) \in \{0\} \times \mathbf{\Sigma}$ on the singularity are decoupled. 

Furthermore, their computation yields that in such a causal neighborhood, the ansatz \eqref{eq:bkl} is valid for $\tau \gg \tau_B$ for some critical time $\tau_B$, but for $\tau \ll \tau_B$ the metric and scalar field will transition to something that resembles \eqref{eq:bkl} but with $p_I(x)$, $\omega^I(x)$ and $p_{\phi}(x)$ replaced by some new $\acute{p}_I(x)$, $\acute{\mathbf{\omega}}^I(x)$ and $\acute{p}_{\phi}(x)$. While $\tau \sim \tau_B$, the spacetime undergoes a nonlinear transition often denoted in the literature as a \emph{BKL} or \emph{Kasner bounce}. These nonlinear transitions, or bounces, will continue to occur either indefinitely, or until the spacetime settles to something satisfying the consistency check.

Miraculously, BKL actually propose a formula related the generalized exponents post-transition and pre-transition: if $p_1 < 0$ and $\omega^1$ is the one-form responsible for the failure above, the
\begin{equation} \label{eq:kasner_relation_bounce}
    \acute{p}_1 = - \frac{p_1}{1 + 2 p_1}, \quad \acute{p}_2 = \frac{p_2 + 2 p_1}{1 + 2 p_1}, \quad \acute{p}_3 = \frac{p_3 + 2 p_1}{1 + 2 p_1}, \quad \acute{p}_{\phi} = \frac{p_{\phi}}{1 + 2 p_1}.
\end{equation}
In \cite{bk73}, it is commented that if $p_{\phi}$ is originally nonzero and the above map is precise, then there will only be finitely many such transitions before the spacetime settles to a regime where all of the $p_I$ are positive.
%To derive this formula, BKL assume that near singularity, the dynamics of the Einstein-

We relate this to our results in surface symmetry. %If a surface symmetric spacetime could be described by the ansatz \eqref{eq:bkl}, then by definition two of the generalised Kasner exponents $p_I$ must be equal everywhere, say $p_2(x) = p_3(x)$. Alongside the relations \eqref{eq:kasner_relation}, we thus have $3$ equations for $4$ unknowns, or alternatively the remaining $p_I$ should be understood once we understand the coefficient $p_{\phi}(x)$.
As our choice of time variable in \eqref{eq:surface_sym_intro} is not $\tau$ in \eqref{eq:bkl} but instead the area-radius $r$, we need to make a change of variables,
%as our choice of time variable is not the $\tau$ in \eqref{eq:bkl}, but instead the area-radius function $r$. 
such that one may interpret the above heuristics in terms of the function $\Psi(x)$ of Theorem~\ref{thm:asymp_smooth}. Due to \eqref{eq:mu_asymp}, we relate $\tau$ to $r$ by $d \tau = e^{\mu} dr \sim r^{\frac{\Psi^2 + 1}{2}} dr$. Roughly setting $\omega^1(x) = e^{\lambda}dx$, and $\omega^2(x), \omega^3(x)$ to correspond to the surface of symmetry $\Sigma$, we can make a formal correspondence between the gauge \eqref{eq:surface_sym_intro} and the ansatz \eqref{eq:bkl}, and the generalised Kasner exponents we get are exactly those in \eqref{eq:kasnerlike_intro}.

According to \eqref{eq:kasnerlike_intro}, the generalized Kasner exponents $p_I(x)$ are positive if and only if $\Psi(x)^2 > 1$, thus we expect $|\Psi(x)| = 1$ to be a threshold between stable and unstable behaviour. Note, however, that since $\omega^1(x) = e^{\lambda} dx$ is integrable in the sense of Frobenius, the BKL inconsistency described above is not triggered in the Einstein--scalar field model. In order to see $|\Psi(x)| = 1$ as a genuine threshold, we introduce more matter to source the right hand side of the Einstein equations, namely the Maxwell field $\mathbf{F}$.

It is a curious coincidence that upon introducing non-trivial $\mathbf{F}$, the BKL ansatz \eqref{eq:bkl} is once again inconsistent with the Einstein--Maxwell equations when one of the $p_I$ is negative. Furthermore, in \cite{bk77} Belinski and Khalatnikov suggest that the BKL bounce phenomenon described above can also occur due to the presence of $\mathbf{F}$, and moreover the transition map \eqref{eq:kasner_relation_bounce} for the exponents $\acute{p_I}, \acute{p}_{\phi}$ is identical in this case. We refer the reader to \cite[Section 1.5]{Me_Kasner} and \cite[Chapter 4]{BelinskiHenneaux} for further comments regarding the role of electromagnetism.

In any case, one expects that upon adding $\mathbf{F}$, our surface symmetric spacetimes will remain stable and self-consistent when $\Psi(x)^2 > 1$, but will be subject to an instability or inconsistency when $\Psi(x)^2 < 1$. Furthermore, via the correspondence \eqref{eq:kasnerlike_intro} and the transition map \eqref{eq:kasner_relation_bounce}, one can check that a bounce corresponds to a transition of the form $\Psi(x) \mapsto \acute{\Psi}(x) = \Psi(x)^{-1}$, after which $\acute{\Psi}^2 > 1$ and one returns to the stable regime. In particular, \emph{in surface symmetry one expects at most one bounce}. This is exactly what is described, at least along some causal curve, in Theorem~\ref{thm:bounce_rough}.

We comment upon how the BKL transition map \eqref{eq:kasner_relation_bounce} was found. The idea of BKL was to assume, due to AVTD behaviour, that in the future light cone of any idealized point on the singularity, the metric $\mathbf{g}$ and any matter fields $\mathbf{\Psi}$ are well-approximated by something spatially homogeneous. Their analysis thus reduces to the spatially homogeneous case, where the dynamics reduce to a system of finite dimensional nonlinear autonomous ODEs. The BKL bounce map \eqref{eq:kasner_relation_bounce} then arises from the instability of certain fixed points in this ODE system, or more precisely certain heteroclinic orbits that emanate from such fixed points. On the other hand, if say all the $p_I$ are positive then the associated fixed point is actually stable and the nonlinear behaviour of the ODE system is suppressed.

Finally, we briefly review the mathematical literature regarding the BKL ansatz \eqref{eq:bkl}. Most mathematical works are related to the stable case (i.e.~$p_I > 0$), the biggest breakthrough being the work of Fournodavlos--Rodnianski--Speck \cite{FournodavlosRodnianskiSpeck}, showing the nonlinear stability (outside of symmetry) of the generalized Kasner spacetimes with $p_I > 0$ in the Einstein--scalar field system, and the recent generalization by Oude Groeniger--Petersen--Ringstr\"om \cite{groeniger2023formation}. See also \cite{RodnianskiSpeck1, RodnianskiSpeck2, SpeckS3, BeyerOliynyk, FajmanUrban}. There are also results which involve prescribing the asymptotic data i.e.~$p_I(x), p_{\phi}(x), \omega^I(x)$ in \eqref{eq:bkl} and ``solving backwards'' to find a spacetime which achieves this near-singularity ansatz at leading order, see for instance \cite{AnderssonRendall, DHRW, FournodavlosLuk}. A recent result of the author \cite{li2024scattering} aligns these points of view, describing a Hilbert space isomorphism between regular Cauchy initial data and the asymptotic data in the context of the \emph{linearized} Einstein--scalar field system around Kasner.

In the unstable case, i.e.~in the study of nonlinear bounces, to the best of the authors knowledge all previous work concerns only spatially homogeneous spacetimes, where the dynamics reduce to ODEs as described above. Rigorous mathematical works include studies of solutions for various Einstein--matter systems in Bianchi symmetry \cite{Weaver_bianchi, RingstromBianchi}, as well as a recent work of the author together with Van de Moortel \cite{MeVdM} where we consider spherically symmetric spacetimes with an extra symmetry in the $\partial_x$ direction\footnote{Such spacetimes are often known as Kantowski-Sachs cosmologies.} solving the Einstein--Maxwell--\emph{charged} scalar field model. In particular \cite{MeVdM} introduces already the transition $\Psi(x) \mapsto \acute{\Psi}(x) = \Psi(x)^{-1}$.

Therefore the current article and our companion article \cite{MeGowdyPaper} are the first works to understand BKL bounces, albeit only a single such bounce, \emph{outside of the spatially homogeneous setting}. A natural conjecture would be to understand an analogue of this result for the full $1+3$-dimensional Einstein--scalar field outside of symmetry.

\begin{conjecture}
    There exists an open set of initial data for the Einstein scalar field equations with \underline{no symmetry assumptions} such that the maximal (past) development arising from such initial data terminates in a Kasner-like spacelike singularity at $t = 0$, but such that the intermediate dynamics (between data and the singularity) exhibit one, or potentially a finite number, of BKL bounces.
\end{conjecture}

A resolution of this conjecture would present a key step towards understanding the full heuristics of BKL in $1+3$-dimensional vacuum; the latter problem is substantially more difficult since \cite{bkl71} suggests that in vacuum there are infinitely many BKL bounces and that the ODE dynamics are chaotic. As mentioned above, addition of the scalar field allows study of a regime where there are only finitely many bounces.

\subsection{Sketch of the proof} \label{sub:intro_proof}

The global existence result Theorem~\ref{thm:global_rough} and the bounce result Theorem~\ref{thm:bounce_rough} are proven simultaneously and the proof consists of three major steps. We illustrate the key steps using two of the equations in the Einstein--Maxwell--scalar field system written in the gauge \eqref{eq:surface_sym_intro}, see Section~\ref{sub:emsfss}. These two equations are the $\mu$ evolution equation \eqref{eq:mu_evol} and the wave equation for $\phi$ \eqref{eq:phi_wave}. Letting $\kappa = 0$ for simplicity, these equations are:
\begin{gather}
    r \partial_r \mu = \frac{1}{2} \left( (r \partial_r \phi)^2 + r^2 e^{2(\mu - \lambda)} (\partial_x \phi)^2 + 1 - \frac{Q^2}{r^2}e^{2\mu} \right), \label{eq:mu_evol_intro} \\[0.3em]
    r \partial_r (r \partial_r \phi) = r^2 e^{2 (\mu - \lambda)} \partial_x^2 \phi + r^2 e^{2(\mu - \lambda)} \partial_x (\mu - \lambda) \partial_x \phi - \frac{Q^2}{r^2} e^{2\mu} \, r \partial_r \phi. \label{eq:phi_wave_intro}
\end{gather}

We now describe the three major steps as follows:
\begin{enumerate}[Step 1:]
    \item
        \textbf{Analysis in the spatially homogeneous case:  } Here, this involves setting all terms involving a $\partial_x$-derivative to zero, and understanding the resulting nonlinear ODE system.

    \item
        \textbf{Linearization of the ODE system: } In the second step, we consider the \emph{best possible behaviour} of the terms involving a $\partial_x$-derivative. We do this by taking commuting $\partial_x$ with the ODE system in Step 1, resulting in a new linear ODE system for certain $\partial_x$ derivatives whose coefficients are given by the solution (i.e.~some orbit) of the ODE system in Step 1.

    \item
        \textbf{Energy estimates: } Finally, one must ensure that we can close our argument without loss of derivatives (necessary due to the existence of top order terms such as $\partial_x^2 \phi$ in \eqref{eq:phi_wave_intro}.) One achieves this via $L^2$ energy estimates, which are allowed to blow up but only at a mild rate as $r \to 0$.
\end{enumerate}

We now explain in more detail how each of these steps applies to our simplified system \eqref{eq:mu_evol_intro}--\eqref{eq:phi_wave_intro}. For Step 1, we remove the terms with $\partial_x$-derivatives in these equations. Letting $\mathscr{P} = r \partial_r \phi$ and $\mathscr{Q} = \frac{Q^2}{r^2} e^{2\mu}$, upon removing these terms and rephrasing in terms of $\mathscr{P}$ and $\mathscr{Q}$ we yield the ODE system:
\begin{equation} \label{eq:ode_pq}
    r \partial_r \mathscr{P} = - \mathscr{P} \mathscr{Q}, \qquad r \partial_r \mathscr{Q} = \mathscr{Q} (\mathscr{P}^2 - 1 - \mathscr{Q}).
\end{equation}
We restrict attention to the region $\mathscr{P} > 0, \mathscr{Q} \geq 0$. Then this ODE system contains a line of fixed points at $\mathscr{Q} = 0$, each of which represents an exact generalized Kasner solution to the Einstein--scalar field system (with $\mathbf{F} = 0$). Moreover these fixed points are (orbitally) stable if $\mathscr{P} \geq 1$ and unstable otherwise. (Note that stability is meant in the direction $r \to 0$.)

The dynamics of the ODE system in the remaining region $\mathscr{P} > 0, \mathscr{Q} > 0$ may be described as a union of \emph{heteroclinic orbits} linking an unstable fixed point to a stable fixed point. In fact, the dynamics can be solved exactly due to the fact that $\mathscr{K} \coloneqq \mathscr{P} + \mathscr{P}^{-1} + \mathscr{Q} \mathscr{P}^{-1}$ turns out to be a conserved quantity of the system. As a consequence, these heteroclinic orbits, which we identify as the BKL bounces, link the unstable fixed point $(\mathscr{P} = \alpha, \mathscr{Q} = 0)$ to the stable fixed point $(\mathscr{P} = \alpha^{-1}, \mathscr{Q} = 0)$, for any $\alpha > 1$. The phase portrait for the ODE system is given below in Figure~\ref{fig:phase_portrait}.

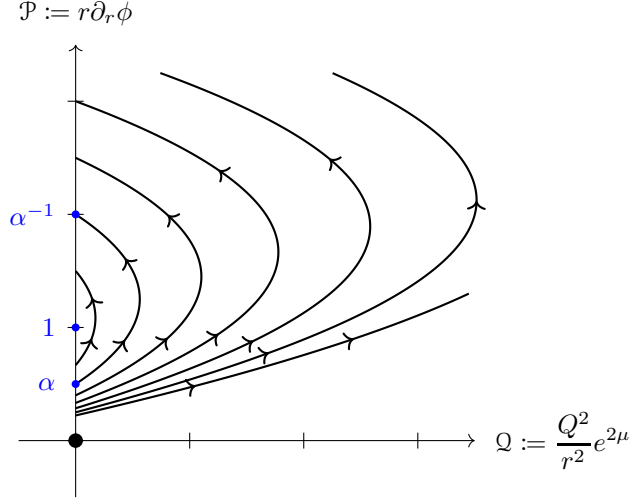
\begin{figure}[ht]
    \centering
    \begin{tikzpicture}[scale=1.5, decoration={markings, 
        mark=at position 0.3 with {\arrow{>}}, mark=at position 0.7 with{\arrow{>}}}]
        %\draw[gray,very thin] (-1.0, 6.0) grid (-1.0, 4.0);
        \draw[->] (-0.5,0) -- (3.5, 0) node[right, xshift=4pt] {$\mathscr{Q} \coloneqq \displaystyle{\frac{Q^2}{r^2} e^{2\mu}}$};
        \draw[->] (0,-0.5) -- (0, 3.5) node[above, yshift=4pt] {$\mathscr{P} \coloneqq r \partial_r \phi$};
        
        \foreach \pos in {0, 1, 2, 3}
          \draw[shift={(\pos,0)}] (0pt,2pt) -- (0pt,-2pt) node[below] {$ $};
        \foreach \pos in {0 ,1,2,3}
          \draw[shift={(0,\pos)}] (2pt,0pt) -- (-2pt,0pt) node[left] {$ $};
        
        \fill (0,0) circle (0.064cm);
        
        \foreach \alf in {1.5, 2, 2.5, 3}
            \draw[thick, variable=\p, domain=1/\alf:\alf, samples=100, postaction={decorate}]
                plot ({-(\p - 1/\alf)*(\p - \alf)}, {\p});
        \foreach \alf in {3.5, 4.0}
            \draw[thick, variable=\p, domain=1/\alf:3.25, samples=100, postaction={decorate}]
                plot ({-(\p - 1/\alf)*(\p - \alf)}, {\p});
        \draw[thick, variable=\p, domain=1/4.5: 1.3, samples=100, postaction={decorate}]
            plot ({-(\p - 1/4.5)*(\p - 4.5)}, {\p});
        
        \fill[blue] (0, 1) circle (1pt)
            node [left ,fill=white,xshift=-4pt] {$1$};
        \fill[blue] (0, 0.5) circle (1pt)
            node [left ,fill=white,xshift=-4pt] {$\alpha$};
        \fill[blue] (0, 2.0) circle (1pt)
            node [left ,fill=white,xshift=-4pt] {$\alpha^{-1}$};
    \end{tikzpicture}
    \captionsetup{justification = centering}
    \caption{Phase portrait showing the dynamics of $\mathscr{P}$ and $\mathscr{Q}$ towards $r = 0$ in the exactly spatially homogeneous case. A bounce is a transition from {\color{blue} $\alpha$} to {\color{blue} $\alpha^{-1}$}.}
    \label{fig:phase_portrait}
\end{figure}

We link this back to our conditions preceding Theorem~\ref{thm:global_rough}. In particular, that the requirement \eqref{eq:weak_subcriticality} of weak subcriticality means that even upon evolution of $\mathscr{P}$ and $\mathscr{Q}$ according to \eqref{eq:ode_pq}, the dynamics will remain in a bounded portion of the $(\mathscr{P}, \mathscr{Q})$-plane, and importantly this will remain true if we introduce small error terms in \eqref{eq:ode_pq}. This will be essential in both Step 2 and Step 3.

We now move to Step 2, which involves commuting \eqref{eq:ode_pq} with $\partial_x$. In our case, it is actually more convenient instead to find evolution equations for $\mathscr{M} = \partial_x \mu = \partial_x \log \mathscr{Q}$ and $\mathscr{N} = \partial_x \mathscr{P}$. The idea is that commutation will yield a \emph{linear} ODE system for $\mathscr{M}$ and $\mathscr{N}$ whose coefficients are functions of the (dynamical) $\mathscr{P}$ and $\mathscr{Q}$, namely the linear system
\begin{equation} \label{eq:ode_mn}
    r \partial_r \mathscr{M} = - \mathscr{Q} \mathscr{M} + \mathscr{P} \mathscr{N}, \qquad r \partial_r \mathscr{N} = - 2 \mathscr{P} \mathscr{Q} \mathscr{N} - \mathscr{Q} \mathscr{N}.
\end{equation}
One then uses \eqref{eq:ode_mn} to find estimates for $\mathscr{M}$ and $\mathscr{N}$ that are \emph{uniform} in the choice of dynamical orbit for $\mathscr{P}$ and $\mathscr{Q}$ (at least for orbits compatible with the boundedness property above). One finds that for any $\upgamma > 0$, one has\footnote{Actually, for the exact system \eqref{eq:ode_mn}, it can be shown that $\mathscr{M} = O((\log r)^2), \mathscr{N} = O(\log r)$, but since in the eventual analysis we encounter error terms it will be more straightforward to allow a little loss.} $\mathscr{M}, \mathscr{N} = O(r^{- \upgamma})$. %The reason for this is that for any $(\mathscr{P}, \mathscr{Q})$ orbit, $\mathscr{Q}$ is small except during the bounce. %S $\mathscr{Q}$ appears in 3 of the 4 coefficients in \eqref{eq:ode_mn}.

The outcome of Step 2 is that heuristically each $\partial_x$-derivative comes with a loss of $r^{- \upgamma}$. Since one may choose $\upgamma < 1$, this confirms AVTD behaviour in the sense that $\partial_x$-derivatives cost less than a $\partial_r$-derivative. For instance, one expects $\partial_x \phi = O(r^{-\upgamma}), \partial_x (\mu - \lambda) = O(r^{-\upgamma})$, and $\partial_x^2 \phi = O(r^{- 2 \upgamma})$ in the equations \eqref{eq:mu_evol_intro}--\eqref{eq:phi_wave_intro}. Using this, as well as that $e^{\mu - \lambda} = O(r)$, one verifies that the terms involving a $\partial_x$-derivative, which we threw away in Step 1, are all of size $O(r^{4 - 2 \upgamma})$. So these are integrable with respect to $\frac{dr}{r}$ towards $r = 0$, and may be treated as negligible errors in the ODEs \eqref{eq:ode_pq}.

Step 3 concerns turning the above paragraph into rigorous bounds. In particular, one must overcome the issue of \emph{derivative loss} due to terms such as $\partial_x^2 \phi$. This is achieved via energy estimates; define for instance the $K$th order energy $\mathcal{E}^{(K)}_{\phi}(r)$ as:
\[
    \mathcal{E}^{(K)}_{\phi}(r) = \frac{1}{2} \int_{\mathbb{S}^1} \left( (r \partial_r \partial_x^K \phi(r, x) )^2 + r^2 e^{2(\mu - \lambda)} (\partial_x^{K+1}\phi(r, x))^2 + (\partial_x^K \phi (r, x))^2 \right) \, dx.
\]
To estimate $\mathcal{E}^{(K)}_{\phi}(r)$, we commute the wave equation \eqref{eq:phi_wave_intro} with $\partial_x^K$, yielding:
\begin{multline*}
    (r \partial_r)^2 \partial_x^K \phi = r^2 e^{2 (\mu - \lambda)} \partial_x^{K+2} \phi - \overbrace{\frac{Q^2}{r^2}e^{2\mu} \left( r \partial_r \partial_x^{K} \phi + (2 \partial_x^K \mu) (r \partial_r \phi) \right)}^{\mathrm{(I)}} \\
    %+ r^2 e^{2(\mu - \lambda)} \partial_x^{K+1} (\mu - \lambda) \partial_x \phi 
    + \underbrace{3 r^2 e^{2(\mu - \lambda)} \partial_x (\mu - \lambda) \partial_x^{K+1} \phi}_{(\mathrm{II})} - \underbrace{\frac{Q^2}{r^2} e^{2\mu} (2 \partial_x \mu) (r \partial_r \partial_x^{K-1} \phi)}_{(\mathrm{III})} + \cdots.
\end{multline*}
There are many more terms in the $\cdots$, but the terms present here will illustrate the ideas necessary in the energy estimate. A standard argument yields the derivative estimate
\[
    \left| r \frac{d}{dr} \mathcal{E}^{(K)}_{\phi}(r) \right| \leq \sqrt{ 2 \mathcal{E}^{(K)}_{\phi}(r) } \cdot \left( \| (\mathrm{I}) + (\mathrm{II}) + (\mathrm{III}) + \cdots \|_{L^2} \right).
\]

It thus suffices to estimate each of $( \mathrm{I} )$, $( \mathrm{II})$ and $(\mathrm{III})$ in $L^2$. Here we need to use Step 1 and Step 2. For the first term in $(\mathrm{I})$, we use that $\frac{Q^2}{r^2}e^{2\mu}$ is less than the maximum value of $\mathscr{Q}$ over all orbits, which is bounded, while $\|r \partial_r \partial_x^K \phi\|_{L^2} \leq \sqrt{2 \mathcal{E}^{(K)}_{\phi}(r)}$ by definition. For the second term in $(\mathrm{I})$, we also use that $r \partial_r \phi$ is less than the maximum value of $\mathscr{P}$, while $\partial_x^K \mu$ is controlled by (the square root) of some other $K$th order energy. We thus conclude that for some constant $A_*$ depending only on $\upeta$,
\[
    \| (\mathrm{I}) \|_{L^2} \leq A_* \sqrt{ 2 \mathcal{E}_{\phi}^{(K)} (r) }.
\]

For the expression $(\mathrm{II})$, we use that $\| r e^{(\mu - \lambda)} \partial_x^{K+1} \phi \|_{L^2} \leq \sqrt{2 \mathcal{E}_{\phi}^{(K)}(r)}$, while the prefactor $r e^{\mu - \lambda} \partial_x(\mu - \lambda)$ is of size $O(r^{2 - \upgamma})$ by Step 2. So there is a constant $C_{\upeta, \upzeta, K}$ depending on $\upeta$, the data and also\footnote{The dependence on the regularity index $K$ is because there are secretly more terms like $(\mathrm{II})$ hidden in $\cdots$.} $K$ so that
\[
    \| (\mathrm{II}) \|_{L^2} \leq C_{\upeta, \upzeta, K} \, r^{2 - \upgamma} \sqrt{ 2 \mathcal{E}_{\phi}^{(K)} (r) }.
\]
Finally, in $(\mathrm{III})$, by Step 1 and Step 2 we know that $\frac{Q^2}{r^2}e^{2\mu} (2 \partial_x \mu)$ is $O(r^{- \upgamma})$. This seems alarming, since this is not integrable with respect to $\frac{dr}{r}$ as $r \to 0$. But the other factor in $(\mathrm{III})$, $r \partial_r \partial_x^{K-1} \phi$ is not dependent on the $K$th order energy, but instead the $(K - 1)$th order energy $\sqrt{2 \mathcal{E}^{(K-1)}_{\phi}(r)}$. So:
\[
    \| (\mathrm{III}) \|_{L^2} \leq C_{\upeta, \upzeta, K} \, r^{- \upgamma} \sqrt{ 2 \mathcal{E}_{\phi}^{(K-1)} (r) }.
\]

Combining all of the above, we obtain the following derivative estimate for $K \geq 1$:
\begin{equation} \label{eq:energy_der_intro}
    \left| r \frac{d}{dr} \mathcal{E}^{(K)}_{\phi}(r) \right| \leq 2 ( A_* + C_{\upeta, \upzeta, K} \, r^{2 - \upgamma} ) \mathcal{E}^{(K)}_{\phi}(r)  + 2  C_{\upeta, \upzeta, K} \, r^{- \upgamma} \sqrt{ \mathcal{E}^{(K)}_{\phi} (r) } \cdot \sqrt{ \mathcal{E}^{(K-1)}_{\phi}(r)} + \cdots.
\end{equation}
When $K = 0$, the last term on the right hand side is absent. Because of the $2 A_*$, even the $0$th order energy $\mathcal{E}^{(0)}(r)$ will blow up as $r^{-2 A_*}$ as $r \to 0$. Furthermore, the appearance of $r^{- \upgamma}$ means that as $K$ increases the rate of blow up also increases. But in any case, one uses this expression to find an energy bound
\[
    \mathcal{E}^{(K)}(r) \leq D_{\upeta, K, \upzeta} \, r^{-2 A_* - 2 K \upgamma},
\]
where $D_{\upeta, K, \upzeta}$ depends on $\upeta$ and the data, as well as the regularity index $K$. It is crucial that $A_*$ is \emph{independent of $K$}. This is because applying an $L^2$-$L^{\infty}$ interpolation argument to the above bound, one shows e.g.~that $\partial_x^2 \phi = O(r^{-2 \gamma - \updelta})$, where $\updelta \to 0$ as $N \to \infty$, where $N$ is the maximum regularity index for which we perform the energy estimate.

To apply Steps 1 to 3 in the nonlinear problem, one uses a standard bootstrap argument; note we actually perform Step 3 first. That is, one first bootstraps $L^{\infty}$-bounds on $0$th order and $1$st order quantities, see \eqref{eq:bootstrap_maxwell}--\eqref{eq:bootstrap_d2x}, and uses this to derive the energy estimate \eqref{eq:energy_der_intro}. From the energy estimate, the interpolation argument is used to control expressions such as $\partial_x^2 \phi$, which are then treated as error terms in the ODE analysis of Step 1 and Step 2. The nonlinear ODE analysis then allows us to improve the bootstrap assumptions \eqref{eq:bootstrap_maxwell}--\eqref{eq:bootstrap_d2x}, completing both the proof of global existence (Theorem~\ref{thm:global_rough}) and first part of the bounce theorem (Theorem~\ref{thm:bounce_rough}). The remainder of Theorem~\ref{thm:bounce_rough} follows from more detailed ODE analysis.% when $Q = 0$ and when $Q \neq 0$.

\subsection{Comparison to results in Gowdy symmetry} \label{sub:intro_gowdy}

In the final part of the introduction, we make an explicit comparison between the present article and those of our companion article \cite{MeGowdyPaper}, since there is almost a direct translation between both the methods and the results of the two papers. 

Gowdy symmetric spacetimes have a long history in the mathematical literature, see \cite{RingstromGowdyReview} and references therein, and concern metrics of the form:
\begin{equation} \label{eq:gowdy}
    \mathbf{g} = - t^{-\frac{1}{2}}e^{\frac{\lambda}{2}} ( - dt^2 + d \theta^2) + t [ e^P (d \sigma + Q d \delta)^2 + e^{-P}d \delta^2],
\end{equation}
where $P, Q, \lambda: (0, +\infty) \times \mathbb{S}^1 \to \R$ depend only on the coordinates $t \in (0, + \infty)$ and $\theta \in \mathbb{S}^1$. The Einstein \emph{vacuum} equations \eqref{eq:einstein} with $\mathbf{T} = 0$ then take the form of a $1+1$-dimensional wave--transport system for $P, Q, \lambda$. Further, the boundary $\{ t = 0 \}$ is (generically) a spacelike singularity exhibiting curvature blowup; in other words Strong Cosmic Censorship \cite{RingstromGowdySCC2} is known for such spacetimes.

With this setup complete, we now provide our ``dictionary'' linking the current paper to \cite{MeGowdyPaper}:
\begin{itemize}
    \item
        Due to the reduction to $1+1$-dimensions, the Gowdy spacetimes have one dynamical time variable $t$ and spatial variable $\theta \in \mathbb{S}^1$, in comparison to $r$ and $x \in \mathbb{S}^1$ in our surface symmetric spacetimes. 

    \item
        Just as our surface symmetric spacetimes may be described as Kasner--like using Theorem~\ref{thm:asymp_smooth}, with the exponents given using $\Psi(x) = \lim_{r \to 0} r \partial_r \phi(r, x)$ and the correspondence \eqref{eq:kasnerlike_intro}, the $\{ t = 0 \}$ singularity of Gowdy spacetimes can also be described as Kasner--like using the asymptotic quantity $V(\theta) = \lim_{t \to 0} (- t \partial_t P(t, \theta))$ and the analogous correspondence
        \begin{equation} \label{eq:kasnerlike_intro_gowdy}
            p_1(\theta) = \frac{V^2(\theta) - 1}{V^2(\theta) + 3}, \quad p_2(\theta) = \frac{2 - 2 V(\theta)}{V^2(\theta) + 3}, \quad  p_3(\theta) = \frac{2 + 2 V(\theta)}{V^2(\theta) + 3}.
        \end{equation}
        Note that since our Gowdy spacetimes solve the Einstein \emph{vacuum} equations the exponents are distinct and verify the Kasner relations \eqref{eq:kasner_relation} with $p_{\phi} \equiv 0$.

        It was proved in \cite{RingstromGowdySCC1} that $V(\theta)$ always exists, and that for \emph{generic} solutions one has $0 < V(\theta) < 1$ for all $\theta \in \mathbb{S} \setminus S$, where $S$ is a finite set consisting of so-called ``spikes'' where $V$ fails to be continuous.

    \item
        There is a distinguished class of Gowdy spacetimes called the \emph{polarized Gowdy spacetimes} for which $Q$ in \eqref{eq:gowdy} \emph{vanishes identically}. For such polarized Gowdy spacetimes, the asymptotic quantity $V(\theta)$ is smooth and may be prescribed freely i.e.~no longer has the restriction $0 < V(\theta) < 1$.

        This can be held in direct analogy with the \emph{uncharged} Einstein--scalar field solutions in surface symmetry; in light of Theorem~\ref{thm:asymp_smooth}, at least if $\kappa \in \{ 0, +1\}$ then the function $\Psi(x)$ can similarly be prescribed freely. On the other hand, if $\mathbf{F} \neq 0$ then necessarily $|\Psi(x)| \geq 1$.

    \item
        Moving on to the results, \cite[Theorem 1.1]{MeGowdyPaper} and Theorem~\ref{thm:global_rough} both characterise initial data for the (unpolarised) Gowdy system at $t = t_0 > 0$ and the Einstein--Maxwell--scalar field system at $r = r_0 > 0$ such that one reaches a singularity at $\{ t = 0 \}$ and $\{ r = 0 \}$ respectively.

        In both cases, this class of initial data is engineered to be large enough to allow for spacetimes with nontrivial intermediate dynamics i.e.~between $0 < t < t_0$ or $0 < r < r_0$ suggestive of BKL bounces. For instance, in \cite[Theorem 1.1]{MeGowdyPaper} we allow $\dot{P}_D \coloneqq t \partial_t P(t_0, \theta)$ satisfying $0 < - \dot{P}_D < 2$.

    \item
        Indeed, \cite[Theorem 1.2]{MeGowdyPaper} is the analogue of our bounce result Theorem~\ref{thm:bounce_rough}, proving that on any causal curve $\gamma$ directed towards the singularity, certain quantities an ODE reminiscent of BKL bounces. That is, defining $\mathscr{P}_{Gowdy} = - t \partial_t P (\gamma(t))$ and $\mathscr{Q}_{Gowdy} = e^P t \partial_{\theta} Q (\gamma(t))$, \cite[Theorem 1.2]{MeGowdyPaper} derives a $2$-dimensional autonomous ODE for $(\mathscr{P}_{Gowdy}, \mathscr{Q}_{Gowdy})$, plus error terms.

        Just as a BKL bounce in the surface symmetric Einstein--Maxwell--scalar field system corresponds to a heteroclinic orbit $(\mathscr{P}, \mathscr{Q}) = ({\color{blue} \alpha}, 0) \mapsto (\mathscr{P}, \mathscr{Q}) = ({\color{blue} \alpha^{-1}}, 0)$ where $0 < {\color{blue} \alpha} < 1$, see Figure~\ref{fig:phase_portrait}, a BKL bounce in the Gowdy symmetric vacuum picture corresponds to a heteroclinic orbit $(\mathscr{P}_{Gowdy}, \mathscr{Q}_{Gowdy}) = ({\color{blue} \alpha}, 0) \mapsto (\mathscr{P}_{Gowdy}, \mathscr{Q}_{Gowdy}) = ({\color{blue} 2 - \alpha}, 0)$ where $1 < {\color{blue} \alpha} < 2$.

        Furthermore, via the correspondence \eqref{eq:kasnerlike_intro_gowdy}, the map $V(\theta) \mapsto 2 - \acute{V}(\theta)$ is exactly the bounce map \eqref{eq:kasner_relation_bounce} found by BKL (at least after switching the roles of $p_1$ and $p_2$).

    \item
        Finally, both this article and \cite{MeGowdyPaper} provide a description of BKL bounces as an \emph{instability mechanism}. Here, this is achieved via perturbations of \emph{uncharged} solutions of the Einstein--scalar field equations to \emph{charged} solutions of the Einstein--Maxwell--scalar field equations, see e.g.~Figure~\ref{fig:bounce_instability} and Corollary~\ref{cor:bounce}. In \cite{MeGowdyPaper}, this is achieved via perturbations of \emph{polarized Gowdy} solutions to generic \emph{unpolarized Gowdy} solutions, see Figure 1 and Corollary~2.3 of \cite{MeGowdyPaper}.

    \item
        There is, however, one feature of \cite{MeGowdyPaper} that does not appear in the present article, namely the aforementioned ``spikes''. Spikes correspond to special orbits of the ODE system converging to unstable fixed points, and generically occur at finitely many $\theta \in \mathbb{S}^1$ in Gowdy spacetimes.

        We comment that spikes do not occur for the Einstein--Maxwell--scalar field in surface symmetry since the rigidity of the electromagnetic field means that for $Q \neq 0$ it is impossible for any orbit to converge to an unstable fixed point. If one considered the Einstein--Maxwell--\emph{charged} scalar field system where $\phi$ is dynamically coupled to $\mathbf{F}$, then spikes would similarly become an issue.
\end{itemize}

\subsection*{Acknowledgements}

The author thanks Mihalis Dafermos for valuable advice in the writing of this manuscript. We also thank Igor Rodnianski and Hans Ringstr\"om for insightful discussions and suggestions.

%auto-ignore

\section{The Einstein--Maxwell--scalar field system in surface symmetry} \label{sec:equations}

\subsection{Surface symmetric spacetimes} \label{sub:eqns_sss}

A \emph{surface symmetric} spacetime contains $2$ linearly independent Killing vector fields, which generate a symmetry group whose orbits are rescaled copies of a surface $\Sigma$ of constant sectional curvature. Restricting to the \emph{trapped region} where each copy of $\Sigma$ is a trapped\footnote{Due to the time reversibility of Einstein's equations, we are agnostic about whether surfaces are future-trapped or past-trapped.} surface (or equivalently having a timelike area-radius function $r$), such spacetime metrics $\mathbf{g}$ may be written in the following form:
\begin{equation} \label{eq:surface_sym}
    \mathbf{g} = - e^{2 \mu} dr^2 + e^{2 \lambda} d x^2 + r^2 d \sigma_{\Sigma}.
\end{equation}

Here $d \sigma_{\Sigma}$ denotes the standard metric on the surface $\Sigma$, rescaled to have constant sectional curvature $\kappa \in \{1, 0 , -1\}$. That is, in some local coordinate chart,
\[
    d \sigma_{\Sigma} = d \theta^2 + \sin_{\kappa}^2 \theta d \varphi^2 = 
    \begin{cases}
        d \theta^2 + \sin^2 \theta \, d \varphi^2 & \quad \text{ if } \Sigma \text{ has universal cover } \mathbb{S}^2, \\[-0.2em]
        d \theta^2 + d \varphi^2 & \quad \text{ if } \Sigma \text{ has universal cover } \R^2, \\[-0.2em]
        d \theta^2 + \sinh^2 \theta \, d \varphi^2 & \quad \text{ if } \Sigma \text{ has universal cover } \mathbb{H}^2.
    \end{cases}
\]
%We shall define $\kappa$ to be $+1$, $0$, or $-1$ depending on whether $\Sigma$ is $\mathbb{S}^2$, $\mathbb{R}^2$, or $\mathbb{H}^2$ respectively.

If a surface symmetric spacetime $(\mathcal{M}, \mathbf{g})$ solves the Einstein equations \eqref{eq:einstein}, the matter fields $\mathbf{\Phi}$ must be surface symmetric in the following sense: the energy momentum tensor $\mathbf{T}[\mathbf{\Phi}]$ is described purely by (1) its projection onto the quotient manifold $\mathcal{Q}$, and (2) its trace with respect to $d \sigma_{\Sigma}$. %-part of the metric. 
In other words, the degrees of freedom of $\mathbf{T}_{\mu\nu}[\mathbf{\Phi}]$ are a two-tensor $T_{\mu\nu}$ on $\mathcal{Q}$ and a function $\mathbf{S}$ on $\mathcal{Q}$, given by
\[
    T_{\mu\nu} = \mathbf{T}_{\mu\nu} \text{ for } \mu, \nu \in \{ r, x \}, \quad \text{ and } \quad \mathbf{S} = \tr_{d \sigma_{\Sigma}} \mathbf{T} = \mathbf{T}_{\theta \theta} + \sin_{\kappa}^{-2} \theta \, \mathbf{T}_{\varphi \varphi}.
\]

Given this decomposition of the energy-momentum tensor $\mathbf{T}[\mathbf{\Phi}]$, one expresses the Einstein equations \eqref{eq:einstein} with respect to the coordinates of \eqref{eq:surface_sym} as follows, where $\tr_{\mathcal{Q}} T = - e^{-2\mu} T_{rr} + e^{-2 \lambda} T_{xx}$.
\begin{gather}
    \label{eq:lambda_evol_T}
    r \partial_r \lambda = r^2 T_{rr} - \frac{1}{2} \left( 1 + \kappa e^{2 \mu} \right), \\[0.8em]
    \label{eq:mu_evol_T}
    r \partial_r \mu = r^2 e^{2 (\mu - \lambda)} T_{xx} + \frac{1}{2} \left( 1 + \kappa e^{2 \mu} \right), \\[0.8em]
    \label{eq:mu_constraint_T}
    \partial_x \mu = r T_{xr}, \\[0.8em]
    \label{eq:mu_wave_T}
    \begin{split}
        (r \partial_r)^2 \mu - r^2 e^{2 (\mu- \lambda)} \partial_x^2 \mu = 
        & (r \partial_r \lambda) (r^2 \tr_{\mathcal{Q}} T + \kappa) e^{2\mu} + r \partial_r \left[ \left( r^2 \tr_{\mathcal{Q}} T + \kappa \right) e^{2 \mu} \right]  \\
        &\quad + r^2 e^{2 (\mu - \lambda)} \partial_x \mu \, \partial_x (\mu - \lambda) + r^2 e^{2 (\mu- \lambda)} \left[ T_{xx} - \frac{1}{2} \tr_{\mathcal{Q}} T e^{2 \lambda} - \frac{\mathbf{S}}{2} e^{2 \lambda} \right].
    \end{split}
\end{gather}
We refer the reader to \cite[Equations (1.3)--(1.6)]{ReinEinsteinVlasov} for the specific case where the matter field $\mathbf{\Phi}$ is Vlasov matter (though that our $r$ and $x$ are instead $t$ and $r$ in \cite{ReinEinsteinVlasov}). Compared to \cite{ReinEinsteinVlasov}, we have manipulated the equations somewhat so that (a) our canonical choice of time derivative is $r \partial_r$ and (b) equation \eqref{eq:mu_wave_T} has wave-like structure.

The claimed wave-like structure is made clear by the observation that with respect to the coordinates of \eqref{eq:surface_sym}, the wave equation \eqref{eq:wave} for $\phi$ is given by:
\begin{equation} \label{eq:phi_wave_T}
    (r \partial_r)^2 \phi - r^2 e^{2(\mu-\lambda)} \partial_x^2 \phi = e^{2 \mu} (r \partial_r \phi) \left( r^2 \tr_{\mathcal{Q}} T + \kappa \right) + r^2 e^{2(\mu-\lambda)} \partial_x \phi \, \partial_x (\mu - \lambda).
\end{equation}
Thus \eqref{eq:mu_wave_T} has similar structural features to \eqref{eq:phi_wave_T}, albeit with a more complicated right hand side.

\subsection{The Einstein--Maxwell--scalar field system} \label{sub:emsfss}

We now specialize the above to the Einstein--Maxwell--scalar field system $(\mathcal{M}, \mathbf{g}, \phi, \mathbf{F})$. The electromagnetic tensor $\mathbf{F}$ is represented by the following $2$-form, where $Q(r, x)$ is a function in $\mathcal{Q}$:
\begin{equation} \label{eq:maxwell_F}
    \mathbf{F} = \frac{2Q}{r^2} e^{\mu + \lambda} dr \wedge dx.
\end{equation}
For the Einstein--Maxwell--scalar field system in surface symmetry, Maxwell's equations \eqref{eq:maxwell} imply that $dQ = 0$, i.e.~that $Q$ is constant\footnote{For instance, one checks that $*\mathbf{F} = 2 Q d \theta \wedge \sin_{\kappa}\!\theta \, d \varphi$, while the second equation in \eqref{eq:maxwell} is equivalent to $d\! * \!\mathbf{F} = 0$. Note that Lemma~\ref{lem:maxwell_F} would fail, on the other hand, if the Maxwell equation has a source and $d\! * \!\mathbf{F} \neq 0$ e.g.~models involving a \emph{charged} scalar field.}: %In other words, the Maxwell field experiences no dynamics in this model, as we record in the following lemma:
\begin{lemma} \label{lem:maxwell_F}
    Let $(\mathcal{M}, \mathbf{g}, \phi, \mathbf{F})$ be a surface symmetric solution to the Einstein--Maxwell--scalar field system \eqref{eq:einstein}--\eqref{eq:maxwell_em} with metric $\mathbf{g}$ given by \eqref{eq:surface_sym} and Maxwell tensor $\mathbf{F}$ given by \eqref{eq:maxwell_F}. Then $Q(r, x)$ is a constant.
\end{lemma}

We now specialize the remaining equations of Section~\ref{sub:eqns_sss} to the Einstein--Maxwell--scalar field model \eqref{eq:einstein}--\eqref{eq:maxwell_em} in the following proposition.

\begin{proposition}[Surface symmetric Einstein--Maxwell--scalar field system] \label{prop:emsfss_surface_sym} 
    Let $(\mathcal{M}, \mathbf{g}, \phi, \mathbf{F})$ be a surface symmetric solution to the Einstein--Maxwell--scalar field system \eqref{eq:einstein}--\eqref{eq:maxwell_em}, with metric $\mathbf{g}$ given by \eqref{eq:surface_sym} and Maxwell tensor $\mathbf{F}$ given by \eqref{eq:maxwell_F}. Then the dynamical quantities $(\lambda, \mu, \phi)$ satisfy the following equations for $(r, x) \in \mathcal{Q}$.
    \begin{gather}
        \label{eq:lambda_evol}
        r \partial_r \lambda = \frac{1}{2} \left( (r \partial_r \phi)^2 + r^2 e^{2(\mu - \lambda)} (\partial_x \phi)^2 - 1 - \kappa e^{2 \mu} + \frac{Q^2}{r^2} e^{2 \mu} \right), \\[0.8em]
        \label{eq:mu_evol}
        r \partial_r \mu = \frac{1}{2} \left( (r \partial_r \phi)^2 + r^2 e^{2(\mu - \lambda)} (\partial_x \phi)^2 + 1 + \kappa e^{2 \mu} - \frac{Q^2}{r^2} e^{2 \mu} \right), \\[0.8em]
        \label{eq:mu_constraint}
        \partial_x \mu = r \partial_r \phi \, \partial_x \phi, \\[0.8em]
        \label{eq:mu_wave}
        \begin{split}
            (r \partial_r)^2 \mu - r^2 e^{2 (\mu- \lambda)} \partial_x^2 \mu = 
            & (r \partial_r \lambda) \left( - \frac{Q^2}{r^2} + \kappa \right) e^{2\mu} + r \partial_r \left[ \left( - \frac{Q^2}{r^2} + \kappa \right) e^{2\mu} \right]  \\
            &\quad + r^2 e^{2 (\mu - \lambda)} \partial_x \mu \, \partial_x (\mu - \lambda) + 2 r^2 e^{2 (\mu- \lambda)} (\partial_x \phi)^2 - \frac{Q^2}{r^2} e^{2 \mu}, 
        \end{split} \\[0.8em]
        \label{eq:phi_wave}
        (r \partial_r)^2 \phi - r^2 e^{2(\mu-\lambda)} \partial_x^2 \phi = (r \partial_r \phi) \left( - \frac{Q^2}{r^2} + \kappa \right) e^{2\mu} + r^2 e^{2(\mu-\lambda)} \partial_x \phi \, \partial_x (\mu - \lambda).
    \end{gather}
    Furthermore $\mu - \lambda$ obeys the transport equation 
    \begin{equation} \label{eq:mulambda_evol}
        r \partial_r (\mu - \lambda) = 1 + \left( - \frac{Q^2}{r^2} + \kappa \right) e^{2\mu}.
    \end{equation}
\end{proposition}

\begin{proof}
    The equations \eqref{eq:lambda_evol}--\eqref{eq:phi_wave} are derived directly from \eqref{eq:lambda_evol_T}--\eqref{eq:phi_wave_T}, upon insertion of the following expressions for the decomposition of the energy momentum tensor $\mathbf{T}_{\mu\nu} = \mathbf{T}_{\mu\nu} [\phi] + \mathbf{T}_{\mu\nu} [\mathbf{F}]$ with respect to the coordinates of \eqref{eq:surface_sym}:
    \begin{gather*}
        T_{rr}[\phi] = e^{2(\mu - \lambda)} T_{xx}[\phi] = \frac{1}{2} \left( (\partial_r \phi)^2 + e^{2(\mu - \lambda)} (\partial_x \phi)^2 \right), \qquad \mathbf{S}[\phi] = e^{-2\mu} (r \partial_r \phi)^2 - r^2 e^{-2 \lambda} (\partial_x \phi)^2, \\[0.4em]
        T_{rr}[\mathbf{F}] = - e^{2(\mu - \lambda)} T_{xx}[\mathbf{F}] = \frac{1}{2} \frac{Q^2}{r^2} e^{2\mu}, \qquad \mathbf{S}[\mathbf{F}] = \frac{4 Q^2}{r^2}.
    \end{gather*}
    We leave the details to the reader.
\end{proof}

\subsection{The initial value problem} \label{sub:ivp}

Prior to our main theorems in Section~\ref{sec:theorem}, we briefly outline local existence for the surface symmetric Einstein--Maxwell--scalar field system as given in Section~\ref{sub:emsfss}. See also statements similar to the following Proposition~\ref{prop:lwp} in \cite{ReinEinsteinVlasov} for the surface symmetric Einstein--Vlasov system and in \cite{TegankongNoutcheguemeRendall} for the surface symmetric Einstein--Vlasov--scalar field system.

However in the present article, we treat the equations differently from \cite{ReinEinsteinVlasov}, as to emphasize the wave-like nature of \eqref{eq:mu_evol}. As a result, though our system is reduced to a $2$-dimensional problem and one can apply the method of characteristics as in \cite{ReinEinsteinVlasov}, we clarify that local existence also holds in energy spaces $H^{s+1} \times H^s$, for $s > \frac{1}{2}$.

\begin{proposition}[Local well-posedness] \label{prop:lwp}
    Let $Q \in \R$ be a constant and $\kappa \in \{-1, 0, +1\}$. Consider solutions $(\phi, \mu, \lambda)$ to the surface symmetric Einstein--Maxwell--scalar field system \eqref{eq:lambda_evol}--\eqref{eq:phi_wave} for $(r, x) \in I \times \mathbb{S}^1$, where $I \subset (0, + \infty)$ is an interval.

    Consider data as follows: let $r_0 > 0$ and let $\phi_D, \mu_D, \lambda_D, \dot{\phi}_D, \dot{\mu}_D, \dot{\lambda}_D: \mathbb{S}^1 \to \R$ be functions with regularity to be specified shortly, satisfying the following three \underline{constraint equations}:
    \begin{gather*}
        \dot{\mu}_D = \frac{1}{2} \left( \dot{\phi}_D^2 + r_0^2 e^{2(\mu_D - \lambda_D)} (\partial_x \phi_D)^2 + 1 + \kappa e^{2\mu}_D - \frac{Q^2}{r_0^2} e^{2 \mu_D} \right), \\[0.6em]
        \dot{\lambda}_D = \frac{1}{2} \left( \dot{\phi}_D^2 + r_0^2 e^{2(\mu_D - \lambda_D)} (\partial_x \phi_D)^2 - 1 - \kappa e^{2\mu}_D + \frac{Q^2}{r_0^2} e^{2 \mu_D} \right), \\[0.6em]
        \partial_x \mu_D = \dot{\phi}_D \, \partial_x \phi_D.
    \end{gather*}
    Let the Banach space $X$ be either $C^{k+1} \times C^k$ for integers $k \geq 0$, or $H^{s+1} \times H^s$ for real numbers $s > \frac{1}{2}$. Then given $(f_D, \dot{f}_D) \in X$ for $f \in \{\phi, \mu, \lambda\}$, there exists a maximal interval $I$ containing $r_0$ and a unique solution $(\phi, \mu, \lambda)$ of the surface symmetric Einstein--Maxwell--scalar field system \eqref{eq:lambda_evol}--\eqref{eq:phi_wave} with 
    \begin{gather}
        (f, r \partial_r f) \in C(I, X) \text{ and } (f, r \partial_r f) |_{r = r_0} = (f_D, \dot{f}_D) \text{ for } f \in \{ \phi, \mu, \lambda \}. \label{eq:f_lwp}
    \end{gather}
\end{proposition}

\begin{proof} (Sketch)
    We prove local-wellposedness for this system by identifying several of the equations in \eqref{eq:lambda_evol}--\eqref{eq:phi_wave} as evolution equations for which one may apply standard local well-posedness, and the remaining equations as constraint equations which are true at $r = r_0$ and are then propagated.% through the remainder of spacetime.

    The equations we identify as evolution equations are \eqref{eq:mu_wave}, \eqref{eq:phi_wave} and \eqref{eq:mulambda_evol}. Note that \eqref{eq:mu_wave} and \eqref{eq:phi_wave} are quasilinear wave equations with respect to the metric $\mathbf{g}$ in \eqref{eq:surface_sym}, while \eqref{eq:mulambda_evol} is a transport equation whose transport vector field $r \partial_r$ is timelike with respect to $\mathbf{g}$. Thus:
    \begin{itemize}
        \item
            If $X = C^{k+1} \times C^k$, then one finds $(\phi, \mu, \lambda)$ solving \eqref{eq:mu_wave}--\eqref{eq:mulambda_evol} and achieving \eqref{eq:f_lwp} by applying the method of characteristics, with respect to the characteristic vector fields $L = \partial_r + e^{\mu - \lambda} \partial_x, \underline{L} = \partial_r - e^{\mu - \lambda} \partial_x$ for \eqref{eq:mu_wave}--\eqref{eq:phi_wave} and the timelike vector field $\partial_r = \frac{1}{2}( L + \underline{L})$ for \eqref{eq:mulambda_evol}.

        \item
            If $X = H^{s+1} \times H^s$ for $s > \frac{1}{2}$, then the existence and uniqueness of $(\phi, \mu, \lambda)$ solving \eqref{eq:mu_wave}--\eqref{eq:mulambda_evol} and achieving \eqref{eq:f_lwp} follows from the local well-posedness for systems of quasilinear wave equations, see for instance Hughes--Kato--Marsden \cite{HughesKatoMarsden}. This is easily extended to the wave--transport system \eqref{eq:mu_wave}--\eqref{eq:mulambda_evol}, particularly since the right hand side of the transport equation \eqref{eq:mulambda_evol} depends only on $\mu$ and not on any of its derivatives.
    \end{itemize}

    It remains to show that the remaining equations \eqref{eq:lambda_evol}--\eqref{eq:mu_constraint} are satisfied. In fact, since we already have \eqref{eq:mulambda_evol} it suffices to show that \eqref{eq:mu_evol}--\eqref{eq:mu_constraint} are satisfied. We do this via \emph{constraint propagation}. The key computation is that from \eqref{eq:lambda_evol}--\eqref{eq:mu_constraint}, if one defines:
    \[
        \mathscr{C}_1 = r \partial_r \mu - \frac{1}{2} \left( (r \partial_r \phi)^2 + r^2 e^{2(\mu - \lambda)} (\partial_x \phi)^2 + 1 + \kappa e^{2\mu} - \frac{Q^2}{r^2} e^{2\mu} \right), \quad \mathscr{C}_2 = \partial_x \mu - r \partial_r \phi \, \partial_x \phi,
    \]
    then one may derive the equations
    \[
        r \partial_r \mathscr{C}_1 = r e^{\mu - \lambda} \partial_x \mathscr{C}_2 + 2 \left( - \frac{Q^2}{r^2} + \kappa \right) e^{2\mu} \, \mathscr{C}_1, \quad r \partial_r \mathscr{C}_2 = r e^{\mu - \lambda} \partial_x \mathscr{C}_1 + 2 \left( 1 + \left( - \frac{Q^2}{r^2} + \kappa \right)  e^{2\mu} \right) \mathscr{C}_2.
    \]

    So $(\mathscr{C}_1, \mathscr{C}_2)$ obey a homogeneous first order linear hyperbolic system with coefficients determined by the previously found solution to \eqref{eq:lambda_evol}--\eqref{eq:mu_constraint}. Therefore, by considering an energy of the form $\mathcal{E}_{\mathscr{C}}(r) = \frac{1}{2} \int_{\mathbb{S}^1} (\mathscr{C}_1^2(r, x) + \mathscr{C}_2^2(r,x) ) \, dx$, integration by parts yields that
    \[
        \left| r \frac{d}{dr} \mathcal{E}_{\mathscr{C}}(r) \right| \leq C \cdot \mathcal{E}_{\mathscr{C}}(r)
    \]
    for some $C$ depending on the solution $(\phi, \mu, \lambda)$. Since the initial data constraints imply $\mathcal{E}_{\mathscr{C}}(r_0) = 0$, we have that $\mathcal{E}_{\mathscr{C}}(r) = 0$ for all $r \in I$. Therefore $\mathscr{C}_1, \mathscr{C}_2 \equiv 0$ and the remaining equations \eqref{eq:lambda_evol}--\eqref{eq:mu_constraint} hold as required.
\end{proof}

Along with the local well-posedness result of Proposition~\ref{prop:lwp}, we also include the following continuation criteria, showing that if the past endpoint of $I$ is not $0$, then the quantity $\mu$ must blow up.

\begin{lemma}[Continuation criterion] \label{lem:continuation}
    Let $(\phi, \mu, \lambda)$ be a solution of the surface symmetric Einstein--Maxwell--scalar field system \eqref{eq:lambda_evol}--\eqref{eq:phi_wave}, arising from initial data as in Proposition~\ref{prop:lwp}. Then if $r_p = \inf I \in [0, r_0)$ is the past endpoint of $I$, then either $r_p = 0$, or one has
    \begin{equation}
        \sup_{(r, x) \in (r_p, r_0) \times \mathbb{S}^1} \mu (r, x) = + \infty.
    \end{equation}
\end{lemma}

\begin{proof}
    Suppose for contradiction that $r_p > 0$ and $\mu (r, x) \leq A < +\infty$ for all $(r, x) \in (r_p, r_0] \times \mathbb{S}^1$. By standard continuation criteria for the quasilinear wave--transport system \eqref{eq:mu_wave}-\eqref{eq:mulambda_evol} %-- which apply both for $X = C^{k+1} \times C^k$ and for $X = H^{s + 1} \times H^s$ -- 
    it will suffice to show that $\partial_x f$ and $\partial_r f$ remain uniformly bounded for $(r, x) \in (r_p, r_0]$ for each of $f = \phi, \mu, \lambda$.

    The first observation is that if $\mu(r, x) \leq A$, then $e^{2\mu} \leq e^{2A}$ and $\frac{Q^2}{r^2} e^{2\mu} \leq Q^2 r_p^{-2} e^{2A}$. So the coefficients appearing in the wave equations \eqref{eq:phi_wave} and \eqref{eq:mu_wave} are uniformly bounded. We apply this first to the wave equation \eqref{eq:phi_wave} for $\phi$. This equation can be shown to be equivalent to:
    \[
        (\partial_r \pm e^{\mu - \lambda} \partial_x ) (\partial_r \phi \mp e^{\mu - \lambda} \partial_x \phi) = \frac{1}{r} \left( - \frac{Q^2}{r^2} + \kappa \right) e^{2\mu} (\partial_r \phi \pm e^{\mu - \lambda} \partial_x \phi).
    \]

    Therefore, if one defines 
    \[
        A(\tilde{r}) = \max \left \{ \sup_{x \in \mathbb{S}^1} | \partial_r \phi (\tilde{r}, x) + e^{\mu - \lambda} \partial_x \phi(\tilde{r}, x) |, \sup_{x \in \mathbb{S}^1} |\partial_r \phi(\tilde{r}, x) - e^{\mu - \lambda} \partial_x \phi (\tilde{r}, x)| \right \},
    \] 
    integrating the above integral curves of $\partial_r \pm e^{\mu - \lambda} \partial_x$ yields the integral estimate
    \[
        A(r) \leq A(r_0) + r_p^{-1} (Q^2 r_p^{-2} + 1) e^{2A} \int_{r}^{r_0} A(\tilde{r}) \, d \tilde{r}.
    \]
    Thus applying Gr\"onwall's inequality, there is a uniform bound $A(r) \leq A_{\phi}$ for all $r \in (r_p, r_0]$. By the definition of $A(r)$, we therefore have $|\partial_r \phi(r, x)| \leq A_{\phi}$ and $|e^{\mu - \lambda} \partial_x \phi(r, x)| \leq A_{\phi}$.

    We now play the same game with \eqref{eq:mu_wave}. There are more terms arising in this equation, and for instance that the $e^{2(\mu - \lambda)} (\partial_x \phi)^2$ in \eqref{eq:mu_wave} needs to be controlled by the previous estimate $|e^{\mu - \lambda} \partial_x \phi| \leq A_{\phi}$. Omitting the details, one finds $A_{\mu}$ so that $|\partial_r \mu(r, x)| \leq A_{\mu}$ and $|e^{\mu - \lambda} \partial_x \mu(r, x)| \leq A_{\mu}$ for $r \in (r_p, r_0]$.

    It remains to estimate derivatives of $\lambda$, or equivalently derivatives of $\mu - \lambda$. The expression $\partial_r (\mu - \lambda)$ may be estimated immediately using \eqref{eq:mulambda_evol}. We also commute \eqref{eq:mulambda_evol} with $\partial_x$, yielding:
    \[
        r \partial_r \partial_x (\mu - \lambda) = \left( - \frac{Q^2}{r^2} + \kappa \right) e^{2\mu} (2 \partial_x \mu) = 2 \left( - \frac{Q^2}{r^2} + \kappa \right) e^{\mu + \lambda} \cdot e^{\mu - \lambda} \partial_x \mu.
    \]
    Having estimated $|e^{\mu - \lambda} \partial_x \mu| \leq A_{\mu}$, it suffices to find an upper bound for $e^{\mu + \lambda}$. One does this by using \eqref{eq:lambda_evol}--\eqref{eq:mu_evol} to show that $r \partial_r (\mu + \lambda) \geq 0$, and therefore $e^{\mu + \lambda}$ decreases as $r$ decreases towards $r_p$. So $r \partial_r \partial_x (\mu - \lambda)$ is uniformly bounded, and thus one finds $A_{\lambda}$ such that $|\partial_x (\mu - \lambda)| \leq A_{\lambda}$.

    The final step is that the previous $\mu$ and $\phi$ estimates yielded $|e^{\mu - \lambda} \partial_x \phi| \leq A_{\phi}$ and $|e^{\mu - \lambda} \partial_x \mu| \leq A_{\mu}$. To conclude, one must remove the $e^{\mu - \lambda}$ prefactors, which is easy since we have an estimate for $|\partial_r (\mu - \lambda)|$ and therefore for $|\mu - \lambda|$ itself. So have uniform bounds for $|\partial_x f|$ and $|\partial_r f|$ for each of $f = \phi, \mu, \lambda$. By standard continuation criteria for wave--transport systems, this allows us to extend the solution beyond $r = r_p$, yielding the required contradiction.
\end{proof}

%auto-ignore

\section{Precise statement of the main theorems} \label{sec:theorem}

We now state the fully detailed versions of our main theorems, the global existence result Theorem~\ref{thm:global}  and the bounce result Theorem~\ref{thm:bounce}, corresponding to our rough Theorem~\ref{thm:global_rough} and Theorem~\ref{thm:bounce_rough} respectively. Before stating our results we shall introduce the various parameters $\upeta, \upgamma, \upzeta, N, \ldots$ that appear throughout, and carefully define the class of initial data to which our results apply.

\subsection{Setup of the initial data} \label{sub:thm_data}

\subsubsection{Notation and key parameters} \label{subsub:param}

In the statement and proof of Theorem~\ref{thm:global} and Theorem~\ref{thm:bounce} we encounter a variety of parameters. Here, we list all of the key parameters and explain briefly their role in the article.

\begin{itemize}
    \item
        The real number $\upeta \geq 2$ is a constant that captures both the maximum and minimum allowed size of the quantity $r \partial_r \phi$ (see already \eqref{eq:global_linfty}). Though $\upeta$ may be chosen arbitrarily large, many of the subsequent parameters will depend on $\upeta$, and all quantitative bounds will depend on the choice of $\upeta$.
        
    \item
        The positive real number $\upgamma$ is a small parameter representing the admissible blow up rate for first-order derivatives, for instance we always have $\partial_x \phi = O(r^{- \upgamma})$. One can use any $\upgamma$ with $0 < \upgamma \leq \frac{1}{100}$.

    \item
        The natural number $N$ represents the largest number of $\partial_x$-derivatives with which we commute our system \eqref{eq:lambda_evol}--\eqref{eq:phi_wave} in deriving energy estimates. That is, our \emph{top-order} energy estimate will control $r \partial_r \partial_x^N f$ and $r e^{\mu - \lambda} \partial_x^{N+1} f$ for $f \in \{\phi, \mu, \lambda\}$. Our choice of $N$ will \underline{depend on $\upeta$} and we expect $N \to \infty$ as $\upeta \to \infty$. Also $K$ will represent an integer with $0 \leq K \leq N$.

    \item
        The real number $\upzeta > 0$ will represent the maximum admissible size of initial data, see already \eqref{eq:data_energy_boundedness}. $\upzeta$ may be arbitrarily large so long as it is compensated by taking $r_*$ small, see below.

    \item
        The real number $C_* > 0$ will be used in the bootstrap argument, see Section~\ref{sub:bootstrap}. The eventual choice of $C_*$ will depend on $\upeta$ and $\upzeta$, though we do not make this explicit. The real number $A_* > 0$ will depend on $C_*$ and will represent the rate at which our energies may blow up.

    \item
        The real number $r_* > 0$ represents how close we require our initial data to be to $r = 0$ in order to obtain our results. That is, our main results will apply to initial data given at $r = r_0$ so long as $0 < r_0 \leq r_*$. Note that $r_*$ will depend on $\upeta$ and $\upzeta$ and we expect $r_* \to 0$ as $\upeta, \upzeta \to \infty$.

    \item
        Other constants, often denoted $C$, will be allowed to depend on all of the aforementioned parameters. We also use $\updelta$ to represent quantities depending on all these parameters (e.g.~$\upeta, \upzeta, r_*$) such that $\updelta \to 0$ as $r_* \to 0$. We often abuse notation and write for instance ``$\updelta + \updelta = \updelta$'' or ``$C_* \updelta = \updelta$'' etc.
\end{itemize}

As per usual, the notation $F \lesssim G$ will mean that $F \leq C G$ where $C$ is allowed to depend on all the parameters above. In cases where we wish to specift that $C$ depends only on a subset of the parameters, say $\upeta$ only, then we write $F \lesssim_{\upeta} G$.

\subsubsection{Initial \texorpdfstring{$L^{\infty}$}{L\_infty} and \texorpdfstring{$L^2$}{L\_2} bounds}

As in Proposition~\ref{prop:lwp}, given $Q \in \R$, $\kappa \in \{ -1, 0, +1 \}$ and $r_0 > 0$, initial data will be given by $(f, r \partial_r f)|_{r = r_0} = (f_D, \dot{f}_D)$ for $f \in \{\phi, \lambda, \mu\}$, where the functions $\phi_D, \mu_D, \lambda_D, \dot{\phi}_D, \dot{\mu}_D, \dot{\lambda}_D$ satisfy certain constraints.

We aim to characterize a large class of initial data for which our main results apply. For some $N$ depending on $\upeta$, this will include initial data with $(f_D, \dot{f}_D) \in H^{N+1} \times H^N$ that obeys the following three conditions:
\begin{itemize}
    \item (\emph{Weak subcriticality}) For $x \in \mathbb{S}^1$:
        \begin{equation} \label{eq:data_linfty}
            \upeta^{-1} \leq \dot{\phi}_D(x) = r \partial_r \phi(r_0, x) \leq \upeta, \qquad \frac{Q^2}{r_0^2} e^{2 \mu_D}(x) = \frac{Q^2}{r_0^2} e^{2 \mu} (r_0, x) \leq 1,    
        \end{equation}

    \item (\emph{Closeness to singularity}) For $x \in \mathbb{S}^1$:
        \begin{equation} \label{eq:data_linfty2}
            e^{2\mu_D} (x) = e^{2 \mu}(r_0, x) \leq \upzeta r_0, \qquad e^{2(\mu_D - \lambda_D)}(x) = e^{2(\mu - \lambda)}(r_0, x) \leq \upzeta r_0,
        \end{equation}

    \item (\emph{Energy boundedness}) The following $L^2$ bound holds:
            \begin{equation} \label{eq:data_energy_boundedness} 
                \frac{1}{2} \sum_{f \in \{ \phi, \lambda, \mu \}} \sum_{K = 0}^N \left(  \int_{\mathbb{S}^1} ( \dot{f}_D^2 + r_0^2 e^{2(\mu_D - \lambda_D)} (\partial_x^{K+1} f_D)^2 + r_0^{2 \upgamma} (\partial_x^K f_D)^2 \right) \,dx \bigg|_{r = r_0} \leq \upzeta.
            \end{equation}
\end{itemize}
Furthermore, we will require $r_*$ to be chosen sufficiently small depending on $\upeta$ and $\upzeta$.

\subsection{Theorem~\ref{thm:global}: Global existence}

Theorem~\ref{thm:global} is the precise statement of our global existence result, see Theorem~\ref{thm:global_rough} for the rough version. In the statement, we also include some more quantitative $L^2$ and $L^{\infty}$ bounds.

\begin{theorem}[Global existence] \label{thm:global}
    Consider initial data $(\phi_D, \mu_D, \lambda_D, \dot{\phi}_D, \dot{\mu}_D, \dot{\lambda}_D)$ at $r = r_0$ as in Proposition~\ref{prop:lwp}. Then there exists $r_* = r_*(\upeta, \upzeta) > 0$ such that if $0 < r_0 \leq r_*$ and the conditions \eqref{eq:data_linfty}--\eqref{eq:data_energy_boundedness} are satisfied, then the maximal interval of existence $I$ for the solution $(\phi, \mu, \lambda)$ to the surface symmetric Einstein--Maxwell--scalar field system \eqref{eq:lambda_evol}--\eqref{eq:phi_wave} satisfies $(0, r_0] \subset I$.

    Furthermore, for $(r, x) \in (0, r_0] \times \mathbb{S}^1$ one has the following $L^{\infty}$ bounds:
    \begin{equation} \label{eq:global_linfty}
        (4 \upeta)^{-1} \leq r \partial_r \phi(r, x) \leq 4 \upeta, \qquad \frac{Q^2}{r^2} e^{2\mu}(r, x) \leq 16 \upeta^2,
    \end{equation}
    and for $0 \leq K \leq N = N(\upeta)$, there exists $C = C(\upeta, \upzeta, K) > 0$ and $A_* = A_*(\upeta) > 0$ such that one has the $L^2$ energy bounds: for fixed $r > 0$,
    \begin{equation} \label{eq:global_energy}
        \sum_{f \in \{ \phi, \lambda, \mu \}} \left(  \int (r \partial_r \partial_x^K f)^2 + (r e^{\mu - \lambda}  \partial_x^{K+1} f)^2 + r^{-2 \upgamma}( \partial_x^K f)^2 \, dx\right) \leq C r^{- 2 A_* - 2 K \upgamma}.
    \end{equation}
\end{theorem}

\begin{remark}
    Theorem~\ref{thm:global} is independent of the choice of sectional curvature $\kappa \in \{-1, 0, +1\}$, as well as the choice of $Q \in \R$, so long as the second equation in \eqref{eq:data_linfty} is satisfied. So Theorem~\ref{thm:global} applies equally to the surface symmetric Einstein--scalar field system and to the surface symmetric Einstein--Maxwell--scalar field system. We make distinctions between these two models in Theorem~\ref{thm:bounce}.
%
%    We remark that to prove global existence towards $r = 0$, at least in the $Q \neq 0$ case, it suffices to show the second inequality in \eqref{eq:global_linfty}, since this implies $\mu \leq \log (4 \upeta r_0 Q^{-1})$, and one may then apply the continuation criterion in Lemma~\ref{lem:continuation}. However, our approach towards global existence uses a bootstrap argument plus energy estimates; indeed global existence follows since \eqref{eq:global_energy} implies that the $N$th order energy remains uniformly bounded in the region $r \geq r_b$ for $r_b > 0$ some ``bootstrap time'', and one has local existence in the space $X = H^{N+1} \times H^N$ in Proposition~\ref{prop:lwp}.
\end{remark}

\subsection{Theorem~\ref{thm:bounce}: BKL bounces}

Theorem~\ref{thm:bounce} is the precise statement of the bounce result, showing that certain quantities obey nonlinear ODEs (plus error terms) along timelike curves. See Theorem~\ref{thm:bounce_rough} for the rough version.

\begin{theorem}[Bounces] \label{thm:bounce}
    Let $(\phi, \mu, \lambda)$ be a solution to the surface symmetric Einstein--scalar field system \eqref{eq:lambda_evol}--\eqref{eq:phi_wave}, arising from initial data $(\phi_D, \mu_D, \lambda_D, \dot{\phi}_D, \dot{\mu}_D, \dot{\lambda}_D)$ obeying \eqref{eq:data_linfty}--\eqref{eq:data_energy_boundedness} and $0 < r_0 \leq r_*$ as in Theorem~\ref{thm:global}.

    Let $\gamma: (0, r_0] \to \mathcal{Q}$ be any timelike curve parameterized by $r$, and let $\mathscr{P}_{\gamma}(r) = r \partial_r \phi (\gamma(r))$ and $\mathscr{Q}_{\gamma}(r) = \frac{Q^2}{r^2}e^{2\mu} (\gamma(r))$, note $\mathscr{Q} \equiv 0$ if and only if $Q = 0$. Then there exist error terms $\mathscr{E}_{\mathscr{P}}(r), \mathscr{E}_{\mathscr{Q}}(r)$ depending on $\gamma$ but with $|\mathscr{E}_{\mathscr{P}}(r)|, |\mathscr{E}_{\mathscr{Q}}(r)| \leq r^{1/2}$ uniformly in the choice of $\gamma$, such that
    \begin{equation} \label{eq:bounce_ode}
        r \frac{d}{dr} \mathscr{P}_{\gamma} = - \mathscr{P}_{\gamma} \mathscr{Q}_{\gamma} + \mathscr{E}_{\mathscr{P}}, \qquad r \frac{d}{dr} \mathscr{Q}_{\gamma} = \mathscr{Q}_{\gamma} ( \mathscr{P}_{\gamma}^2 - 1 - \mathscr{Q}_{\gamma} - \mathscr{E}_{\mathscr{Q}} ).
    \end{equation}
    Furthermore, there exists $C = C(\upeta, \upzeta) > 0$ such that:
    \begin{enumerate}[(i)]
        \item \label{item:bounce_i}
            If $Q = 0$, then $\mathscr{P}_{\gamma}(r)$ converges to $\mathscr{P}_{\gamma, \infty}$ as $r \to 0$ and $|\mathscr{P}_{\gamma, \infty} - \mathscr{P}_{\gamma}(r_0)| \leq C r_0^{1/2}$.

        \item \label{item:bounce_ii}
            If $Q \neq 0$, then $\mathscr{Q}_{\gamma}(r)$ converges to $0$ as $r \to 0$, while $\mathscr{P}_{\gamma}(r)$ converges to $\mathscr{P}_{\gamma, \infty}$ as $r \to 0$, necessarily satisfying $\mathscr{P}_{\gamma, \infty} \geq 1$ as well as:
            \begin{equation} \label{eq:bounce_asymp}
                | \mathscr{P}_{\gamma, \infty} - \max \{ \mathscr{P}_{\gamma}(r_0), \mathscr{P}_{\gamma}(r_0)^{-1} \} | \leq C ( \mathscr{Q}_{\gamma}(r_0) + r_0^{1/2} ).
            \end{equation}
    \end{enumerate}
\end{theorem}

A corollary of Theorem~\ref{thm:bounce} is the following stability and instability statement regarding \emph{charged} perturbations of a class of \emph{uncharged} solutions to the surface symmetric Einstein--scalar field system. We assume that $\kappa \in \{0, +1\}$ so as to be able to apply Theorem~\ref{thm:asymp_smooth} for the unperturbed uncharged spacetime.

\begin{corollary}[Stability / Instability] \label{cor:bounce}
    Let $(\phi, \mu, \lambda)$ be a \underline{smooth} (i.e.~$C^{\infty}$) solution to the surface symmetric Einstein--scalar field system \eqref{eq:lambda_evol}--\eqref{eq:phi_wave} with $Q = 0$ and $\kappa \in \{0, +1\}$, arising from initial data given by $(\phi_D, \mu_D, \lambda_D, \dot{\phi}_D, \dot{\mu}_D, \dot{\lambda}_D)$ at $r = r_1 > 0$. By Theorem~\ref{thm:asymp_smooth}, there exists smooth $\Psi, \Xi: \mathbb{S}^1 \to \R$ such that $\phi(r, x) = \Psi(x) \log r + \Xi(x) + o(1)$ as $r \to 0$.

    Suppose further that $\Psi(x) > 0$ for all $x \in \mathbb{S}^1$. Then there exists $N \in \mathbb{N}$ such that for possibly charged perturbations of the initial data, i.e.~$|Q| \leq \varepsilon$, and $\|(\tilde{f}_D, \tilde{\dot{f}}_D) - (f_D, \dot{f}_D) \|_{H^{N+1}\times H^N} \leq \varepsilon$ for $f \in \{\phi, \mu, \lambda\}$, then for $\varepsilon$ sufficiently small the solution $(\tilde{\phi}, \tilde{\mu}, \tilde{\lambda})$ arising from data $(\tilde{\phi}_D, \tilde{\mu}_D, \tilde{\lambda}_D, \tilde{\dot{\phi}}_D, \tilde{\dot{\mu}}_D, \tilde{\dot{\lambda}}_D)$ and charge $Q$ has (a) global existence towards $r = 0$, (b) is such that for any $x \in \mathbb{S}^1$ the limit $\tilde{\Psi}(x) = \lim_{r \to 0} r \partial_r \phi(r, x)$ exists and moreover for any $\tilde{\varepsilon} > 0$, if $\varepsilon$ is chosen small enough depending on $\tilde{\varepsilon}$ then:
    \begin{enumerate}[(i)]
        \item $|\tilde{\Psi}(x) - \Psi(x)| \leq \tilde{\varepsilon}$ if $Q = 0$,
        \item $|\tilde{\Psi}(x) - \max \{ \Psi(x), \Psi(x)^{-1} \}| \leq \tilde{\varepsilon}$ if $Q \neq 0$.
    \end{enumerate}
\end{corollary}

\begin{remark}
    We consider Corollary~\ref{cor:bounce} a stability / instability result because while stability holds in the sense that while for our perturbations we still have global existence towards $r = 0$ and preservation of features such as matter and curvature blow-up (i.e.~stability), if in the case where the original $\Psi(x)$ has $0 < \Psi(x_0) < 1$ for some $x_0 \in \mathbb{S}^1$, any perturbed spacetime with $Q \neq 0$ will have a corresponding $\tilde{\Psi}(x)$ with $\tilde{\Psi}(x_0) \approx \Psi(x_0)^{-1}$, which is ``far away'' from $\Psi(x_0)$.

    We do not include cases where $\Psi(x)$ can be $0$ in the unperturbed spacetime. In fact, in this context there are examples where one does even have stability in the sense of global existence towards $r = 0$, e.g.~the perturbation of Schwarzschild to Reissner-Nordstr\"om. It remains an open problem to study cases where $\Psi(x)$ is allowed to change sign; this is outside the scope of this article, given our dependence on $\upeta \geq 2$ and the weak subcriticality condition \eqref{eq:data_linfty}.
\end{remark}

%auto-ignore

\section{Bootstrap assumptions, energies and interpolation lemmas} \label{sec:energy}

Theorems~\ref{thm:global} and \ref{thm:bounce} will be proven simultaneously. As outlined at the end of Section~\ref{sub:intro_proof}, the proof will proceed via a bootstrap argument. That is, we suppose that a solution $(\phi, \mu, \lambda)$ to the surface symmetric Einstein--Maxwell--scalar field system \eqref{eq:lambda_evol}--\eqref{eq:phi_wave} exists in the interval $r \in [r_b, r_0]$, where $r_b > 0$ is a ``bootstrap time''.

We moreover suppose that the solution satisfies certain pointwise bootstrap assumptions in this interval $[r_b, r_0]$, namely \eqref{eq:bootstrap_maxwell}--\eqref{eq:bootstrap_d2x} below. The strategy is to then use these bootstrap assumptions to derive energy estimates, where energies are suitable $L^2$ type integrals to be defined in Section~\ref{sub:energy}. Though the energy estimates will blow up mildly towards $r = 0$, they will be used together with an ODE-based argument to eventually improve our $L^{\infty}$ bootstrap assumptions, completing the proof of Theorem~\ref{thm:global}. Further analysis of the ODEs will yield Theorem~\ref{thm:bounce}.

\subsection{The \texorpdfstring{$L^{\infty}$}{L\_infty} bootstrap assumptions} \label{sub:bootstrap}

As mentioned in Section~\ref{subsub:param}, $C_* = C_*(\upeta, \upzeta) > 0$ will represent a large real number to be chosen later in the argument. In the proof of Theorem~\ref{thm:global}, we make repeated reference to the following four low order $L^{\infty}$ bootstrap assumptions:
\begin{gather}
    \label{eq:bootstrap_maxwell} \tag{B1}
    \frac{Q^2}{r^2} e^{2 \mu} \leq C_*, \quad e^{2 \mu} + e^{2(\mu - \lambda)} \leq C_* r, \\[0.8em]
    \label{eq:bootstrap_rdr} \tag{B2}
    |r \partial_r \phi|, \, |r \partial_r \mu|, \, |r \partial_r \lambda| \leq C_*, \\[0.8em]
    \label{eq:bootstrap_dx} \tag{B3}
    |\partial_x \phi|, \, |\partial_x \mu|, \, |\partial_x \lambda| \leq C_* r^{- \upgamma}, \\[0.8em]
    \label{eq:bootstrap_d2x} \tag{B4}
    |r \partial_r \partial_x \phi|, \, |r \partial_r \partial_x \mu|, \, |r \partial_r \partial_x \lambda| \leq C_* r^{- \upgamma}.
\end{gather}
Using the assumptions on the initial data \eqref{eq:data_linfty}--\eqref{eq:data_energy_boundedness} and Sobolev embedding, one can choose $C_*$ depending on $\upzeta$ such that the \eqref{eq:bootstrap_maxwell}--\eqref{eq:bootstrap_d2x} will hold in a neighborhood of $r = r_0$. Thus we can always find some $r_b$ so that we can carry out our bootstrap argument in the interval $r \in [r_b, r_0]$.

\subsection{Energies} \label{sub:energy}

Below, we define the $L^2$ type quantities that will appear in our energy estimates of Section~\ref{sec:l2}. Since eventually have a hierarchy of energy estimates (see already Proposition~\ref{prop:energy_hierarchy}), we define a separate energy at each order of differentiability, up to a top-order energy with $N$ derivatives, with $N$ eventually chosen depending on $\upeta$.

\begin{definition} \label{def:energy}
    Let $0 \leq K \leq N$. Define the following $K$th order energies at fixed $r$:
    \begin{gather}
        \mathcal{E}^{(K)}_{\phi}(r) \coloneqq \frac{1}{2} \int \left( (r \partial_r \partial_x^K \phi)^2 + r^2 e^{2(\mu - \lambda)} (\partial_x^{K+1} \phi)^2 + (\partial_x^K \phi)^2 \right) dx, \\[0.8em]
        \mathcal{E}^{K)}_{\mu}(r) \coloneqq \frac{1}{2} \int \left( (r \partial_r \partial_x^K \mu)^2 + r^2 e^{2(\mu - \lambda)} (\partial_x^{K+1} \mu)^2 + (\partial_x^K \mu)^2 \right) dx, \\[0.8em]
        \mathcal{E}^{(K)}_{\lambda}(r) \coloneqq \frac{1}{2} \int \left( (r \partial_r \partial_x^K \lambda)^2 + r^2 e^{2(\mu - \lambda)} (\partial_x^{K+1} \lambda)^2 + (\partial_x^K \lambda)^2 \right) dx, \\[0.8em]
        \mathcal{E}^{(K)}(r) \coloneqq \mathcal{E}^{(K)}_{\phi}(r) + \mathcal{E}^{(K)}_{\mu}(r) + \mathcal{E}^{(k)}_{\lambda}(r).
    \end{gather}
\end{definition}

\begin{remark}
    Even at order $K=0$, the energy $\mathcal{E}_{\phi}^{(0)}(r)$ includes the lower order term $\phi^2$. Even when that Theorem~\ref{thm:asymp_smooth} applies, this will blow up as $O( (\log r)^2 )$. However, since we allow our energies to blow up towards $r = 0$, this $O((\log r)^2)$ blow up will not be of any concern. It will be useful to include the $(\partial_x^K \phi)^2$ term in $\mathcal{E}_{\phi}^{(K)}(r)$ since this is how we will eventually get pointwise control on $\partial_x$-derivatives of $\phi$ without the problematic weight $r e^{\mu - \lambda}$ which vanishes as $r \to 0$.
\end{remark}

\subsection{Sobolev--type inequalities}

We finish this section with two technical lemmas that will allow us to relate $L^2$ norms and $L^{\infty}$ norms of certain quantities. These will be useful in going from our low order $L^{\infty}$-norms to energy estimates, and vice versa.

\begin{lemma}[Sobolev interpolation inequality] \label{lem:interpolation}
    Let $N, K$ be integers with $0 \leq N < K$, and let $f: \mathbb{S}^1 \to \R$ be such that $\partial_x^K f \in L^2(\mathbb{S}^1)$. Then the following $L^{\infty}$--$L^2$ interpolation inequality holds:
    \begin{equation}
        \| \partial_{x}^N f \|_{L^{\infty}(\mathbb{S}^1)} \lesssim_{N, K} \| f \|_{L^{\infty}(\mathbb{S}^1)}^{1 - \alpha} \| \partial_x^K f \|_{L^2(\mathbb{S}^1)}^{\alpha},
        \qquad
        \text{ where } \alpha = \frac{N}{K - \frac{1}{2}}.
    \end{equation}
\end{lemma}

\begin{proof}
    This Gagliardo--Nirenberg type inequality is standard, see for instance \cite[Lecture II]{NirenbergElliptic}.
\end{proof}

\begin{lemma}[Weighted Sobolev product inequality] \label{lem:weightedl2}
    Let $M, N \geq 0$ be integers and let $K = M + N$. Let $w: \mathbb{S}^1 \to \R_{> 0}$ be a weight function and let $W = \log w$. Then for $f, g$ sufficiently regular one has
    \begin{multline}
        \| w \, \partial_x^M f \, \partial_x^N g \|_{L^2(\mathbb{S}^1)}
        \lesssim_{M, N} 
        \| f \|_{L^{\infty}(\mathbb{S}^1) }
        \left(
            \| w \, \partial_x^K g \|_{L^2(\mathbb{S}^1)} + \sum_{j=1}^{K-1} \| \partial_x W \|_{L^{\infty}(\mathbb{S}^1)}^j \| w \, \partial_x^{K - j} g \|_{L^2(\mathbb{S}^1)}
        \right)
        \\[-0.2em]
        + 
        \| g \|_{L^{\infty}(\mathbb{S}^1) }
        \left(
            \| w \, \partial_x^K f \|_{L^2(\mathbb{S}^1)} + \sum_{j=1}^{K-1} \| \partial_x W \|_{L^{\infty}(\mathbb{S}^1)}^j \| w \, \partial_x^{K - j} f \|_{L^2(\mathbb{S}^1)}
        \right).
    \end{multline}
\end{lemma}

\begin{proof}
    This is a generalization of \cite[Lemma 6.16]{RingstromCauchy} where $\R^n$ is replaced by the unit circle $\mathbb{S}^1$ and we also introduce a weight function $w$. A complete proof of this lemma is provided in Appendix~\ref{app:weightedl2}.
\end{proof}

%auto-ignore

\section{The energy estimate hierarchy} \label{sec:l2}

We now derive energy estimates for $\mathcal{E}^{(K)}(r)$, at orders $0 \leq K \leq N$, where $N = N(\upeta)$ is chosen sufficiently large. The derivation of such energy estimates will involve commuting the equations \eqref{eq:lambda_evol}--\eqref{eq:phi_wave} with up to $K$ $\partial_x$-derivatives. See Section~\ref{sub:intro_proof} for an overview.

For the junk and lower order terms in the hierarchy (where the precise value of coefficients arising in the commuted equations are not crucial), we introduce the following schematic notation: expressions such as
\[
    \sum_{k_p + k_1 + \ldots + k_i = K} \partial_x^{k_p} f * \partial_x^{k_1} g * \cdots * \partial_x^{k_i} g
\]
will represent some linear combination of products of the form $\partial_x^{k_p} f \cdot \partial_x^{k_1} g \cdots \partial_x^{k_i} g$ such that $i \geq 1$, $k_p \geq 1$ and $k_j \geq 1$ for all $1 \leq j \leq i$ and $k_p + k_1 + \ldots + k_i = K$. We emphasize that \underline{unless explicitly stated otherwise}, in these schematic sums $i$ and $j$ will be positive integers, as are the indices $k_p, k_1, k_i$ etc. 
In the event that any index e.g.~$k_p$ is allowed to be $0$, this will be explicitly stated, and similarly if there are further constraints on any index.

\subsection{Energy estimates for the scalar field \texorpdfstring{$\phi$}{ϕ}} \label{sec:l2_phi}

\begin{proposition} \label{prop:phi_energy}
    Let $(\mu, \lambda, \phi, Q)$ be a solution to the Einstein--Maxwell--scalar field system \eqref{eq:lambda_evol}--\eqref{eq:phi_wave} obeying the bootstrap assumptions \eqref{eq:bootstrap_maxwell}--\eqref{eq:bootstrap_d2x}. Then there exists a constant $A_*$ depending only on $C_*$, as well as constants $C_1^{(K)}$ and $C_2^{(K)}$ depending on $C_*$ and the regularity index $K \in \{0, 1, \ldots, N \}$, such that
    \begin{equation} \label{eq:energy_der_phi}
        \left| r \frac{d}{dr} \mathcal{E}_{\phi}^{(K)} (r) \right| \leq
        2 A_* \mathcal{E}^{(K)}(r) + C^{(K)}_1 r^{1 - \upgamma} \mathcal{E}^{(K)}(r) + \sum_{k=0}^{K-1} C^{(K)}_2 r^{- 2\upgamma (K - k)} \mathcal{E}^{(k)}(r),
    \end{equation}
    where it is understood that the final term is absent if $K=0$.
\end{proposition}

\begin{proof}
    For any $K \geq 1$, commuting the wave equation \eqref{eq:phi_wave} with $\partial_x^K$ yields
    \begin{align}
        \addtocounter{equation}{1}
        (r \partial_r)^2 \partial_x^K \phi - r^2 e^{2(\mu - \lambda)} \partial_x^2 \partial_x^K \phi 
        &= \left( - \frac{Q^2}{r^2} + \kappa \right) e^{2 \mu} \, (r \partial_r \partial_x^K \phi + 2 r \partial_r \phi \, \partial_x^K \mu )
        \tag{\theequation a} \label{eq:phi_wave_K_a} \\[0.6em]
        &\mkern-36mu + r^2 e^{2 (\mu - \lambda)} \sum_{\substack{k_p + k_1 + \cdots + k_i = K + 2}} \partial_x^{k_p} \phi * \partial_x^{k_1} (\mu - \lambda) * \cdots * \partial_x^{k_i} (\mu - \lambda)
        \tag{\theequation b} \label{eq:phi_wave_K_b} \\[0.6em]
        &\mkern-36mu + \left( - \frac{Q^2}{r^2} + \kappa \right) e^{2\mu} \sum_{\substack{0 \leq k_p < K, k_1 < K \\ k_p + k_1 + \cdots + k_i = K}} r \partial_r \partial_x^{k_p} \phi * \partial_x^{k_1} \mu * \cdots * \partial_x^{k_i} \mu.
        \tag{\theequation c} \label{eq:phi_wave_K_c}
    \end{align}
    Note that in the uncommuted case $K = 0$, the final term on the RHS of \eqref{eq:phi_wave_K_a} is absent, as is the term \eqref{eq:phi_wave_K_c}; we leave the derivation of \eqref{eq:energy_der_phi} in this case to the reader.

    The first line \eqref{eq:phi_wave_K_a} is the leading order contribution that gives rise to the $2 A_* \mathcal{E}^{(K)}(r)$ on the RHS of \eqref{eq:energy_der_phi}. It remains to estimate the remaining lines \eqref{eq:phi_wave_K_b} and \eqref{eq:phi_wave_K_c} in $L^2$. We start with \eqref{eq:phi_wave_K_c}. Using the bootstrap assumptions \eqref{eq:bootstrap_maxwell}--\eqref{eq:bootstrap_rdr} and the Sobolev product estimate Lemma~\ref{lem:weightedl2} with $w=1$, one finds that
    \begin{align*}
        \| \text{\eqref{eq:phi_wave_K_c}} \|_{L^2} 
        &\leq C_*^2 \, 
        \bigg \|  \sum_{\substack{0 \leq k_p < K, k_1 < K \\ k_p + k_1 + \cdots + k_i = K}} r \partial_r \partial_x^{k_p} \phi * \partial_x^{k_1} \mu * \cdots * \partial_x^{k_i} \mu \bigg \|_{L^2} \\[0.4em]
        &\lesssim_{C_*}
        \sum_{k=0}^{K-1} \left( \| r \partial_r \partial_x^k \phi \|_{L^2} + \| \partial_x^k \mu \|_{L^2} \right) \cdot \left( \| r \partial_r \partial_x \phi \|_{L^{\infty}} + \| \partial_x \mu \|_{L^{\infty}} \right)^{K-k}.
    \end{align*}
    As this is the first time we derive such an estimate, we include a more detailed explanation. In the first line, we use \eqref{eq:bootstrap_maxwell} to estimate first $\left| - \frac{Q^2}{r^2} e^{2\mu} + \kappa e^{2\mu} \right| \leq C_*$, and if $k_p = 0$ we get another $C_*$ from $r \partial_r \phi$ and \eqref{eq:bootstrap_rdr}.

    To get the second line, for summands such that there are $i+1$ terms in the product, note that the maximum number of $\partial_x$-derivatives landing on either $r \partial_r \phi$ or $\mu$ is $K - i$. Furthermore, (repeated) use of the product estimate Lemma~\ref{lem:weightedl2} allows us to put exactly this many $\partial_x$-derivatives on some term in the product, which we estimate in $L^2$, while the remaining $i$ terms in the product are estimated by either $\| \partial_x (r \partial_r \phi) \|_{L^{\infty}}$ or $\| \partial_x \mu \|_{L^{\infty}}$. By setting $k = K - i \in \{1, \ldots, k-1\}$, we obtain the desired estimate.

    By using the remaining bootstrap assumptions \eqref{eq:bootstrap_dx}--\eqref{eq:bootstrap_d2x} to estimate $ \| r \partial_r \partial_x \phi \|_{L^{\infty}} + \| \partial_x \mu \|_{L^{\infty}}$, as well as the definition of $\mathcal{E}^{(k)}(r)$, we therefore deduce that
    \begin{equation} \label{eq:phi_wave_K_c_est}
        \| \text{\eqref{eq:phi_wave_K_c}} \|_{L^2} \lesssim_{C_*, K} \sum_{k=0}^{K-1} r^{- \upgamma (K - k)} \sqrt{ \mathcal{E}^{(k)}(r) }.
    \end{equation}

    We move onto the $L^2$ estimate for \eqref{eq:phi_wave_K_b}. For this term we first use \eqref{eq:bootstrap_maxwell} to remove one of the factors of $e^{\mu - \lambda}$, but for the remaining factor of $e^{\mu - \lambda}$ we use Lemma~\ref{lem:weightedl2} with $w = r e^{\mu - \lambda}$. In particular, $\partial_x W = \partial_x \log w = \partial_x (\mu - \lambda)$, and one gets
    \begin{align*}
        \| \text{\eqref{eq:phi_wave_K_b}} \|_{L^2} 
        &\leq r C_* \, 
        \bigg \| r e^{\mu - \lambda} \sum_{\substack{k_p + k_1 + \cdots + k_i = K + 2}} \partial_x^{k_p} \phi * \partial_x^{k_1} (\mu - \lambda) * \cdots * \partial_x^{k_i} (\mu - \lambda) \bigg \|_{L^2} \\[0.6em]
        &\lesssim_{C_*}
        r \sum_{k=1}^{K+1} \left( \| r e^{\mu - \lambda} \partial_x^k \phi \|_{L^2} + \| r e^{\mu - \lambda} \partial_x^k (\mu - \lambda) \|_{L^2} \right) \cdot \left( \| \partial_x \phi \|_{L^{\infty}} + \| \partial_x (\mu - \lambda) \|_{L^{\infty}} \right)^{K+2-k}.
    \end{align*}
    It is important we applied Lemma~\ref{lem:weightedl2} in such a way that the top order term $\partial_x^{K+1} \phi$ has the weight $r e^{\mu - \lambda}$ in front to match with the energy $\mathcal{E}^{(K)}(r)$. Fortunately, the weighted version of Lemma~\ref{lem:weightedl2} means that any additional factors of $\| \partial_x(\mu - \lambda) \|_{L^{\infty}}$ reduce the number of derivatives on either $r e^{\mu - \lambda} \partial_x^k \phi$ or $r e^{\mu - \lambda} \partial_x^k (\mu - \lambda)$, and we obtain the above.

    Using the bootstrap assumption \eqref{eq:bootstrap_dx} and the definition of $\mathcal{E}^{(k)}(r)$, we thus deduce
    \begin{equation} \label{eq:phi_wave_K_b_est}
        \| \text{\eqref{eq:phi_wave_K_b}} \|_{L^2} \lesssim_{C_*, K} \sum_{k=0}^{K} r^{1 - \upgamma (K + 1 - k)} \sqrt{ \mathcal{E}^{(k)}(r) }.
    \end{equation}
    The additional power of $r$ means that the top order term $\sqrt{\mathcal{E}^{(K)}(r)}$ appears with a weight that is suitably integrable as $r \downarrow 0$, and thus does not contribute to the $A_* \mathcal{E}^{(K)}(r)$ on the RHS of \eqref{eq:energy_der_phi}.

    Now, to conclude, for $0 \leq K \leq N$ we write
    \begin{align*}
        r \partial_r \left( \frac{1}{2} \left( (r \partial_r \partial_x^K \phi)^2 + r^2 e^{2(\mu-\lambda) } (\partial_x^{K+1} \phi)^2 + (\partial_x^K \phi)^2 \right) \right) 
        & \\[0.6em]
        & \mkern-360mu = (r \partial_r \partial_x^K \phi) \left[ (r \partial_r)^2 \partial_x^K \phi - r^2 e^{2(\lambda - \mu)} \partial_x^{K+2} \phi + \partial_x^K \phi \right] + r^2 e^{2(\mu - \lambda)} (\partial_x^{K+1} \phi)^2 \\[0.6em]
        & \mkern-300mu + \partial_x \left( r^2 e^{2(\mu - \lambda)} r \partial_r \partial_x^K \phi \cdot \partial_x^{K+1} \phi \right) - 2 r^2 e^{2(\mu - \lambda)} \partial_x(\mu-\lambda) \cdot r \partial_r \partial_x^K \phi \cdot \partial_x^{K+1} \phi.
    \end{align*}

    Integrating over $x \in \mathbb{S}^1$ so that the first term in the last line vanishes, and inserting the commuted wave equation, one has the following identity for the $r \frac{d}{dr}$ derivative of the energy:
    \begin{align*}
        r \frac{d}{dr} \mathcal{E}^{(K)}_{\phi}(r) 
        &= \int_{\mathbb{S}^1} \left[ (r \partial_r  \partial_x^{K} \phi) \left( \text{\eqref{eq:phi_wave_K_a}} + \partial_x^K \phi \right) + r^2 e^{2(\mu - \lambda)} (\partial_x^{K+1} \phi)^2 \right] \, dx \\[0.6em]
        &\quad + \int_{\mathbb{S}^1} \left[ (r \partial_r \partial_x^K \phi) \left( \text{\eqref{eq:phi_wave_K_b}} + \text{\eqref{eq:phi_wave_K_c}} - 2 r^2 e^{2(\mu - \lambda)} \partial_x (\mu - \lambda) \partial_x^{K+1} \phi \right) \right] \, dx.
    \end{align*}
    From the bootstrap assumptions \eqref{eq:bootstrap_maxwell}--\eqref{eq:bootstrap_rdr}, one may estimate $\| \text{\eqref{eq:phi_wave_K_a}} \|_{L^2} \leq 10 C_*^2 \sqrt{2 \mathcal{E}^{(K)}(r)}$. Using Cauchy--Schwarz, the first integral in the above expression can thus be bounded as:
    \[
        \left| \int_{\mathbb{S}^1} \left[ (r \partial_r  \partial_x^{K} \phi) \left( \text{\eqref{eq:phi_wave_K_a}} + \partial_x^K \phi \right) + r^2 e^{2(\mu - \lambda)} (\partial_x^{K+1} \phi)^2 \right] \, dx \right|
        \leq 2 (10 C_*^2 + 9) \mathcal{E}^{(K)}(r).
    \]

    On the other hand, the latter integral can be estimated using \eqref{eq:phi_wave_K_b_est} and \eqref{eq:phi_wave_K_c_est} -- the expression involving $r^2 e^{2(\mu - \lambda)} \partial_x(\mu - \lambda)$ can be bounded in the same way as \eqref{eq:phi_wave_K_b} -- one yields that for some $C_1^{(K)}, C_2^{(K)} > 0$,
    \begin{multline*}
        \left | \int_{\mathbb{S}^1} \left[ (r \partial_r \partial_x^K \phi) \left( \text{\eqref{eq:phi_wave_K_b}} + \text{\eqref{eq:phi_wave_K_c}} - 2 r^2 e^{2(\mu - \lambda)} \partial_x (\mu - \lambda) \partial_x^{K+1} \phi \right) \right] \, dx \right| \\[0.6em]
        \leq \left( C_1^{(K)} r^{1 - \upgamma} \sqrt{\mathcal{E}^{(K)}(r)} + \sqrt{C_2^{(K)}} \sum_{k=0}^{K-1} r^{-\upgamma(K - k)} \sqrt{\mathcal{E}^{(k)}(r)} \right) \cdot \sqrt{\mathcal{E}^{(K)}(r)}.
    \end{multline*}
    Combining these and applying Young's inequality, Proposition~\ref{prop:phi_energy} follows, with $A_* = 10 C_*^2 + 10$.
\end{proof}

\subsection{Energy estimates for the metric variables \texorpdfstring{$\mu$}{μ} and \texorpdfstring{$\lambda$}{λ}}

\begin{proposition} \label{prop:mu_energy}
    Let $(\mu, \lambda, \phi, Q)$ be a solution to the Einstein--Maxwell--scalar field system \eqref{eq:lambda_evol}--\eqref{eq:phi_wave} obeying the bootstrap assumptions \eqref{eq:bootstrap_maxwell}--\eqref{eq:bootstrap_d2x}. Then there exists a constant $A_*$ depending only on $C_*$, as well as constants $C_1^{(K)}$ and $C_2^{(K)}$ depending on $C_*$ and the regularity index $K \in \{0, 1, \ldots, N \}$, such that
    \begin{equation} \label{eq:energy_der_mu}
        \left| r \frac{d}{dr} \mathcal{E}_{\mu}^{(K)} (r) \right| \leq
        2 A_* \mathcal{E}^{(K)}(r) + C^{(K)}_1 r^{1 - \upgamma} \mathcal{E}^{(K)}(r) + 
        \begin{cases}
            \sum_{k=0}^{K-1} C^{(K)}_2 r^{- 2 \upgamma (K - k)} \mathcal{E}^{(k)}(r), & \text{ if } K \neq 0, \\
            2 A_* & \text{ if } K = 0.
        \end{cases}
    \end{equation}
\end{proposition}

\begin{proof}
    We consider the equation \eqref{eq:mu_wave} as a wave equation for $\mu$. As in the proof of Proposition~\ref{prop:phi_energy}, we first commute \eqref{eq:mu_wave} with $\partial_x^K$. Instead of doing this directly, we first rewrite \eqref{eq:mu_wave} as:
    \begin{equation*}
        \begin{split}
            (r \partial_r)^2 \mu - r^2 e^{2 (\mu- \lambda)} \partial_x^2 \mu = 
            &\left( -\frac{Q^2}{r^2} + \kappa \right) e^{2 \mu} r \partial_r (\lambda + 2 \mu) + r^2 e^{2 (\mu - \lambda)} \partial_x \mu \, \partial_x (\mu - \lambda) + 2 r^2 e^{2 (\mu- \lambda)} (\partial_x \phi)^2 + \frac{Q^2}{r^2} e^{2 \mu}, 
        \end{split}
    \end{equation*}

    Now commuting this equation with $\partial_x^K$, one yields for $1 \leq K \leq N$:
    \begin{align}
        (r \partial_r)^2 \partial_x^K \mu - r^2 e^{2(\mu - \lambda)} \partial_x^2 \partial_x^K \mu  \nonumber \addtocounter{equation}{1} \\[0.6em]
        &\mkern-60mu = \left( - \frac{Q^2}{r^2} + \kappa \right) e^{2 \mu} (r \partial_r \partial_x^K (\lambda + 2 \mu) + 2 r \partial_r (\lambda + 2 \mu) \, \partial_x^K \mu ) + \frac{Q^2}{r^2} e^{2\mu} \, \partial_x^K \mu
        \tag{\theequation a} \label{eq:mu_wave_K_a} \\[0.6em]
        &\mkern-36mu + r^2 e^{2 (\mu - \lambda)} \sum_{\substack{k_p + k_1 + \cdots + k_i = K + 2}} \partial_x^{k_p} \mu * \partial_x^{k_1} (\mu - \lambda) * \cdots * \partial_x^{k_i} (\mu - \lambda)
        \tag{\theequation b} \label{eq:mu_wave_K_b} \\[0.6em]
        &\mkern-36mu + r^2 e^{2 (\mu - \lambda)} \sum_{\substack{k_p + k_q + k_1 + \cdots + k_i = K + 2}} \partial_x^{k_p} \phi * \partial_x^{k_q} \phi * \partial_x^{k_1} (\mu - \lambda) \cdots * \partial_x^{k_i} (\mu - \lambda)
        \tag{\theequation c} \label{eq:mu_wave_K_c} \\[0.6em]
        &\mkern-36mu + \left( - \frac{Q^2}{r^2} + \kappa \right) e^{2\mu} \sum_{\substack{k_p \geq 0, k_1 < K \\ k_p + k_1 + \cdots + k_i = K}} r \partial_r \partial_x^{k_p} (\lambda + 2 \mu) * \partial_x^{k_1} \mu * \cdots * \partial_x^{k_i} \mu
        \tag{\theequation d} \label{eq:mu_wave_K_d} \\[0.6em]
        &\mkern-36mu + \frac{Q^2}{r^2} e^{2\mu} \sum_{\substack{1 \leq k_1 < K \\ k_1 + \cdots k_i = K}} \partial_x^{k_1} \mu * \cdots * \partial_x^{k_i} \mu.
        \tag{\theequation e} \label{eq:mu_wave_K_e}.
    \end{align}
    Note for $K = 0$, the terms \eqref{eq:mu_wave_K_d} and \eqref{eq:mu_wave_K_e} are absent, as is the expression involving $r \partial_r ( \lambda + 2 \mu) \partial_x^K \mu$ in \eqref{eq:mu_wave_K_a}. We leave the details in this case to the reader, highlighting that the (uncommuted) term $\frac{Q^2}{r^2} e^{2\mu}$ in \eqref{eq:mu_wave_K_a} is responsible for the final term $A_*$ arising in \eqref{eq:energy_der_mu}.

    We now proceed exactly as in the proof of Proposition~\ref{prop:phi_energy}. We first bound the terms \eqref{eq:mu_wave_K_b} and \eqref{eq:mu_wave_K_c}. Just as we did for \eqref{eq:phi_wave_K_b} in the proof of Proposition~\ref{prop:phi_energy}, we apply Lemma~\ref{lem:weightedl2} with $w = r e^{\mu - \lambda}$, to get
    \begin{equation} \label{eq:mu_wave_K_bc_est}
        \| \text{\eqref{eq:mu_wave_K_b}} \|_{L^2} + \| \text{\eqref{eq:mu_wave_K_c}} \|_{L^2} \lesssim_{C_*, K} \sum_{k=0}^K r^{1 - \upgamma(K + 1 - k)} \sqrt{ \mathcal{E}^{(k)} (r)}.
    \end{equation}
    Moving onto \eqref{eq:mu_wave_K_d} and \eqref{eq:mu_wave_K_e}, we now apply Lemma~\ref{lem:weightedl2} with $w=1$. The result is
    \begin{equation} \label{eq:mu_wave_K_de_est}
        \| \text{\eqref{eq:mu_wave_K_d}} \|_{L^2} + \| \text{\eqref{eq:mu_wave_K_e}} \|_{L^2} \lesssim_{C_*, K} \sum_{k=0}^{K-1} r^{- \upgamma(K - k)} \sqrt{ \mathcal{E}^{(k)} (r)}.
    \end{equation}
    We highlight that \eqref{eq:mu_wave_K_de_est} does not have the top order term $\mathcal{E}^{(K)}(r)$ on the RHS, while \eqref{eq:mu_wave_K_bc_est} does feature this top order term but with a favourable weight $r^{1 - \upgamma}$.
    
    Now, just as in the proof of Proposition~\ref{prop:phi_energy} one can use the $\partial_x^K$-commuted wave equation for $\mu$ to derive the following derivative identity for $\mathcal{E}^{(K)}_{\mu}(r)$:
    \begin{align*}
        r \frac{d}{dr} \mathcal{E}^{(K)}_{\mu}(r) 
        &= \int_{\mathbb{S}^1} \left[ (r \partial_r  \partial_x^{K} \mu) \left( \text{\eqref{eq:mu_wave_K_a}} + \partial_x^K \mu \right) + r^2 e^{2(\mu - \lambda)} (\partial_x^{K+1} \mu)^2 \right] \, dx \\[0.6em]
        &\quad + \int_{\mathbb{S}^1} \left[ (r \partial_r \partial_x^K \mu) \left( \text{\eqref{eq:mu_wave_K_b}} + \text{\eqref{eq:mu_wave_K_c}} + \text{\eqref{eq:mu_wave_K_d}} + \text{\eqref{eq:mu_wave_K_e}} - 2 r^2 e^{2(\mu - \lambda)} \partial_x (\mu - \lambda) \partial_x^{K+1} \mu \right) \right] \, dx.
    \end{align*}
    Using the structure of \eqref{eq:mu_wave_K_a} and the bootstrap assumptions \eqref{eq:bootstrap_maxwell}--\eqref{eq:bootstrap_rdr}, one finds $\| \text{\eqref{eq:mu_wave_K_a}} \|_{L^2} \leq 20 C_*^2 \sqrt{2 \mathcal{E}^{(K}(r)}$, and therefore
    \[
        \left| \int_{\mathbb{S}^1} \left[ (r \partial_r  \partial_x^{K} \mu) \left( \text{\eqref{eq:mu_wave_K_a}} + \partial_x^K \mu \right) + r^2 e^{2(\mu - \lambda)} (\partial_x^{K+1} \mu)^2 \right] \, dx \right|
        \leq 2 (20 C_*^2 + 19) \mathcal{E}^{(K)}(r).
    \]

    Furthermore, combining \eqref{eq:mu_wave_K_bc_est} and \eqref{eq:mu_wave_K_de_est} yields that for some $C_1^{(K)}, C_2^{(K)} > 0$:
    \begin{multline*}
        \left | \int_{\mathbb{S}^1} \left[ (r \partial_r \partial_x^K \mu) \left( \text{\eqref{eq:mu_wave_K_b}} + \text{\eqref{eq:mu_wave_K_c}} + \text{\eqref{eq:mu_wave_K_d}} + \text{\eqref{eq:mu_wave_K_e}}  - 2 r^2 e^{2(\mu - \lambda)} \partial_x (\mu - \lambda) \partial_x^{K+1} \mu \right) \right] \, dx \right| \\[0.6em]
        \leq \left( C_1^{(K)} r^{1 - \upgamma} \sqrt{\mathcal{E}^{(K)}(r)} + \sqrt{C_2^{(K)}} \sum_{k=0}^{K-1} r^{-\upgamma(K - k)} \sqrt{\mathcal{E}^{(k)}(r)} \right) \cdot \sqrt{\mathcal{E}^{(K)}(r)}.
    \end{multline*}
    Inserting these bounds into the derivative identity, we get \eqref{eq:energy_der_mu} with $A_* = 20(C_*^2 + 1)$.
\end{proof}

%\subsubsection{The commuted \texorpdfstring{$\mu - \lambda$}{mu -- lambda} transport equation}
It remains to prove the energy estimate for $\lambda$. In light of Proposition~\ref{prop:mu_energy}, it will be enough to prove a corresponding estimate for the modified energy $\mathcal{E}_{\mu - \lambda}^{(K)}(r)$, defined by:
\[
    \mathcal{E}^{(K)}_{\mu - \lambda}(r) \coloneqq \frac{1}{2} \int_{\mathbb{S}^1} \left( (r \partial_r \partial_x^K (\mu - \lambda))^2 + r^2 e^{2(\mu - \lambda)} (\partial_x^{K+1} (\mu - \lambda))^2 + (\partial_x^K (\mu - \lambda))^2 \right) dx.
\]
With this definition, we have the following derivative estimate for the energy:

\begin{proposition} \label{prop:mulambda_energy}
    Let $(\mu, \lambda, \phi, Q)$ be a solution to the Einstein--Maxwell--scalar field system \eqref{eq:lambda_evol}--\eqref{eq:phi_wave} obeying the bootstrap assumptions \eqref{eq:bootstrap_maxwell}--\eqref{eq:bootstrap_d2x}. Then there exists a constant $A_*$ depending only on $C_*$, as well a constant $C_2^{(K)}$ depending on $C_*$ and the regularity index $K \in \{0, 1, \ldots, N \}$, such that
    \begin{equation} \label{eq:energy_der_mulambda}
        \left| r \frac{d}{dr} \mathcal{E}_{\mu-\lambda}^{(K)} (r) \right| \leq
        2 A_* \mathcal{E}^{(K)}(r) + 
        \begin{cases}
            \sum_{k=0}^{K-1} C^{(K)}_2 r^{- 2 \upgamma (K - k)} \mathcal{E}^{(k)}(r), & \text{ if } K \neq 0, \\
            2 A_* & \text{ if } K = 0.
        \end{cases}
    \end{equation}
\end{proposition}

\begin{proof}
    Here we instead commute the transport equation \eqref{eq:mulambda_evol} with $\partial_x^K$, yielding for $1 \leq K \leq N$:
    \begin{equation} \label{eq:mulambda_evol_K}
        r \partial_r \partial_x^K ( \mu - \lambda ) = \left( - \frac{Q^2}{r^2} + \kappa \right) e^{2\mu} \, \partial_x^K \mu + \left( - \frac{Q^2}{r^2} + \kappa \right) e^{2\mu}\sum_{\substack{k_1 < K \\ k_1 + \cdots + k_i = K}} \partial_x^{k_1} \mu * \cdots * \partial_x^{k_i} \mu,
    \end{equation}
    while for $K = 0$ the final term is not present but there is an additional $1$ (see the RHS of \eqref{eq:mulambda_evol}).

    From this equation we can derive further equations, namely
    \begin{multline} \label{eq:mulambda_evol_K_2}
        r \partial_r (r e^{\mu - \lambda} \partial_x^{K+1} ( \mu - \lambda )) = \left( - \frac{Q^2}{r^2} + \kappa \right) e^{2\mu} \, r e^{\mu - \lambda} \partial_x^{K+1} \mu + (1 + r \partial_r(\mu - \lambda)) \, r e^{\mu - \lambda} \partial_x^{K+1} (\mu - \lambda) \\[0.6em]
        + \left( - \frac{Q^2}{r^2} + \kappa \right) e^{2\mu} \, r e^{\mu - \lambda} \sum_{\substack{k_1 < K + 1 \\ k_1 + \cdots + k_i = K + 1}} \partial_x^{k_1} \mu * \cdots * \partial_x^{k_i} \mu,
    \end{multline}
    as well as
    \begin{multline} \label{eq:mulambda_evol_K_3}
        r \partial_r (r \partial_r \partial_x^{K} ( \mu - \lambda )) = \left( - \frac{Q^2}{r^2} + \kappa \right) e^{2\mu} \, r \partial_r \partial_x^K \mu + 
        \left( \frac{2Q^2}{r^2} e^{2\mu} + 2 r \partial_r \mu \left( - \frac{Q^2}{r^2} + \kappa \right) e^{2 \mu} \right) \partial_x^K \mu +
        r \partial_r \partial_x^K (\mu - \lambda)
        \\[0.6em]
        + \left( \frac{2Q^2}{r^2} e^{2\mu} + 2 r \partial_r \mu \left( - \frac{Q^2}{r^2} + \kappa \right) e^{2 \mu} \right) \sum_{\substack{k_1 < K \\ k_1 + \cdots + k_i = K}} \partial_x^{k_1} \mu * \cdots * \partial_x^{k_i} \mu \\[0.6em]
        + \left( - \frac{Q^2}{r^2} + \kappa \right) e^{2\mu} \, \sum_{\substack{k_1 < K \\ k_1 + \cdots + k_i = K}} r \partial_r \partial_x^{k_1} \mu * \cdots * \partial_x^{k_i} \mu,
    \end{multline}

    From \eqref{eq:mulambda_evol_K_3}, \eqref{eq:mulambda_evol_K_2} and \eqref{eq:mulambda_evol_K}, repeated use of the bootstrap assumptions \eqref{eq:bootstrap_maxwell}--\eqref{eq:bootstrap_d2x} and the weighted product estimate Lemma~\ref{lem:weightedl2} yields that for some $C_2^{(K)} > 0$,
    \begin{multline*}
        \| r \partial_r \partial_x^K (\mu - \lambda) \|_{L^2} + \| r \partial_r (r e^{\mu - \lambda} \partial_x^{K+1} (\mu - \lambda)) \|_{L^2} + \| r \partial_r (r \partial_r \partial_x^K (\mu - \lambda)) \|_{L^2}
        \\[0.6em] 
        \leq (A_* - 1) \sqrt{2 \mathcal{E}^{(K)}(r)} + \sqrt{C_2^{(K)}} \sum_{k=0}^{K-1} r^{- \upgamma(K - k)} \sqrt{\mathcal{E}^{(k)}(r)}.
    \end{multline*}
    Proposition~\ref{prop:mulambda_energy} then follows immediately from this and Young's inequality.
\end{proof}

\subsection{The energy hierarchy} \label{sub:energy_induction}

We now use Propositions~\ref{prop:phi_energy}, \ref{prop:mu_energy} and \ref{prop:mulambda_energy} together with the initial data bound \eqref{eq:data_energy_boundedness}, to show that the total energy of order $K$, $\mathcal{E}^{(K)}(r)$, grows at most polynomially in $r$ as $r \downarrow 0$, and moreover that the rate of blow-up depends only mildly in $K$. 

\begin{proposition} \label{prop:energy_hierarchy}
    Let $(\mu, \lambda, \phi, Q)$ be a solution to the Einstein--Maxwell--scalar field system \eqref{eq:lambda_evol}--\eqref{eq:phi_wave} in the interval $r \in [r_b, r_0]$, such that the solution obeys the bootstrap assumptions \eqref{eq:bootstrap_maxwell}--\eqref{eq:bootstrap_d2x}. Assuming also the bound \eqref{eq:data_energy_boundedness} for the initial data, then there exist a constant $A_*$ depending only on $C_*$, as well as constants $C^{(K)}$ depending on $C_*$, $K$ and $\upzeta$, such that for $0 \leq K \leq N$, the total energy $\mathcal{E}^{(K)}(r)$ satisfies the bound: 
    \begin{equation} \label{eq:energy_hierarchy}
        \mathcal{E}^{(K)}(r) \leq C^{(K)} \, r^{- 2 A_* - 2 K \upgamma}.
    \end{equation}
\end{proposition}

\begin{proof}
    Combining Propositions~\ref{prop:phi_energy}, \ref{prop:mu_energy} and \ref{prop:mulambda_energy}, it is straightforward to show that for some $A_* > 0$ depending only on $C_*$ and constants $C^{(K)}_1, C^{(K)}_2 > 0$ (which are allowed to differ from those of the previous propositions), one has the following derivative estimate:
    \begin{equation*}
        \left| r \frac{d}{dr} \mathcal{E}^{(K)}(r) \right| \leq
        2 A_* \mathcal{E}^{(K)}(r) + C^{(K)}_1 r^{1 - \upgamma} \mathcal{E}^{(K)}(r) +
        \begin{cases}
            C^{(K)}_2 \sum_{k = 0}^{K-1} r^{- 2 \upgamma(K-k)} \mathcal{E}^{(k)}(r), & \text{ if } K \neq 0, \\
            2 A_*, & \text{ if } K = 0.
        \end{cases}
    \end{equation*}
    Since we shall integrate towards $r = 0$, the derivative estimate we actually use is the following:
    \begin{equation*}
        r \frac{d}{dr} \mathcal{E}^{(K)}(r) \geq
        - 2 A_* \mathcal{E}^{(K)}(r) - C^{(K)}_1 r^{1 - \upgamma} \mathcal{E}^{(K)}(r) -
        \begin{cases}
            C^{(K)}_2 \sum_{k = 0}^{K-1} r^{- 2 \upgamma(K-k)} \mathcal{E}^{(k)}(r), & \text{ if } K \neq 0, \\
            2 A_*, & \text{ if } K = 0.
        \end{cases}
    \end{equation*}

    In fact, using the integrating factor $r^{2 A_*}$, which is \emph{crucially independent of $K$}, we write:
    \begin{equation} \label{eq:energy_der}
        r \frac{d}{dr} \left( r^{2 A_*} \mathcal{E}^{(K)}(r) \right) \geq
        - C^{(K)}_1 r^{1 - \upgamma} \left( r^{2 A_*} \mathcal{E}^{(K)}(r) \right) -
        \begin{cases}
            C^{(K)}_2 \sum_{k = 0}^{K-1} r^{- 2 \upgamma(K-k)} \left( r^{2 A_*} \mathcal{E}^{(k)}(r) \right), & \text{ if } K \neq 0, \\
            r^{2 A_*} A_*, & \text{ if } K = 0.
        \end{cases}
    \end{equation}
     
    We will now use \eqref{eq:energy_der} and induction on $K \in \{0, \ldots, N \}$ to show that 
    \begin{equation} \label{eq:energy_hierarchy_induction}
        r^{2 A_*} \mathcal{E}^{(K)}(r) \lesssim r^{- 2 \upgamma K},
    \end{equation}
    where the implied constant is now allowed to depend on $C_*$, $K$ and $\upzeta$. This is equivalent to \eqref{eq:energy_hierarchy}. Note that the dependence on $\upzeta$ comes from the fact that the initial data bound \eqref{eq:data_energy_boundedness} implies that for $A_* \geq \upgamma$,
    \begin{equation} \label{eq:energy_hierarchy_init}
        \sum_{k=0}^N r_0^{2 A_*} \mathcal{E}^{(k)}(r_0) \leq \upzeta.
    \end{equation}

    We begin with the base case $K = 0$. Applying Gr\"onwall's inequality to \eqref{eq:energy_der} for $K = 0$, we obtain that for any $r \in [r_b, r_0]$, one has
    \[
        r^{2 A_*} \mathcal{E}^{(0)}(r) \leq \exp(F^{(0)}(r_0, r)) \cdot r_0^{2 A_*} \mathcal{E}^{(0)}(r_0) + \int^{r_0}_r \exp( F^{(0)} (\tilde{r}, r)) \tilde{r}^{2 A_*} 2 A_* \frac{d \tilde{r}}{\tilde{r}},
    \]
    \[
        \text{ where } \qquad F^{(0)}(s_a, s_b) = \int^{s_a}_{s_b} C_1^{(0)} \tilde{r}^{1 - \upgamma} \frac{d \tilde{r}} {\tilde{r}}.
    \]
    Since $F^{(0)}(s_a, s_b)$ is uniformly bounded for $s_a, s_b \in [r_b, r_0]$, it follows from the initial data bound \eqref{eq:energy_hierarchy_init} that \eqref{eq:energy_hierarchy_induction} holds for $K = 0$.

    Moving onto the induction step, assume that \eqref{eq:energy_hierarchy_induction} holds for $0 \leq K < \bar{K} \leq N$; we wish to prove it also holds for $K = \bar{K}$. Applying Gr\"onwall's inequality to \eqref{eq:energy_der} for $K = \bar{K}$, we have that
    \[
        r^{2 A_*} \mathcal{E}^{(\bar{K})}(r) \leq \exp(F^{(\bar{K})}(r_0, r)) \cdot r_0^{2 A_*} \mathcal{E}^{(\bar{K})}(r_0) + \int^{r_0}_r \exp( F^{(\bar{K})} (\tilde{r}, r)) C_2^{(\bar{K})} \sum_{k=0}^{\bar{K}-1} \tilde{r}^{-2\upgamma(\bar{K}-k)} \tilde{r}^{2 A_*} \mathcal{E}^{(k)}(\tilde{r}) \frac{d \tilde{r}}{\tilde{r}},
    \]
    \[
        \text{ where } \qquad F^{(\bar{K})}(s_a, s_b) = \int^{s_a}_{s_b} C_1^{(\bar{K})} \tilde{r}^{1 - \upgamma} \frac{d \tilde{r}} {\tilde{r}}.
    \]
    Since $F^{(K)}(s_a, s_b)$ is uniformly bounded for $s_a, s_b \in [r_b, r_0]$, it follows from the initial data bound \eqref{eq:energy_hierarchy_init} and the inductive hypothesis for $\tilde{r}^{2 A_*} \mathcal{E}^{(k)}(\tilde{r})$ that
    \[
        r^{A_*} \mathcal{E}^{(\bar{K})}(r) \lesssim \upzeta + \int_r^{r_0} \sum_{k=0}^{\bar{K}-1} \tilde{r}^{-2\upgamma(\bar{K} - k)} \cdot \tilde{r}^{- 2 \upgamma k} \frac{d \tilde{r}}{\tilde{r}} \lesssim r^{-2 \upgamma \bar{K}}
    \]
    as required. This completes the proof of the proposition.
\end{proof}

\subsection{Derivation of bounce ODEs}

We now apply Proposition~\ref{prop:energy_hierarchy} together with the Sobolev interpolation of Lemma~\ref{lem:interpolation} to provide $L^{\infty}$ bounds for low order derivatives of $\phi$, $\mu$ and $\lambda$, so that we may treat certain equations in the Einstein--Maxwell--scalar field system \eqref{eq:mu_evol}--\eqref{eq:phi_wave} as ODEs without worrying about losing derivatives.

\begin{lemma} \label{lem:low_order_linfty}
    Let $(\mu, \lambda, \phi, Q)$ be as in Proposition~\ref{prop:energy_hierarchy}. Then for $N$ chosen sufficiently large there exists a family of constants $\updelta = \updelta ( C_*, \upzeta, r_*) > 0$ with $\updelta \downarrow 0$ as $r_* \downarrow 0$, such that for any $0 \leq k \leq 3$, $f \in \{ \phi, \mu, \lambda \}$, and $r \in [r_b, r_0]$ one has
    \begin{equation} \label{eq:low_order_linfty}
        \| \partial_x^k f (r, \cdot) \|_{L^{\infty}} + \| r \partial_r \partial_x^k f (r, \cdot) \|_{L^{\infty}} \leq \updelta r^{-1/2}.
    \end{equation}
\end{lemma}

\begin{proof}
    The energy estimate \eqref{eq:energy_hierarchy} for $K = N$, together with the definition of $\mathcal{E}^{(K)}(r)$ (see Definition~\ref{def:energy}), implies the following $L^2$-estimates for $f \in \{ \phi, \mu, \lambda \}$:
    \begin{equation} \label{eq:l2_top}
        \| \partial_x^N f \|_{L^2}^2 + \| r \partial_r \partial_x^N f \|_{L^2}^2 \leq 2 C^{(N)} r^{-2 A_* - 2 N \upgamma}.
    \end{equation}
    Note that while $C^{(N)}$ depends on $N$, the number $A_*$ does not, and we later choose $N$ depending on $A_*$.

    We now interpolate between \eqref{eq:l2_top} and the low-order $L^{\infty}$ bootstrap assumptions \eqref{eq:bootstrap_rdr}--\eqref{eq:bootstrap_dx}. Applying Lemma~\ref{lem:interpolation}, for $f \in \{ \phi, \mu, \lambda \}$ and $0 \leq k \leq 3$ one finds:
    \[
        \| r \partial_r \partial_x^k f (r, \cdot ) \|_{L^{\infty}} \lesssim_{N} \| r \partial_r f (r, \cdot) \|_{L^{\infty}}^{1-\alpha} \| r \partial_r \partial_x^N f (r, \cdot) \|_{L^2}^{\alpha}, \qquad \text{ where } \alpha = \frac{k}{N - \frac{1}{2}}.
    \]
    Inserting the bound \eqref{eq:l2_top} and the bootstrap assumption \eqref{eq:bootstrap_rdr}, one finds
    \[
        \| r \partial_r \partial_x^k f (r, \cdot) \|_{L^{\infty}} \lesssim_{N, C_*} \left( r^{- 2 A_* - 2 N \upgamma } \right)^{\frac{\alpha}{2}} = r^{- k \upgamma} \cdot \left( r^{- 2 A_* - \upgamma} \right)^{\frac{\alpha}{2}}.
    \]

    Since we chose $\upgamma < \frac{1}{100}$, we have $r^{- k \upgamma} < r^{- 1/8}$. Furthermore, for $N$ chosen sufficiently large (depending on $A_*$ and $k$) one can guarantee $(r^{- 2 A_* - \upgamma})^{\frac{\alpha}{2}} < r^{- 1/8}$. Thus for this choice of $N$ there exists a constant $C_{N, C_*}$ such that
    \begin{equation} \label{eq:interp_rdr}
        \| r \partial_r \partial_x^k f (r, \cdot) \|_{L^{\infty}} \leq C_{N, C_*} r^{- 1/ 4}.
    \end{equation}
    To get a similar estimate for $\| \partial_x^k f (r, \cdot) \|_{L^{\infty}}$ we first apply the Lemma~\ref{lem:interpolation} to the initial data assumption \eqref{eq:data_energy_boundedness}, yielding $ \| r \partial_r \partial_x^k f (r_0, \cdot) \|_{L^{\infty}} \lesssim \upzeta r^{-\upgamma} \lesssim \upzeta r^{- 1/4}$. Then simply integrating \eqref{eq:interp_rdr} from $r = r_0$, there exists a constant $C_{N, C_*, \upzeta}$ so that
    \begin{equation} \label{eq:interp_base} 
        \| \partial_x^k f (r, \cdot) \|_{L^{\infty}} \leq C_{N, C_*, \upzeta} r^{- 1/ 4}.
    \end{equation}

    Combining \eqref{eq:interp_rdr} and \eqref{eq:interp_base}, and redefining the constant $C_{N, C_*, \upzeta}$ appropriately, one has
    \begin{equation*} %\label{eq:low_order_linfty}
        \| \partial_x^k f (r, \cdot) \|_{L^{\infty}} + \| r \partial_r \partial_x^k f (r, \cdot) \|_{L^{\infty}} \leq ( C_{N, C_*, \upzeta} r^{1/4} ) \cdot r^{-1/2}.
    \end{equation*}
    Recalling that $r \leq r_0 \leq r_*$, letting $\updelta = C_{N, C_*, \upzeta} \cdot r_*^{1/4}$ completes the proof of the lemma.
\end{proof}

\begin{corollary} \label{cor:ode}
    Let $(\mu, \lambda, \phi, Q)$ be as in Proposition~\ref{prop:energy_hierarchy}. Then for $N$ chosen sufficiently large there exists a family of constants $\updelta = \updelta(C_*, \upzeta, r_*) > 0$, with $\updelta \downarrow 0$ as $r_* \downarrow 0$, such that for all $(r, x) \in [r_b, r_0] \times \mathbb{S}^1$ and all $a \in [-1, 1]$, one has:
    \begin{gather}
        \left| (r \partial_r + a r e^{\mu - \lambda} \partial_x) (r \partial_r \phi) + (r \partial_r \phi) \frac{Q^2}{r^2} e^{2\mu} \right| \, (r, x) \leq \updelta r^{1/2}, \label{eq:P_ode_err} \\[0.6em]
        \left| (r \partial_r + a r e^{\mu - \lambda} \partial_x) \left( \frac{Q^2}{r^2} e^{2\mu} \right) - \frac{Q^2}{r^2}e^{2\mu} \left( (r \partial_r \phi)^2 - 1 - \frac{Q^2}{r^2} e^{2\mu} \right) \right| \, (r, x) \leq \frac{Q^2}{r^2} e^{2 \mu} \cdot \updelta r^{1/2}. \label{eq:Q_ode_err}
    \end{gather}
    Furthermore, for the same $\updelta = \updelta(C_*, \upzeta, r_*) > 0$ one has:
    \begin{gather}
        \left| (r \partial_r + a r e^{\mu - \lambda} \partial_x) (\partial_x \mu) + \frac{Q^2}{r^2} e^{2\mu} \partial_x \mu - (r \partial_r \phi) (r \partial_r \partial_x \phi) \right| \, (r, x) \leq \updelta r^{1/2}, \label{eq:M_ode_err} \\[0.6em]
        \left| (r \partial_r + a r e^{\mu - \lambda} \partial_x) (r \partial_r \partial_x \phi) + \frac{2 Q^2}{r^2}e^{2\mu} (r \partial_r \phi) \partial_x \mu + \frac{Q^2}{r^2} e^{2\mu} r \partial_r \partial_x \phi \right| \, (r, x) \leq \updelta r^{1/2}. \label{eq:N_ode_err}
    \end{gather}
\end{corollary}

\begin{proof}
    This follows immediately from Lemma~\ref{lem:low_order_linfty} and the equations \eqref{eq:mu_evol} and \eqref{eq:phi_wave}. We derive the equations \eqref{eq:Q_ode_err} and \eqref{eq:N_ode_err}, leaving the remaining equations to the reader. For \eqref{eq:Q_ode_err}, using \eqref{eq:mu_evol} we derive:
    \begin{align*}
        r \partial_r \left( \frac{Q^2}{r^2} e^{2 \mu} \right) 
        &= \frac{Q^2}{r^2} e^{2 \mu} \left( (r \partial_r \phi)^2 + r^2 e^{2(\mu - \lambda)} (\partial_x \phi)^2 + \left( - \frac{Q^2}{r^2} + \kappa \right) e^{2 \mu} - 2 \right) \\[0.6em]
        &= \frac{Q^2}{r^2} e^{2 \mu} \left( (r \partial_r \phi)^2 - 1 - \frac{Q^2}{r^2}e^{2\mu} \right)
        + \frac{Q^2}{r^2} e^{2 \mu} \left( r^2 e^{2 (\mu - \lambda)} (\partial_x \phi)^2 + \kappa e^{2 \mu} \right).
    \end{align*}

    Noting there is also an $are^{\mu - \lambda} \partial_x$-derivative on the left hand side of \eqref{eq:Q_ode_err}, it thus suffices to show that
    \begin{equation} \label{eq:Q_ode_err_bound}
        r^2 e^{2 (\mu - \lambda)} (\partial_x \phi)^2 + |\kappa e^{2\mu}| + 2 |a r e^{\mu - \lambda} \partial_x \mu | \leq \updelta r^{1/2}.
    \end{equation}
    Inserting the bootstrap assumptions \eqref{eq:bootstrap_maxwell} to bound $e^{\mu - \lambda}$ and $e^{2 \mu}$, and Lemma~\ref{lem:low_order_linfty} to bound $\partial_x \phi$ and $\partial_x \mu$ ,
    \[
        r^2 e^{2 (\mu - \lambda)} (\partial_x \phi)^2 + |\kappa e^{2 \mu}| + 2 |a r e^{\mu - \lambda} \partial_x \mu | \leq C_*^2 r^3 ( \updelta r^{-1/2} )^2 + C_* r + 2 C_* r^{3/2} (\updelta r^{-1/2}).
    \]
    Since the powers of $r$ appearing on the right hand side all exceed $r^{1/2}$, it is clear that upon redefining $\updelta$ that one has the estimate \eqref{eq:Q_ode_err_bound}.

    We move onto \eqref{eq:N_ode_err}. For this, we commute the wave equation \eqref{eq:phi_wave} once with $\partial_x$, then add an additional term of $a r e^{\mu - \lambda} r \partial_r \partial_x^2 \phi$, to yield
    \begin{multline*}
        (r \partial_r)^2 \partial_x \phi + a r e^{\mu - \lambda} r \partial_r \partial_x^2 \phi + 2 \frac{Q^2}{r^2} e^{2 \mu} (\partial_x \mu) (r \partial_r \phi) + \frac{Q^2}{r^2} e^{2\mu} r \partial_r \partial_x \phi = 
        \\[0.6em]
        r^2 e^{2(\mu - \lambda)} \left( \partial_x^3 \phi + 3 \partial_x(\mu - \lambda) \partial_x^2 \phi + \partial_x^2(\mu - \lambda) \partial_x \phi \right) + \kappa e^{2\mu} \left( r \partial_r \partial_x \phi + 2 (\partial_x \mu) (r \partial_r \phi) \right) + a r e^{\mu - \lambda} r \partial_r \partial_x^2 \phi.
    \end{multline*}
    Using the bootstrap assumptions \eqref{eq:bootstrap_maxwell}--\eqref{eq:bootstrap_rdr} and Lemma~\ref{lem:low_order_linfty}, the right hand side may be bounded by
    \[
        C_* r^2 ( \updelta r^{-1/2} + (\updelta r^{-1/2})^2 ) + C_* r ( \updelta r^{-1/2} + 2 C_* \updelta r^{-1/2}) + C_* r^2 \updelta r^{-1/2}.
    \]

    Since we are allowed to redefine $\updelta = C_* \updelta$, this quantity is bounded by $\updelta r^{1/2}$, as required.
\end{proof}

%auto-ignore
\section{Low order ODE analysis} \label{sec:ode}

Here, we now study ODE systems of the type found in Corollary~\ref{cor:ode}, with the eventual aim of improving the bootstrap assumptions \eqref{eq:bootstrap_maxwell}--\eqref{eq:bootstrap_d2x}. As explained in Section~\ref{sub:intro_proof}, the first step is studying a nonlinear ODE system derived from \eqref{eq:P_ode_err}--\eqref{eq:Q_ode_err}, and then studying the linearization of this ODE system, derived from \eqref{eq:M_ode_err}--\eqref{eq:N_ode_err}.

\subsection{The BKL bounce ODE}

\begin{lemma} \label{lem:ode_bounce}
    For $r_0 \leq 1$, let $\mathscr{P}, \mathscr{Q}: [r_b, r_0] \subset \R_{>0}\to \R_{\geq 0}$ satisfy the following ODEs, where the error terms $\mathscr{E}_i$ obey $|\mathscr{E}_i| \leq \updelta r^{1/2}$ for $i = 1, 2$:
    \begin{gather}
        \label{eq:p_eq}
        r \partial_r \mathscr{P} = - \mathscr{P} \mathscr{Q} + \mathscr{E}_1, \\[0.5em]
        \label{eq:q_eq}
        r \partial_r \mathscr{Q} = \mathscr{Q} ( \mathscr{P}^2 - 1 - \mathscr{Q} + \mathscr{E}_2).
    \end{gather}
    Suppose furthermore that for some $\upeta > 2$, one has $\upeta^{-1} \leq \mathscr{P}(r_0) \leq \upeta$ and $\mathscr{Q}(r_0) \leq 1$. Then for $\updelta$ chosen sufficiently small (depending on $\upeta$) the solution obeys the following bounds for $r \in [r_b, r_0]$.
    \begin{equation} \label{eq:pq_bounds}
        (4 \upeta)^{-1} \leq \mathscr{P}(r) \leq 4 \upeta, \qquad \mathscr{Q}(r) \leq 16 \upeta^2.
    \end{equation}
\end{lemma}

\begin{proof}
    We proceed by using a continuity / bootstrap argument, with the bootstrap assumption being \eqref{eq:pq_bounds}. So we assume \eqref{eq:pq_bounds} holds on an interval $[\tilde{r}, r_0] \subset [r_b, r_0]$, then show we may actually improve upon \eqref{eq:pq_bounds}.

    For the improvement step, we use an approximately conserved quantity of the ODE (which would be exactly conserved if $\mathscr{E}_i \equiv 0 $); let $\mathscr{K}$ be defined by
    \begin{equation} \label{eq:ode_k}
        \mathscr{K} \coloneqq \mathscr{P} + \mathscr{P}^{-1} + \mathscr{Q} \mathscr{P}^{-1}.
    \end{equation}
    From \eqref{eq:p_eq} and \eqref{eq:q_eq}, one checks that $r \partial_r \mathscr{K} = \mathscr{E}_1 (1 - \mathscr{P}^{-2} - \mathscr{P}^{-2} \mathscr{Q}) + \mathscr{P}^{-1} \mathscr{Q} \mathscr{E}_2$. Thus assuming the bootstrap assumption \eqref{eq:pq_bounds}, $r \partial_r \mathscr{K}$ may be bounded by
    \[
        |r \partial_r \mathscr{K} | \leq (1 + 16 \upeta^2 + 64 \upeta^3 + 256 \upeta^4) \updelta r^{1/2}.
    \]

    Thus for $\updelta = \updelta(\upeta)$ chosen sufficiently small, one may guarantee that $\int^{r_0}_{r} |\partial_r \mathscr{K}(\tilde{r})| d \tilde{r} \leq \frac{1}{2} \upeta$. Using also the initial data assumptions $\upeta^{-1} \leq \mathscr{P}(r_0) \leq \upeta$ and $\mathscr{Q}(r_0) \leq 1$, we have
    \[
        \mathscr{K}(r) \leq \mathscr{K}(r_0) + \frac{1}{2} \upeta = \mathscr{P}(r_0) + \mathscr{P}^{-1}(r_0) + \mathscr{Q} \mathscr{P}^{-1}(r_0) + \frac{1}{2} \upeta \leq \frac{7}{2} \upeta.
    \]
    Hence for all $r$, one has
    \[
        \mathscr{P}(r) + \mathscr{P}^{-1}(r) + \mathscr{Q} \mathscr{P}^{-1}(r) \leq \frac{7}{2} \upeta,
    \]
    from which one can read off $\mathscr{P}(r) + \mathscr{P}(r) \leq \frac{7}{2} \upeta$ and $\mathscr{Q}(r) \leq \left( \frac{7}{2} \upeta \right)^2$, improving upon the bootstrap assumption \eqref{eq:pq_bounds}. By a standard continuity argument, \eqref{eq:pq_bounds} therefore holds in the entire interval $[r_b, r_0]$.
\end{proof}

\subsection{The equations of variation}

\begin{lemma} \label{lem:ode_variation}
    For $r_0 < 1$, let $\mathscr{P}, \mathscr{Q}: [r_b, r_0] \to \R_{\geq 0}$ satisfy the assumptions of Lemma~\ref{lem:ode_bounce}. Further, let $\mathscr{M}, \mathscr{N}: [r_b, r_0] \to \R$ obey the following ODEs, where $|\mathscr{E}_i| \leq \updelta r^{1/2}$ for $i=1,2,3,4$:
    \begin{gather}
        r \partial_r \mathscr{M} = - \mathscr{Q} \mathscr{M} + \mathscr{P} \mathscr{N} + \mathscr{E}_3, \label{eq:m_eq} \\[0.5em]
        r \partial_r \mathscr{N} = - 2 \mathscr{P} \mathscr{Q} \mathscr{M} - \mathscr{Q} \mathscr{N} + \mathscr{E}_4. \label{eq:n_eq}
    \end{gather}
    For $0 < \upgamma < \frac{1}{100}$ fixed and some $\upzeta > 0$, impose the following conditions at initial data: $|\mathscr{M}(r_0)| + |\mathscr{N}(r_0)| \leq \upzeta r_0^{- \upgamma}$. Then for $\updelta$ chosen sufficiently small depending on $\upeta$, $\upgamma$ and $\upzeta$, there exists a constant $D > 0$ depending on the same constants $\upeta$, $\upgamma$ and $\upzeta$, such that for all $r \in [r_b, r_0]$, one has
    \begin{equation} \label{eq:mn_bounds}
        |\mathscr{M} (r)| + |\mathscr{N} (r)| \leq D r^{-\upgamma}.
    \end{equation}
\end{lemma}

\begin{proof}
    Let $b = \frac{\upgamma}{10 \upeta}$ be a small parameter used for notational convenience. We rewrite the system \eqref{eq:m_eq}--\eqref{eq:n_eq} in the following matrix form, upon conjugation by a linear transformation $\mathscr{M} \mapsto b \mathscr{M}$:
    \begin{equation} \label{eq:mn_matrix}
        r \partial_r \begin{bmatrix} b \mathscr{M} \\ \mathscr{N} \end{bmatrix} =
        \underbrace{
            \begin{bmatrix} - \mathscr{Q} & b \mathscr{P} \\ - 2 b^{-1} \mathscr{P} \mathscr{Q} & - \mathscr{Q} \end{bmatrix}
        }_{\eqqcolon \mathbf{L}}
        \begin{bmatrix} b \mathscr{M} \\ \mathscr{N} \end{bmatrix}
        +
        \begin{bmatrix} b \mathscr{E}_3 \\ \mathscr{E}_4 \end{bmatrix}.
    \end{equation}

    Our goal is to bound the operator norm of the matrix $\mathbf{L} = \mathbf{L}(r)$, as a function of $r \in [r_b, r_0]$. As $b$ is chosen small, such a bound on the operator norm will depend on the size of $\mathscr{Q} = \mathscr{Q}(r)$. Note our definition of the operator norm will be with respect to the $\ell^{\infty}$-norm on $\R^2$, i.e.~for a $2 \times 2$ matrix $\mathbf{M}$ we write 
    \[
        \| \mathbf{M} \|_{op} \coloneqq \sup_{ \mathrm{x} \in \R^2 \setminus \{0\}} \frac{ \| \mathbf{M} \mathrm{x} \|_{\ell^{\infty}}} { \| \mathrm{x} \|_{\ell^{\infty}}}.
    \]

    \medbreak \noindent
    \underline{Case 1: $\mathscr{Q}(r) \leq \frac{1}{16} b \upeta^{-1}$:} 
    \medbreak \noindent
    When this holds it can be immediately checked from this bound and \eqref{eq:pq_bounds} that each of the matrix elements of $\mathbf{L}(r)$ are bounded in absolute value by $\frac{\upgamma}{2}$. Since for a $2 \times 2$ matrix $\mathbf{M}$ we have:
    \[
        \| \mathbf{M} \|_{op} \leq 2 \max_{i,j} | \mathbf{M}_{ij} |, 
    \]
    this implies that
    \begin{equation} \label{eq:L_op_small}
        \| \mathbf{L}(r) \|_{op} \leq \upgamma \qquad \text{ whenever } \mathscr{Q}(r) \leq \frac{1}{16} b \upeta^{-1}.
    \end{equation}

    \medbreak \noindent
    \underline{Case 2: $\mathscr{Q}(r) > \frac{1}{16} b \upeta^{-1}$:} 
    \medbreak \noindent
    With no additional smallness for $\mathscr{Q}(r)$ the only bounds we have for $\mathbf{L}(r)$ will come from the bound \eqref{eq:pq_bounds} from Lemma~\ref{lem:ode_bounce}. By bounding each of the matrix elements of $\mathbf{L}(r)$ using \eqref{eq:pq_bounds}, one may deduce $\| \mathbf{L}(r) \|_{op} \leq 256 b^{-1} \upeta^3$.

    What is needed in this case is that we can control the size of the set $B = \{ r \in [r_b, r_0] : \mathscr{Q}(r) > \frac{1}{16} b \upeta^{-1} \}$, at least for $\updelta$ sufficiently small. To justify this, suppose $\updelta < \frac{1}{128} b \upeta^{-2}$. Then using \eqref{eq:pq_bounds} and the equation \eqref{eq:p_eq} for $r \partial_r \mathscr{P}$, one has that for $r \in B$
    \[
        r \partial_r \mathscr{P} \leq - (4 \eta)^{-1} \cdot \frac{1}{16} b \upeta^{-1} + \updelta \leq - \frac{1}{128} b \upeta^{-2}.
    \]
    But from \eqref{eq:pq_bounds}, $\mathscr{P}(r)$ is bounded between $(4 \upeta)^{-1}$ and $4 \upeta$. Furthermore, for all $r \in [r_b, r_0]$, not necessarily in $B$, one has $r \partial_r \mathscr{P} (r) \leq \updelta r^{1/2} < \frac{1}{128}b \upeta^{-2} r^{1/2}$. So
    \begin{align*}
        4 \upeta - (4 \upeta)^{-1} \geq \int_r^{r_0} - r \partial_r \mathscr{P} (\tilde{r}) \, \frac{d\tilde{r}}{\tilde{r}}
        &= \int_{\tilde{r} \in B} - r \partial_r P(\tilde{r}) \, \frac{d \tilde{r}}{\tilde{r}} + \int_{r \not\in B} - r \partial_r P (\tilde{r}) \, \frac{d \tilde{r}}{\tilde{r}}  \\
        &\geq \nu(B) \cdot \frac{1}{128} b \upeta^{-2} - \frac{1}{64} b \upeta^{-2} r_0^{1/2}.
    \end{align*}
    Here $\nu(B)$ is the measure of the set $B$ with respect to the measure $\frac{dr}{r}$.
    Using the smallness of $b$, this leads to the estimate $\nu(B) \leq 512 b^{-1} \upeta^3$. %As a result, the set $R$ is either empty, or must comprise of a single closed interval $[r_{R0}, r_{R1}]$ such that $\log( \frac{r_{R1}}{r_{R0}} ) \leq 512 a^{-1} \upeta^3$.

    Collecting all this information, we have that
    \begin{gather} \label{eq:L_op_big}
        \| \mathbf{L}(r) \|_{op} \leq 256 b^{-1} \upeta^3 \qquad \text{ whenever } \mathscr{Q} > \frac{1}{16} b \upeta^{-1},
        \\[0.5em]
        \text{where } B = \left \{ r \in [r_b, r_0] : \mathscr{Q}(r) > \frac{1}{16} b \upeta^{-1} \right \} \text{ has } \nu(B) \leq 512 b^{-1} \upeta^3. %[r_{R0}, r_{R1}] \text{ where } \log( \frac{r_{R1}}{r_{R0}} ) \leq 512 a^{-1} \upeta^3.
        \nonumber
    \end{gather}

    \medbreak %\noindent
    We now use \eqref{eq:L_op_small} and \eqref{eq:L_op_big} to complete the proof. Integrating \eqref{eq:mn_matrix} and applying $\| \cdot \|_{\ell^{\infty}}$ norms on vectors, one finds the integral inequality:
    \[
        \left \| \begin{bmatrix} b \mathscr{M} \\ \mathscr{N} \end{bmatrix} \right \|_{\ell^{\infty}} \mkern-18mu (r) \leq 
        \left \| \begin{bmatrix} b \mathscr{M} \\ \mathscr{N} \end{bmatrix} \right \|_{\ell^{\infty}} \mkern-18mu (r_0)
        + \bigintsss_r^{r_0}
        \left \| \begin{bmatrix} b \mathscr{E}_3 \\ \mathscr{E}_4 \end{bmatrix} \right \|_{\ell^{\infty}} \mkern-18mu (\tilde{r}) \frac{d \tilde{r}}{\tilde{r}}
        +  \bigintsss_r^{r_0}
        \| \mathbf{F} (\tilde{r}) \|_{op} \left \| \begin{bmatrix} b \mathscr{M} \\ \mathscr{N} \end{bmatrix} \right \|_{\ell^{\infty}} \mkern-18mu (\tilde{r}) \frac{d \tilde{r}}{\tilde{r}}.
    \]

    By using the initial data assumption $|\mathscr{M}(r_0)| + |\mathscr{N}(r_0)| \leq \upzeta r_0^{- \upgamma}$ along with the error bound $| \mathscr{E}_i | \leq \updelta r^{1/2}$, there exists some constant $D_1$ depending on $\upeta$, $\upgamma$ and $\upzeta$ such that
    \[
        \left \| \begin{bmatrix} b \mathscr{M} \\ \mathscr{N} \end{bmatrix} \right \|_{\ell^{\infty}} \mkern-18mu (r) \leq 
        D_1 r_0^{- \upgamma}
        +  \bigintsss_r^{r_0}
        \| \mathbf{F} (\tilde{r}) \|_{op} \left \| \begin{bmatrix} b \mathscr{M} \\ \mathscr{N} \end{bmatrix} \right \|_{\ell^{\infty}} \mkern-18mu (\tilde{r}) \frac{d \tilde{r}}{\tilde{r}}.
    \]
    Thus Gr\"onwall's inequality in integral form implies that
    \[
        \left \| \begin{bmatrix} a \mathscr{M} \\ \mathscr{N} \end{bmatrix} \right \|_{\ell^{\infty}} \mkern-18mu (r) \leq 
        D_1 r_0^{-\upgamma} \cdot \exp ( \int^{r_0}_r \| \mathbf{F}(\tilde{r}) \|_{op} \frac{d \tilde{r}}{\tilde{r}} ).
    \]

    We now apply \eqref{eq:L_op_small} and \eqref{eq:L_op_big} to estimate the integral inside the $\exp$. Using the characterisation of the set $R$ in the case of \eqref{eq:L_op_big}, the result is that there exists some $D_2$ depending on $\upeta$ and $\upgamma$ such that
    \[
        \int^{r_0}_r \| \mathbf{F}(\tilde{r}) \|_{op} \frac{d \tilde{r}}{\tilde{r}} \leq \int_{\tilde{r} \in B} 126b^{-1} \upeta^3 \frac{d \tilde{r}}{\tilde{r}} + \int_{\tilde{r} \not \in B} \upgamma \, \frac{d \tilde{r}}{\tilde{r}}  \leq 126 b^{-1} \upeta^3 \nu(B) + \int^{r_0}_r \upgamma \, \frac{d \tilde{r}}{\tilde{r}} = \upgamma \log( \frac{r_0}{r} ) + D_2.
    \]
    Inserting this into the above inequality and using that $b = \frac{\upgamma}{10 \upeta}$ has explicit dependence on $\upgamma$ and $\upeta$ we get \eqref{eq:mn_bounds}, as required.
\end{proof}

\section{Global existence towards \texorpdfstring{$r = 0$}{t=0}}

\subsection{Proof of Theorem~\ref{thm:global}}

The proof of Theorem~\ref{thm:global} will follow from a bootstrap argument. By local existence (Proposition~\ref{prop:lwp}), there exists some $r_b \in (0, r_0)$ such that a solution to the Einstein--Maxwell--scalar field system \eqref{eq:mu_evol}--\eqref{eq:phi_wave} exists in the interval $r \in [r_b, r_0]$, which moreover satisfies the bootstrap assumptions \eqref{eq:bootstrap_maxwell}--\eqref{eq:bootstrap_d2x}. 

We show that assuming this, we are able to improve upon the bootstrap assumptions, for instance by showing that \eqref{eq:bootstrap_maxwell}--\eqref{eq:bootstrap_d2x} hold with $C_*$ replaced by $C_*/2$. By a standard continuity argument, one can allow $r_b$ to be any real number in the interval $(0, r_0)$, giving global existence towards $r = 0$.

For this purpose, let $(\phi, \mu, \lambda, Q)$ be a solution of the Einstein--Maxwell--scalar field system in our bootstrap region $r \in [r_b, r_0]$ with the initial data assumptions \eqref{eq:data_linfty}, \eqref{eq:data_linfty2} and \eqref{eq:data_energy_boundedness}. Then the results of Section~\ref{sec:l2} all apply, in particular Proposition~\ref{prop:energy_hierarchy}, Lemma~\ref{lem:low_order_linfty} and Corollary~\ref{cor:ode}. For convenience, let us record the result of Lemma~\ref{lem:low_order_linfty} again here: for $f \in \{ \phi, \mu, \lambda \}$ and $0 \leq k \leq 3$, there exists a family of constants $\updelta = \updelta( C_*, \upzeta, r_* )$ with $\updelta \downarrow 0$ as $r_* \downarrow 0$ so that
\begin{equation} \label{eq:low_order_linfty_copy}
    \| \partial_x^k f (r, \cdot) \|_{L^{\infty}} + \| r \partial_r \partial_x^k f (r, \cdot) \|_{L^{\infty}} \leq \updelta r^{-1/2}.
\end{equation}

The improvement of the bootstrap assumptions will be proceed in the following steps:

\medbreak \noindent
\underline{Step 1: ODE analysis on timelike curves} 
\medbreak \noindent
The first step will be to use Corollary~\ref{cor:ode} and the results of Section~\ref{sec:ode} to provide $L^{\infty}$ bounds for the following $4$ key quantities:
\[
    r \partial_r \phi, \quad \frac{Q^2}{r^2} e^{2\mu}, \quad \partial_x \mu, \quad r \partial_r \partial_x \phi.
\]

For this purpose, let $\gamma: [r_b, r_0] \to \mathcal{Q}$ with $\gamma(r) = (r, x(r))$ be a $C^1$ \emph{past-directed timelike curve}, one example being a curve of constant $x$. In view of the metric \eqref{eq:surface_sym}, $\gamma$ being timelike implies that $|\frac{dx}{dr}| \leq e^{\mu - \lambda}$. For such a curve $\gamma(r)$, let us define
\begin{gather}
    \mathscr{P}(r) \coloneqq r \partial_r \phi ( \gamma(r)), \quad \mathscr{Q}(r) \coloneqq \frac{Q^2}{r^2} e^{2\mu} (\gamma(r)), 
    \label{eq:PQ_def} \\[0.6em]
    \mathscr{M}(r) \coloneqq \partial_x \mu ( \gamma(r)), \quad \mathscr{N}(r) \coloneqq r \partial_r \partial_x \phi ( \gamma (r)).
\end{gather}

Note the initial data assumption \eqref{eq:data_linfty} implies that $\upeta^{-1} \leq \mathscr{P}(r_0) \leq \upeta$ and $\mathscr{Q}(r_0) \leq 1$. Furthermore, the assumption \eqref{eq:data_energy_boundedness}, together with Sobolev embedding on $\mathbb{S}^1$, implies there is a constant $C_{\mathbb{S}^1}$ independent of all parameters such that $|\mathscr{M}(r_0)|, |\mathscr{N}(r_0)| \leq C_{\mathbb{S}^1} \upzeta r_0^{-\upgamma}$. Furthermore, from \eqref{eq:P_ode_err}--\eqref{eq:Q_ode_err} of Corollary~\ref{cor:ode} with $a = a(r) = \frac{dx}{dr} \cdot (e^{\mu - \lambda})^{-1}$, one has
\begin{gather}
    \left| r \frac{d}{dr} \mathscr{P}(r) + \mathscr{P}(r) \mathscr{Q}(r) \right| \, (r, x) \leq \updelta r^{1/2}, \label{eq:P_ode_err_timelike} \\[0.6em]
    \left| r \frac{d}{dr} \mathscr{Q}(r) - \mathscr{Q}(r) \left( \mathscr{P}(r)^2 - 1 - \mathscr{Q}(r) \right) \right| \, (r, x) \leq \mathscr{Q}(r) \cdot \updelta r^{1/2}. \nonumber
\end{gather}

This is exactly what we need to apply the nonlinear ODE results of Lemma~\ref{lem:ode_bounce}. Therefore we have $(4 \upeta)^{-1} \leq \mathscr{P}(r) \leq 4 \upeta$ and $\mathscr{Q}(r) \leq 16 \upeta^2$. This is true for all timelike curves in our bootstrap spacetime, including all constant $x$ curves, therefore for all $(r, x) \in [r_b, r_0] \times \mathbb{S}^1$, one has
\begin{equation} \label{eq:pq_improvement}
    (4 \upeta)^{-1} \leq r \partial_r \phi (r, x) \leq 4 \upeta, \qquad \frac{Q^2}{r^2} e^{2\mu} (r, x) \leq 16 \upeta^2.
\end{equation}

We move onto $\mathscr{M}(r)$ and $\mathscr{N}(r)$. Using now \eqref{eq:M_ode_err}--\eqref{eq:N_ode_err} from Corollary~\ref{cor:ode}, we also have
\begin{gather*}
    \left| r \frac{d}{dr} \mathscr{M}(r) + \mathscr{Q}(r) \mathscr{M}(r) - \mathscr{P}(r) \mathscr{N}(r) \right| \, (r, x) \leq \updelta r^{1/2}, \\[0.6em]
    \left| r \frac{d}{dr} \mathscr{N}(r) + 2 \mathscr{P}(r) \mathscr{Q}(r) \mathscr{M}(r) + \mathscr{Q}(r) \mathscr{N}(r) \right| \, (r, x) \leq \updelta r^{1/2}.
\end{gather*}
Together with the initial data assumption $|\mathscr{M}(r_0)|, |\mathscr{N}(r_0)| \leq C_{\mathbb{S}^1} \upzeta r_0^{-\upgamma}$, this is what we need to apply Lemma~\ref{lem:ode_variation} -- note the change from $\upzeta$ to $C_{\mathbb{S}^1} \upzeta$ is harmless. Therefore for some $D > 0$ depending on $\upeta$, $\upgamma$, $\upzeta$ but not on $C_*$, we have $|\mathscr{M} (r)| + |\mathscr{N} (r)| \leq D r^{-\upgamma}$. This is true for all timelike curves $\gamma(r)$, so for all $(r, x) \in [r_b, r_0] \times \mathbb{S}^1$,
\begin{equation} \label{eq:mn_improvement}
    | \partial_x \mu (r, x) | \leq D r^{-\upgamma}, \qquad |r \partial_r \partial_x \phi (r, x) | \leq D r^{- \upgamma}.
\end{equation}

\medbreak \noindent
\underline{Step 2: Improving the bootstrap assumption \eqref{eq:bootstrap_maxwell}} 
\medbreak \noindent
By choosing $C_* \geq 32 \upeta^2$ and using \eqref{eq:pq_improvement}, we have already improved the first half of \eqref{eq:bootstrap_maxwell}. However, in order to improve the remaining parts of \eqref{eq:bootstrap_maxwell}, it will be necessary to upgrade the pointwise bound on $\frac{Q^2}{r^2} e^{2\mu}$ to an integrated bound on timelike curves.

More precisely, let $\gamma: [r_b, r_0] \to \R$ be a timelike curve as in Step 1. We rewrite the ODE-type estimate \eqref{eq:P_ode_err_timelike} in the following fashion, using crucially the lower bound for $r \partial_r \phi$ in \eqref{eq:pq_improvement}:
\[
    \left| r \frac{d}{dr} \log \mathscr{P}(r) + \mathscr{Q}(r) \right| \, (r, x) \leq 4 \upeta \updelta r^{1/2}.
\]
We now integrate this with respect to the differential $\frac{dr}{r}$. One finds:
\begin{equation*}
    \left| \int_r^{r_0} \mathscr{Q}(\tilde{r}) \frac{d\tilde{r}}{\tilde{r}} \right| \leq \left| \big [ \log \mathscr{P}(r) \big ]_r^{r_0} \right| + 4 \upeta \updelta \leq \log (16 \upeta^2) + 4 \upeta \updelta,
\end{equation*}
where the second inequality follows from the upper and lower bounds in \eqref{eq:pq_improvement}. We apply this estimate to the case where $\gamma(r) = (r, x)$ is a timelike curve of constant $x$; we have for $D_1 = \log(16\upeta^2) + 4 \upeta \updelta$ (which is crucially independent of $C_*$) the estimate
\begin{equation} \label{eq:Q_integral}
    \left| \int_r^{r_0} \frac{Q^2}{r^2} e^{2\mu} (\tilde{r}, x) \frac{d \tilde{r}}{\tilde{r}} \right| \leq D_1.
\end{equation} 

To apply this to improve bounds for $e^{2\mu}$ and $e^{2(\mu - \lambda)}$, we use the equations \eqref{eq:lambda_evol} and \eqref{eq:mu_evol} to derive that
\begin{gather}
    r \partial_r (r^{-1} e^{2\mu}) = \left( (r \partial_r \phi)^2 + r^2 e^{2(\mu - \lambda)} (\partial_x \phi)^2 + \kappa e^{2\mu} - \frac{Q^2}{r^2} e^{2\mu} \right) r^{-1} e^{2\mu} \geq \left( - e^{2\mu} - \frac{Q^2}{r^2} e^{2\mu} \right) r^{-1}e^{2\mu}, \label{eq:reu_eq} \\[0.6em]
    r \partial_r (r^{-1} e^{2(\mu - \lambda)}) = \left( 1 + 2 \kappa e^{2\mu} - \frac{2 Q^2}{r^2} e^{2\mu} \right) r^{-1} e^{2(\mu - \lambda)} \geq 2 \left( - e^{2\mu} - \frac{Q^2}{r^2} e^{2\mu} \right) r e^{2(\mu - \lambda)}. \label{eq:reum_eq}
\end{gather}

Using the bootstrap assumption \eqref{eq:bootstrap_maxwell} and that $r_* \geq r_0$ may be chosen small, let us choose $r_*$ so that alongside \eqref{eq:Q_integral}, we have that $\int^{r_0}_r e^{2\mu} (\tilde{r}, x) \frac{d \tilde{r}}{\tilde{r}} \leq 1$. Then integrating \eqref{eq:reu_eq} and \eqref{eq:reum_eq} along constant $x$-curves, as well as applying the initial data assumption \eqref{eq:data_linfty2},
\begin{gather*}
    r^{-1} e^{2\mu} (r, x) \leq e^{D_1 + 1} r_0^{-1} e^{2\mu}(r_0, x) \leq e^{D_1 + 1} \upzeta, \\[0.6em]
    r^{-1} e^{2(\mu - \lambda)} (r, x) \leq e^{2(D_1 + 1)} r_0^{-1} e^{2(\mu - \lambda)} (r_0, x) \leq e^{2(D_1 + 1)} \upzeta.
\end{gather*}
Thus for $C_*$ chosen larger than $2 e^{2(D_1 + 1)} \upzeta$, we may improve the latter part of \eqref{eq:bootstrap_maxwell}.

\medbreak \noindent
\underline{Step 3: Improving the bootstrap assumption \eqref{eq:bootstrap_rdr}} 
\medbreak \noindent

The bootstrap assumption for $|r \partial_r \phi|$ in \eqref{eq:bootstrap_rdr} is already improved due to \eqref{eq:pq_improvement}, so long as we choose $C_* \geq 8 \upeta$. To improve $|r \partial_r \lambda|$ and $|r \partial_r \mu|$, we use \eqref{eq:lambda_evol} and \eqref{eq:mu_evol}. Using Lemma~\ref{lem:low_order_linfty} to estimate the term involving $\partial_x \phi$, and using both \eqref{eq:pq_improvement} and $e^{2\mu} \leq e^{D_1 + 1} \upzeta r$ and $e^{2(\mu - \lambda)} \leq e^{2(D_1 + 1)} \upzeta r$ from Step 2, %we get:
\begin{align*}
    | r \partial_r \lambda | 
    &\leq \frac{1}{2} \left( (r \partial_r \phi)^2 + r^2 e^{2(\mu - \lambda)} (\partial_x \phi)^2 + 1 + e^{2\mu} + \frac{Q^2}{r^2} e^{2\mu} \right) \\[0.6em]
    &\leq \frac{1}{2} \left( 16 \upeta^2 + e^{2 (D_1 + 1)} \upzeta \updelta^2 r^{3 - 2 \upgamma} + 1 + e^{D_1 + 1} \upzeta r  + 16 \upeta^2 \right).
\end{align*}
By $r \leq r_0 \leq 1$, the right hand side is bounded independently of $C_*$, and thus by choosing $C_* \geq 2 \cdot \mbox{RHS}$, we have improved the bootstrap assumption for $|r \partial_r \lambda|$. The improvement for $|r \partial_r \mu|$ follows from an identical calculation.

\medbreak \noindent
\underline{Step 4: Improving the bootstrap assumptions \eqref{eq:bootstrap_dx} and \eqref{eq:bootstrap_d2x}} 
\medbreak \noindent

We complete the proof by improving the remaining bootstrap assumptions \eqref{eq:bootstrap_dx} and \eqref{eq:bootstrap_d2x}. We start with \eqref{eq:bootstrap_d2x}, noting that the bootstrap assumption for $|r \partial_r \partial_x \phi|$ is already improved due to \eqref{eq:mn_improvement}, so long as $C_* \geq 2D$.
To improve the bootstrap assumptions for $|r \partial_r \partial_x \lambda|$ and $|r \partial_r \partial_x \mu|$, we differentiate \eqref{eq:lambda_evol} and \eqref{eq:mu_evol} with respect to $x$, then combine several earlier estimates. Differentiating \eqref{eq:lambda_evol} yields
\[
    r \partial_r \partial_x \lambda = (r \partial_r \phi) (r \partial_r \partial_x \phi) + r^2 e^{2(\mu - \lambda)} \partial_x (\mu - \lambda) (\partial_x \phi)^2 + r^2 e^{2(\mu - \lambda)} \partial_x \phi \, \partial_x^2 \phi - \kappa e^{2\mu} \partial_x \mu + \frac{Q^2}{r^2} e^{2\mu} \partial_x \mu.
\]
Therefore using all of \eqref{eq:pq_improvement}, \eqref{eq:mn_improvement}, Lemma~\ref{lem:low_order_linfty} and the estimates for $e^{2\mu}$ and $e^{2(\mu - \lambda)}$ in Step 2, one yields
\[
    |r \partial_r \partial_x \lambda| \leq (4 \upeta D + e^{2(D_1 + 1)} \upzeta D^3 r^{3 - 2 \upgamma} + e^{2(D_1+ 1)} \upzeta D^2 r^{3 - \upgamma} + e^{D_1 + 1} \upzeta D r + 16\upeta^2 D ) r^{- \upgamma}.
\]
One can carry out an identical calculation for $|r \partial_r \partial_x \mu|$. Therefore for some $D_2$ independent of $C_*$, one has
\begin{equation} \label{eq:b4_improvement}
    |r \partial_r \partial_x \phi|, |r \partial_r \partial_x \lambda|, |r \partial_r \partial_x \mu| \leq D_2 r^{- \upgamma}.
\end{equation}

Finally, to improve \eqref{eq:bootstrap_dx}, we integrate \eqref{eq:b4_improvement} from $r = r_0$, and use the initial data assumption \eqref{eq:data_energy_boundedness} -- together with Sobolev embedding -- to get $|\partial_x \phi (r_0)| + |\partial_x \mu (r_0)| + |\partial_x \lambda (r_0)| \leq C_{\mathbb{S}^1} \upzeta r_0^{-\upgamma}$. Thereby integrating \eqref{eq:b4_improvement} yields:
\begin{equation} \label{eq:b3_improvement}
    |\partial_x \phi|, |r \partial_r \partial_x \lambda|, |r \partial_r \partial_x \mu| \leq (D_2 \upgamma^{-1} + C_{\mathbb{S}^1} \upzeta) r^{- \upgamma}.
\end{equation}

Thus upon choosing $C_* \geq 2 ( D_2 \upgamma^{-1} + C_{\mathbb{S}^1} \upzeta)$, \eqref{eq:b4_improvement} and \eqref{eq:b3_improvement} improve the remaining bootstrap assumptions \eqref{eq:bootstrap_d2x} and \eqref{eq:bootstrap_dx}. By the standard bootstrap argument, this concludes the proof of Theorem~\ref{thm:global}. (The two estimates \eqref{eq:global_linfty} and \eqref{eq:global_energy} follow from \eqref{eq:pq_improvement} and Proposition~\ref{prop:energy_hierarchy} respectively.) \qed

%auto-ignore

\section{BKL bounces}

\subsection{Proof of Theorem~\ref{thm:bounce}}

The assumptions of Theorem~\ref{thm:bounce} mean that our global existence result Theorem~\ref{thm:global} applies, as well as all the results of Section~\ref{sec:l2}. In particular Corollary~\ref{cor:ode} applies, and one derives the ODEs \eqref{eq:bounce_ode} as in the proof of Theorem~\ref{thm:global}, see e.g.~\eqref{eq:P_ode_err_timelike}; note here we allow $\updelta = 1$. 

It remains to prove the convergence results (\ref{item:bounce_i}) and (\ref{item:bounce_ii}). The case (\ref{item:bounce_i}) is easy, since $Q = 0$ implies $\mathscr{Q}_{\gamma} \equiv 0$, and the first equation in \eqref{eq:bounce_ode} yields $|r \partial_r \mathscr{P}_{\gamma}| \leq r^{1/2}$. Since $r^{-1/2}$ is integrable towards $r = 0$, this immediately yields that $\mathscr{P}_{\gamma}(r) \to \mathscr{P}_{\gamma, \infty}$ as $r \to \infty$ with $|\mathscr{P}_{\gamma}(r) - \mathscr{P}_{\gamma, \infty}| \leq 2 r^{1/2}$.

We move onto the more interesting (\ref{item:bounce_ii}). We first show that $\mathscr{Q}_{\gamma}(r) \to 0$ as $r \to 0$. For this purpose, we rearrange the first equation in \eqref{eq:bounce_ode} to yield
\begin{equation*}
    \mathscr{Q}_{\gamma} = - \mathscr{P}_{\gamma}^{-1} \cdot r \frac{d}{dr} \mathscr{P}_{\gamma}(r) + \mathscr{P}_{\gamma}^{-1} \cdot \mathscr{E}_{\mathscr{P}}.
\end{equation*}
Using from Theorem~\ref{thm:bounce} that $(4 \upeta)^{-1} \leq \mathscr{P}_{\upgamma} \leq 4 \upeta$, one thus yields that
\begin{equation} \label{eq:q_int}
    \int_{r}^{r_0} \mathscr{Q}_{\upgamma}(\tilde{r}) \, \frac{d\tilde{r}}{\tilde{r}} \leq \log( \frac{\mathscr{P}_{\gamma}(r)}{\mathscr{P}_{\gamma}(r_0)}) + 4 \upeta r_0^{1/2} \leq \log( 16 \upeta^2 ) + 4 \upeta r_0^{1/2}.
\end{equation}

In particular, this integral is finite, and it follows that $\mathscr{Q}_{\upgamma}(r)$ must tend to $0$ as $r \to 0$ in an averaged sense. In particular there is a sequence $\{ r_k \}$ with $r_k \to 0$ such that $\mathscr{Q}_{\gamma}(r_k) \to 0$ as $k \to \infty$. To upgrade this sequential convergence to convergence of $\mathscr{Q}_{\gamma}(r)$, we use the second equation in \eqref{eq:bounce_ode}. Since $\mathscr{P}_{\upgamma}^2 - 1 - \mathscr{Q} - \mathscr{E}_{\mathscr{Q}}$ is bounded, there is some constant $C$ such that upon integrating the second equation in \eqref{eq:bounce_ode}, for $0 < r < r_k$,
\begin{equation*}
    |\mathscr{Q}_{\gamma}(r) - \mathscr{Q}_{\gamma}(r_k)| \leq C \int_{r}^{r_k} \mathscr{Q}_{\gamma}(\tilde{r}) \frac{d \tilde{r}}{\tilde{r}}.
\end{equation*}
Since we have already established that the integral in \eqref{eq:q_int} is finite, the right hand side of this inequality tends to $0$ as $k \to \infty$. It thus follows that $\mathscr{Q}_{\gamma}(r) \to 0$ as $r \to 0$.

To show that $\mathscr{P}_{\gamma}(r)$ converges as $r \to 0$, we again use the boundedness of the integral \eqref{eq:q_int}. Using this alongside the boundedness of $\mathscr{P}_{\gamma}(r)$, one finds that the right hand side of the first equation in \eqref{eq:bounce_ode} is integrable towards $r = 0$ with respect to $\frac{dr}{r}$. Thus there exists some $\mathscr{P}_{\gamma, \infty}$ so that $\mathscr{P}_{\gamma}(r) \to \mathscr{P}_{\gamma, \infty}$ as $r \to 0$. To show that $\mathscr{P}_{\gamma, \infty} \geq 1$, the second equation in \eqref{eq:bounce_ode} can be rearranged to
\[
    \mathscr{Q}_{\gamma}^{-1} \cdot r \frac{d}{dr} \mathscr{Q}_{\gamma}(r) = \mathscr{P}_{\gamma}^2 - 1 - \mathscr{Q}_{\gamma} - \mathscr{E}_{\mathscr{Q}} \to \mathscr{P}_{\gamma, \infty}^2 - 1 \quad \text{ as } r \to 0.
\]
In particular, for $\mathscr{Q}_{\gamma}(r)$ to remain bounded as $r \to 0$ it must be that the right hand side of this is nonnegative i.e.~$\mathscr{P}_{\gamma}^2 \geq 1$. 

Finally, to show \eqref{eq:bounce_asymp}, we use the almost conserved quantity $\mathscr{K}$ encountered in the proof of Lemma~\ref{lem:ode_bounce}. Recall from this proof that defining $\mathscr{K}_{\gamma} = \mathscr{P}_{\gamma} + \mathscr{P}_{\gamma}^{-1} + \mathscr{Q}_{\gamma} \mathscr{P}_{\gamma}^{-1}$, one can show
\[
    \left| r \frac{d}{dr} \mathscr{K}_{\gamma} \right| \lesssim r^{1/2}.
\]
Using the convergence of $\mathscr{P}_{\gamma}$ and $\mathscr{Q}_{\gamma}$, it therefore holds that there exists $\mathscr{K}_{\gamma, \infty} = \mathscr{P}_{\gamma, \infty} + \mathscr{P}_{\gamma, \infty}^{-1}$ so that $\mathscr{K}_{\gamma}(r) \to \mathscr{K}_{\gamma, \infty}$ as $r \to 0$ and moreover,
\[
    \left| \mathscr{K}_{\gamma}(r) - \mathscr{K}_{\gamma, \infty} \right| \lesssim r^{1/2}.
\]
Substituting $r = r_0$, one therefore finds that
\[
    \left| (\mathscr{P}_{\gamma, \infty} + \mathscr{P}_{\gamma, \infty}^{-1}) - (\mathscr{P}_{\gamma}(r_0) + \mathscr{P}_{\gamma}^{-1}(r_0) + \mathscr{Q}_{\gamma}(r_0) \mathscr{P}_{\gamma}^{-1}(r_0)) \right| \lesssim r_0^{1/2}.
\]
Using this and $\mathscr{P}_{\gamma, \infty} \geq 1$, the estimate \eqref{eq:bounce_asymp} follows by solving a quadratic inequality. Note, of course, that \eqref{eq:bounce_asymp} is most interesting in the case that $\mathscr{Q}_{\gamma}(r_0)$ is small. This completes the proof of Theorem~\ref{thm:bounce}. \qed

\subsection{Proof of Corollary~\ref{cor:bounce}}

Finally, we apply Theorems~\ref{thm:global} and \ref{thm:bounce} to prove this stability / instability result. Let $(\phi, \mu, \lambda, Q)$ be as stated, so that by Theorem~\ref{thm:asymp_smooth} there exist smooth functions $\Psi(x)$, $\Xi(x)$, $M(x)$ and $\Lambda(x)$ on $\mathbb{S}^1$ so that:
\begin{gather*}
    \phi(r, x) = \Psi(x) \log r + \Xi (x) + \mathrm{Err}_{\phi}(r, x), \\[0.4em]
    \mu(r, x) = \tfrac{1}{2} (\Psi(x)^2 + 1) \log r + M (x) + \mathrm{Err}_{\mu}(r, x), \\[0.4em]
    \lambda(r, x) = \tfrac{1}{2} (\Psi(x)^2 - 1) \log r + \Lambda(x) + \mathrm{Err}_{\lambda}(r, x),
\end{gather*}
where the error terms and their $r \partial_r$-derivatives tend to $0$ as $r \to 0$ in the $C^{\infty}$ topology.

Since we also assume that $\Psi(x) > 0$ for all $x \in \mathbb{S}^1$ and $\mathbb{S}^1$ is compact, there exists some $\upeta \geq 2$ such that for some $0 < \tilde{r}_1 \leq r_1$, one has $(\upeta / 2)^{-1} \leq r \partial_r \phi(r, x) \leq \upeta / 2$ for all $0 < r < \tilde{r}_1$. Let $N$ be chosen (depending on $\upeta$) as in Theorem~\ref{thm:global}, $\upgamma < \frac{1}{100}$ be fixed, and consider the energy: at fixed $r > 0$
\[
    \mathcal{E}(r) = \frac{1}{2} \sum_{f \in \{ \phi, \lambda, \mu \}} \sum_{K = 0}^N \left( \int_{\mathbb{S}^1} (r \partial_r \partial_x^K f)^2 + r^2 e^{2(\mu - \lambda)} (\partial_x^{K+1} f)^2 + r^{2 \upgamma} (\partial_x^K f)^2 \, dx \right).
\]
By the asymptotics of Theorem~\ref{thm:asymp_smooth}, one has that 
\[
    \mathcal{E}(r) \to \frac{1}{2} \sum_{K = 0}^N \int_{S^1} \left( (\partial_x^K \Psi)^2 + (\partial_x^K (\tfrac{1}{2}(\Psi^2 + 1)))^2  + ( \partial_x^K (\tfrac{1}{2}(\Psi^2 - 1)))^2 \right) \, dx \quad \text{ as } r \to 0.
\]
So letting $\upzeta$ be four times this limit, and making $\tilde{r}_1$ smaller if necessary, it follows that $\mathcal{E}(r) \leq \upzeta/2$ for $0 < r < \tilde{r}_1$. Similarly using Theorem~\ref{thm:asymp_smooth} and increasing $\upzeta$ if necessary, one can guarantee $r^{-1} e^{2(\mu - \lambda)} \leq \upzeta/2$ and $r^{-1} e^{2\mu} \leq \upzeta/2$ for $0 < r < \tilde{r}_1$.

By these considerations, there will exist some $0 < r_0 < \tilde{r}_1$ so that if we abuse notation and \emph{redefine} the objects $(\phi_D, \mu_D, \lambda_D, \dot{\phi}_D, \dot{\mu}_D, \dot{\lambda}_D)$ as
\[
    (\phi_D, \mu_D, \lambda_D, \dot{\phi}_D, \dot{\mu}_D, \dot{\lambda}_D) = (\phi, \mu, \lambda, r \partial_r \phi, r \partial_r \mu, r \partial_r \lambda)|_{r = r_0},
\]
then $(\phi_D, \mu_D, \lambda_D, \dot{\phi}_D, \dot{\mu}_D, \dot{\lambda}_D)$ will obey the three assumptions \eqref{eq:data_linfty}, \eqref{eq:data_linfty2} and \eqref{eq:data_energy_boundedness} required for the applications of our results Theorem~\ref{thm:global} and Theorem~\ref{thm:bounce}, and in fact satisfy these assumptions with $\upeta$ and $\upzeta$ replaced by $\upeta/2$ and $\upzeta/2$ respectively. We also assume that $r_0$ is chosen small enough to satisfy $r_0 < r_*(\upeta, \upzeta)$.

We now prove Corollary~\ref{cor:bounce} with the initial data set given by this $(\phi_D, \mu_D, \lambda_D, \dot{\phi}_D, \dot{\mu}_D, \dot{\lambda}_D)$ at $r = r_0$ rather than the $r = r_1$ initial data stated. This is not an issue, as a standard Cauchy stability argument yields that perturbations of size $\varepsilon$ at $r = r_1$ correspond to perturbations of size $\tilde{\varepsilon}$ at $r = r_0$, with $\tilde{\varepsilon} \to 0$ as $\varepsilon \to 0$.

The remainder of the argument is then a direct application of Theorem~\ref{thm:global} and Theorem~\ref{thm:bounce}. Indeed, as stated in the corollary let $|Q| \leq \varepsilon$, and $\|(\tilde{f}_D, \tilde{\dot{f}}_D) - (f_D, \dot{f}_D) \|_{H^{N+1}\times H^N} \leq \varepsilon$ for $f \in \{\phi, \mu, \lambda\}$. Then for $\varepsilon$ sufficiently small, it follows that the perturbed data $(\tilde{\phi}_D, \tilde{\mu}_D, \tilde{\lambda}_D, \tilde{\dot{\phi}}_D, \tilde{\dot{\mu}}_D, \tilde{\dot{\lambda}}_D)$ still obeys the assumptions \eqref{eq:data_linfty}, \eqref{eq:data_linfty2} and \eqref{eq:data_energy_boundedness} for our choice of $\upeta$ and $\upzeta$. (Note the unperturbed background used $\upeta / 2$ and $\upzeta / 2$ so that the perturbed spacetime satisfies the required bounds with $\upeta$ and $\upzeta$.) 

Thus one can apply Theorem~\ref{thm:global} to the perturbed data, showing that (a) the resulting spacetime still features global existence towards $r = 0$. For (b), we use Theorem~\ref{thm:bounce}, applied to timelike curves $\gamma: r \mapsto (r, x)$ with constant $x$-coordinate. Then $\Psi(x)$ is then just $\mathscr{P}_{\gamma, \infty}$ for the original spacetime while $\tilde{\Psi}(x)$ is the same thing for the perturbed spacetime (which still exists due to Theorem~\ref{thm:bounce}), denoted by $\tilde{\mathscr{P}}_{\gamma, \infty}$.

To conclude, note that when $Q = 0$, by Theorem~\ref{thm:bounce}(\ref{item:bounce_i}) one has
\[
    \tilde{\mathscr{P}}_{\gamma, \infty} = \tilde{\dot{\phi}}_D + O(r_0^{1/2}) = \dot{\phi}_D + O(\varepsilon) + O(r_0^{1/2}) = \mathscr{P}_{\gamma, \infty} + O(\varepsilon) + O(r_0^{1/2}).
\]
By the aforementioned Cauchy stability argument, we are allowed to take $r_0$ as small as required, in particular we may choose the $O(r_0^{1/2})$ term to be less than $\frac{1}{2} \tilde{\varepsilon}$. Then choose $\varepsilon$ so that the $O(\varepsilon)$ term is also smaller than $\frac{1}{2} \tilde{\varepsilon}$, while maintaining that the size of the perturbation at $r = r_1$ is still less than $\varepsilon$. This yields (i).

When $Q \neq 0$, the argument is almost identical, except that by Theorem~\ref{thm:bounce}(\ref{item:bounce_ii}), one instead has
\[
    \tilde{\mathscr{P}}_{\gamma, \infty} = \max \{ \tilde{\dot{\phi}}_D, \tilde{\dot{\phi}}_D^{-1} \} + O(r_0^{1/2}) + O\left( \frac{Q^2}{r^2} e^{2\mu} (r_0)\right) = \max\{\mathscr{P}_{\gamma, \infty}, \mathscr{P}_{\gamma, \infty}^{-1}\} + O(\varepsilon) + O(r_0^{1/2}),
\]
where we used that $|Q| \leq \varepsilon$. (ii) therefore follows from Cauchy stability considerations as before. \qed

%auto-ignore
\appendix

\section{Asymptotics at high regularity} \label{app:asymp}

In Appendix~\ref{app:asymp}, we prove Theorem~\ref{thm:asymp_smooth}, regarding asymptotics of $(\phi, \mu, \lambda, Q)$ solving the surface symmetric Einstein--Maxwell--scalar system \eqref{eq:mu_evol}--\eqref{eq:phi_wave} near $\{ r = 0 \}$. Recall that we either take the case $Q = 0$ and $\kappa \in \{ 0, +1 \}$, or take \hyperlink{assump1}{Assumption 1} and \hyperlink{assump2}{Assumption 2} as given. We proceed in a series of lemmas.

\begin{lemma} \label{lem:expdecay}
    In either of the cases considered, there exist constants $C > 0$ and $\alpha > 0$ such that $\mu$ satisfies
    \begin{equation} \label{eq:expdecay}
        \frac{Q^2}{r^2} e^{2\mu} + e^{2\mu} \leq C r^{\alpha}.
    \end{equation}
    In particular, by the contnuation criterion Lemma~\ref{lem:continuation} one has global existence towards $\{ r = 0 \}$. 

    Furthermore, we also have the following estimates for $e^{\mu + \lambda}$:
    \begin{equation} \label{eq:expdecay2}
        \frac{Q^2}{r^2} e^{\mu + \lambda} + e^{\mu + \lambda} \leq C r^{\alpha - 1}.
    \end{equation}
\end{lemma}

\begin{proof}
    If $Q = 0$, by \eqref{eq:mu_evol} we have $r \partial_r \mu \geq \frac{1}{2}(1 + \kappa e^{2\mu})$. Thus, by integrating in the direction $r \downarrow 0$, if $\kappa \in \{ 0, +1 \}$, $e^{2\mu} \lesssim r$ is immediate. If instead $\kappa = -1$ we use \hyperlink{assump1}{Assumption 1} which tells us that $e^{2\mu}$ tends to $0$ uniformly in $r$. Therefore we can still show $e^{2\mu} \lesssim r^{1/2}$. So \eqref{eq:expdecay} holds with $\alpha = \frac{1}{2}$.

    If, on the other hand $Q \neq 0$, then we use \hyperlink{assump2}{Assumption 2}, which says that $r \partial_r \mu \geq \frac{1}{2}(2 + \alpha)$. Directly integrating this towards $r \downarrow 0$, the estimate \eqref{eq:expdecay} follows for $\alpha$ as in Assumption 2.

    The second estimate, \eqref{eq:expdecay2}, comes from combining the negative of \eqref{eq:mulambda_evol}, i.e.
    \[
        r \partial_r (\lambda - \mu) = - 1 - \kappa e^{2\mu} + \frac{Q^2}{r^2} e^{2\mu},
    \]
    with \eqref{eq:expdecay} to get $e^{\lambda - \mu} \lesssim r^{-1}$. We then multiply this with \eqref{eq:expdecay} to get \eqref{eq:expdecay2}.
\end{proof}

\begin{lemma}[Wave estimates] \label{lem:waveupper}
    Let $f$ obey the following equation, for $F$ some inhomogeneity
    \begin{equation} \label{eq:wave_f}
        (r \partial_r)^2 f - r^2 e^{2(\mu - \lambda)} \partial_x^2 f = \left( - \frac{Q^2}{r^2} + \kappa \right) e^{2\mu} r \partial_r f + r^2 e^{2(\mu - \lambda)} \partial_x(\mu - \lambda) \partial_x f + F. 
    \end{equation}
    Then for any $\upgamma > 0$ satisfying $\upgamma < \min \{ \frac{\alpha}{2}, \frac{1}{2} \}$ where $\alpha > 0$ is as in Lemma~\ref{lem:expdecay}, one has:
    \begin{multline} \label{eq:wave_f_est}
        | (r \partial_r + r e^{\mu - \lambda} \partial_x) f |(r, x) + | (r \partial_r - r e^{\mu - \lambda} \partial_x) f | (r, x) + r^{\upgamma} |f|(r, x) \\
        \lesssim
        r_0 \sup_{x \in \mathbb{S}^1} | \partial_r f(r_0, x) | + r_0 \sup_{x \in \mathbb{S}^1} |e^{\mu - \lambda} \partial_x f(r_0, x)| + r_0^{\upgamma} \sup_{x \in \mathbb{S}^1} |f(r_0, x)| + \int^{r_0}_{r} \sup_{\substack{\tilde{r} \leq s \leq r_0 \\ x \in \mathbb{S}^1}} |F(s, x)| \, \frac{d \tilde{r}}{\tilde{r}}.
    \end{multline}
\end{lemma}

\begin{proof}
    We introduce the notation $\sup_{H_{\tilde{r}}} |f|$ to denote $\sup_{x \in \mathbb{S}^1} |f(\tilde{r}, x)|$, and define
    \[
        A(\tilde{r}) = \max \left \{ \sup_{H_{\tilde{r}}} |( \partial_r + e^{\mu - \lambda}\partial_x ) f |, \, \sup_{H_{\tilde{r}}} | (\partial_r - e^{\mu - \lambda} \partial_x) f |, \, \sup_{H_{\tilde{r}}} r^{- 1 + \upgamma} |f| \right \} 
    \]
    We aim to derive an integral inequality for $A(r)$ to which we may apply Gr\"onwall's inequality. Using the wave equation \eqref{eq:wave_f} for $f$, as well as the transport equation \eqref{eq:mulambda_evol}, one may derive the following:
    \begin{gather*}
        (r \partial_r - r e^{\mu - \lambda} \partial_x) (\partial_r + e^{\mu - \lambda} \partial_x) f = - (\partial_r - e^{\mu - \lambda} \partial_x ) f + \left( - \frac{Q^2}{r^2} + \kappa \right) e^{2\mu} (\partial_r + e^{\mu - \lambda} \partial_x ) f + \frac{1}{r} F.
    \end{gather*}

    Therefore using Lemma~\ref{lem:expdecay}, one has
    \[
        | (r \partial_r - r e^{\mu - \lambda} \partial_x) (\partial_r + e^{\mu - \lambda} \partial_x) f (\tilde{r}, x) | \geq - A (\tilde{r}) - C \tilde{r}^{\alpha} A (\tilde{r}) - \frac{1}{\tilde{r}} \sup_{H_{\tilde{r}}} |F|.
    \]
    The idea is now simply to integrate this towards $r = 0$, via the characteristic (null) vector field $\underline{L} = r \partial_r - r e^{\mu - \lambda} \partial_x$. Indeed, let $\gamma_{\bar{L}}: (0, r_0] \to \mathcal{Q}$ be a curve whose tangent vector is $\underline{L}$ and which is parameterized by $r$. That is, $\gamma_{\underline{L}}(r) = (r, X(r))$ with $\frac{dX}{dr} = - e^{\mu - \lambda}$. Then integrating the above yields
    \[
        |((\partial_r + e^{\mu - \lambda} \partial_x) f) (\gamma_{\underline{L}}(r)) | \leq
        |((\partial_r + e^{\mu - \lambda} \partial_x) f) (\gamma_{\underline{L}}(r_0)) | + \int^{r_0}_r (1 + C \tilde{r}^{\alpha} ) A (\tilde{r}) \, \frac{d \tilde{r}}{\tilde{r}} + \int_r^{r_0} \frac{1}{\tilde{r}^2} \sup_{H_{\tilde{r}}} |F| \, d \tilde{r}.
    \]

    Since this is true for all curves $\gamma_{\underline{L}}$ which are integral curves of $\underline{L}$, we therefore have that for all $x \in \mathbb{S}^1$:
    \[
        |(\partial_r + e^{\mu - \lambda} \partial_x) f)(r, x)| \leq
        \sup_{H_{r_0}} |(\partial_r + e^{\mu - \lambda} \partial_x) f| + \int^{r_0}_r (1 + C \tilde{r}^{\alpha} ) A (\tilde{r}) \, \frac{d \tilde{r}}{\tilde{r}} + \int_r^{r_0} \frac{1}{\tilde{r}^2} \sup_{H_{\tilde{r}}} |F| \, d \tilde{r}.
    \]
    On the other hand, a similar computation involving integral curves of $L = r \partial_r + r e^{\mu - \lambda} \partial_x$ yields:
    \[
        |(\partial_r - e^{\mu - \lambda} \partial_x) f)(r, x)| \leq
        \sup_{H_{r_0}} |(\partial_r - e^{\mu - \lambda} \partial_x) f| + \int^{r_0}_r (1 + C \tilde{r}^{\alpha} ) A (\tilde{r}) \, \frac{d \tilde{r}}{\tilde{r}} + \int_r^{r_0} \frac{1}{\tilde{r}^2} \sup_{H_{\tilde{r}}} |F| \, d \tilde{r}.
    \]
    Finally, since $ \partial_r (r^{- 1 + \upgamma} f) = - (1 - \upgamma) r^{- 2 + \upgamma} f + \frac{1}{2} r^{- 1 + \upgamma} ( \partial_r + e^{\mu - \lambda} \partial_x) f + \frac{1}{2} r^{- 1 + \upgamma} ( \partial_r - e^{\mu - \lambda}\partial_x)f$,
    \[
        r^{- 1 + \upgamma} |f| (r, x) \leq \sup_{H_{r_0}} ( r^{- 1 + \upgamma} |f| ) + \int^{r_0}_{r} \left( 1 - \upgamma + \tilde{r}^{\upgamma} \right) A(\tilde{r}) \, \frac{d \tilde{r}}{\tilde{r}}.
    \]

    Combining all of these, we deduce that
    \begin{align*}
        A(r) 
        &\leq A(r_0) + \int^{r_0}_r (1 + C \tilde{r}^{\alpha} + \tilde{r}^{\upgamma} ) A(\tilde{r}) \frac{d \tilde{r}}{\tilde{r}} + \int^{r_0}_r \frac{1}{\tilde{r}^2} \sup_{H_{\tilde{r}}} |F| \, d\tilde{r} \\
        &\leq A(r_0) + \int^{r_0}_r (1 + C \tilde{r}^{\alpha} + \tilde{r}^{\upgamma} ) A(\tilde{r}) \frac{d \tilde{r}}{\tilde{r}} + \frac{1}{r} \sup_{\substack{\tilde{r} \leq s \leq r_0 \\ x \in \mathbb{S}^1}} |F(s, x)|.
    \end{align*}
    One now applies Gr\"onwall's inequality together with $\int^{r_0}_r (1 + C \tilde{r}^{\alpha} + \tilde{r}^{\upgamma}) \frac{d\tilde{r}}{\tilde{r}} \leq \log( \frac{r_0}{r} ) + \log D$ where $D > 0$ is some constant, to find that 
    \[
        r A(r) \lesssim r_0 A(r_0) + \int^{r_0}_{r} \frac{1}{\tilde{r}} \sup_{\substack{\tilde{r} \leq s \leq r_0\\ x \in \mathbb{S}^1}}|F(s, x)| \, d \tilde{r},
    \]
    which is equivalent to \eqref{eq:wave_f_est}.
\end{proof}

\begin{lemma}[Upper bounds at all orders] \label{lem:linfty_upper}
    For $(\phi, \lambda, \mu)$ as in Theorem~\ref{thm:asymp_smooth}, for all $0 \leq j \leq k$ there exists a constant $C_j$ such that
    \begin{equation} \label{eq:linfty_smooth}
        \max_{f \in \{\phi, \mu, \lambda \}} |r \partial_r \partial_x^j f| + 
        \max_{f \in \{\phi, \mu, \lambda \}} |r e^{\mu - \lambda} \partial_x^{j+1} f| \leq C_j. 
    \end{equation}
    Moreover $C_j$ depends only on up to $j$ $\partial_x$-derivatives of $r \partial_r \phi$, $r \partial_r \mu$ and $r \partial_r \lambda$ and up to $j+1$ $\partial_x$-derivatives of $\phi$, $\mu$ and $\lambda$ at data i.e.~at $r = r_0$.
\end{lemma}

\begin{proof}
    For the sake of brevity, we will focus on proving the part of \eqref{eq:linfty_smooth} corresponding to $f = \phi$, and mention at the end how one gets the corresponding estimates for $\mu$ and $\lambda$.

    Due to \eqref{eq:phi_wave}, $\phi$ satisfies the wave equation \eqref{eq:wave_f} with $F = 0$. Therefore Lemma~\ref{lem:waveupper} immediately yields that for some $C_0 > 0$ depending only on the $C^0$ norm of $r \partial_r \phi$ and $\partial_x \phi$ at data, one has
    \begin{equation} \label{eq:linfty_smooth_0}
        |r \partial_r \phi| + |r e^{\mu - \lambda} \partial_x \phi| \leq C_0.
    \end{equation}

    The most challenging step is to get the estimate \eqref{eq:linfty_smooth} for $j = 1$. For convenience, we further introduce the notation $\sup_{\mathcal{D}_{\tilde{r}, r_0}} |f| = \sup_{\tilde{r} \leq s \leq r_0, x \in \mathbb{S}^1} |F(s, x)|$, and define the quantity
    \[
        B(\tilde{r}) = \max \left \{ \sup_{\mathcal{D}_{\tilde{r}, r_0}} |r \partial_r \partial_x \phi|, \, \sup_{\mathcal{D}_{\tilde{r}, r_0}} |r e^{\mu - \lambda} \partial_x^2 \phi|, \, \sup_{\mathcal{D}_{\tilde{r}, r_0}} r^{\upgamma} |\partial_x \phi| \right\}.
    \]
    It will be necessary to first find some preliminary estimates for $\partial_x(\mu - \lambda)$ and $\partial_x^2(\mu - \lambda)$. To do so, we use the transport equation \eqref{eq:mulambda_evol} for $\mu - \lambda$ and differentiate in $x$, to get:
    \begin{equation} \label{eq:mulambda_evol_x}
        r \partial_r \partial_x (\mu - \lambda) = \left( - \frac{Q^2}{r^2} + \kappa \right) e^{2\mu} (2 \partial_x \mu) = \left( - \frac{Q^2}{r^2} + \kappa \right) e^{2\mu} 2 (r \partial_r \phi) (\partial_x \phi).
    \end{equation}
    Note that the second equality followed from the constraint equation \eqref{eq:mu_constraint}.

    We then estimate the right hand side of \eqref{eq:mulambda_evol_x} by:
    \[
        \left| \left( - \frac{Q^2}{r^2} + \kappa \right) e^{2 \mu} 2 (r \partial_r \phi) (\partial_x \phi) \right|
        \leq 
        \left| r^{-1} \left( \frac{Q^2}{r^2} + 1 \right) e^{\mu + \lambda} \cdot 2 r \partial_r \phi \cdot r e^{\mu - \lambda} \partial_x \phi \right|
        \leq
        r^{-1} 2 C_0^2 \cdot \left( \frac{Q^2}{r^2} + 1 \right) e^{\mu + \lambda},
    \]
    for $C_0$ as in \eqref{eq:linfty_smooth_0}. Then using \eqref{eq:expdecay2} in Lemma~\ref{lem:expdecay}, we have that $|r \partial_r \partial_x(\mu - \lambda)| \lesssim r^{-2 + \alpha}$. Integrating, we get
    \begin{equation} \label{eq:mulambda_x}
        | \partial_x ( \mu - \lambda ) | \lesssim r^{-2 + \alpha}.
    \end{equation}

    Differentiating \eqref{eq:mulambda_evol_x} once again in $x$, we get
    \[
        r \partial_r \partial_x^2 (\mu - \lambda) = \left( - \frac{Q^2}{r^2} + \kappa \right) e^{2\mu} \left[ 4(r \partial_r \phi)^2 (\partial_x \phi)^2 + 2 (r \partial_r \partial_x \phi) (\partial_x \phi) + 2 ( r \partial_r \phi) (\partial_x^2 \phi) \right].
    \]
    Using the definition of $B(r)$ and similar methods to that of the estimate for \eqref{eq:mulambda_x}, we can estimate $|r \partial_r \partial_x^2 (\mu - \lambda)| \lesssim r^{- 2 + \alpha} [ r^{- \upgamma} B(r) + B(r) ]$. Since $B(r)$ increases as $r \to 0$, we integrate this towards $r = 0$, and get
    \begin{equation} \label{eq:mulambda_xx}
        | \partial_x^2 ( \mu - \lambda ) | \lesssim r^{- 2 + \alpha - \upgamma} B(\tilde{r}) + \sup_{H_{r_0}} |\partial_x^2 ( \mu - \lambda )|.
    \end{equation}

    We now commute the wave equation \eqref{eq:phi_wave} with $\partial_x$ to get that
    \begin{multline*}
        (r \partial_r)^2 \partial_x \phi + r^2 e^{2(\mu - \lambda)} \partial_x^3 \phi = \left( - \frac{Q^2}{r^2} + \kappa \right) e^{2\mu} (r \partial_r \partial_x \phi) + r^2 e^{2(\mu - \lambda)} \partial_x (\mu - \lambda) \partial_x^2 \phi \\
        + 2 r^2 e^{2(\mu - \lambda)} \partial_x (\mu - \lambda) \partial_x^2 \phi + r^2 e^{2(\mu - \lambda)} \partial_x^2 (\mu - \lambda) \partial_x \phi + 2 r^2 e^{2(\mu - \lambda)} ( \partial_x (\mu - \lambda) )^2 \partial_x \phi \\ + \left( - \frac{Q^2}{r^2} + \kappa \right) e^{2\mu} (2 \partial_x \mu) ( r \partial_r \phi).
    \end{multline*}
    Denote the latter two lines of this commuted equation by $F^{(1)}$. To estimate $F^{(1)}$, we will use \eqref{eq:mulambda_x}, \eqref{eq:mulambda_xx}, and finally also $e^{\mu - \lambda} \lesssim r$ (which follows easily from \eqref{eq:mulambda_evol} and \eqref{eq:expdecay}). We have therefore
    \begin{gather*}
        |r^2 e^{2 (\mu - \lambda)} \partial_x (\mu - \lambda) \partial_x^2 \phi | \lesssim r^{1 + \alpha} e^{\mu - \lambda} |\partial_x^2 \phi| \lesssim r^{\alpha} B(r), \\
        |r^2 e^{2(\mu - \lambda)} (\partial_x(\mu - \lambda))^2 \partial_x \phi | \lesssim r^4 \cdot r^{2(-2+\alpha)} |\partial_x \phi| \lesssim r^{2\alpha - \upgamma} B(r), \\
        |r^2 e^{2(\mu - \lambda)} \partial_x^2 (\mu - \lambda) \partial_x \phi | = |r e^{\mu - \lambda} \partial_x^2(\mu - \lambda)| \cdot |r e^{\mu - \lambda} \partial_x \phi| \lesssim r^{\alpha - \upgamma} B(r),
    \end{gather*}
    while by Lemma~\ref{lem:expdecay}, the constraint equation \eqref{eq:mu_evol} and \eqref{eq:linfty_smooth_0} we have:
    \begin{gather*}
        \left| \left( - \frac{Q^2}{r^2} + \kappa \right) e^{2\mu} (2 \partial_x \mu) (r \partial_r \phi) \right| \lesssim \left( \frac{Q^2}{r^2} + 1 \right) e^{2\mu} \cdot |\partial_x \phi | \lesssim r^{\alpha - \upgamma} B(r).
    \end{gather*}

    The upshot is that $|F^{(1)}(r)| \lesssim r^{\alpha - \upgamma} B(r)$. Therefore using the commuted equation and Lemma~\ref{lem:waveupper}, one shows that $B(r) \lesssim B(r_0) + \int^r_{r_0} \tilde{r}^{\alpha - \upgamma - 1} B(\tilde{r}) d \tilde{r}$. Since $0 < \upgamma < \alpha$, one may apply Gr\"onwall's inequality to show that $B(r)$ is uniformly bounded by some constant $C_1$, thereby \eqref{eq:linfty_smooth} holds for $j=1$ and $f = \phi$.

    There are two things left to do: get estimates for $\mu$ and $\lambda$, and then get estimates for higher derivatives i.e.~$j \geq 1$. For $\mu$, one uses the wave equation for $\mu$, equation \eqref{eq:mu_wave}. Even for the uncommuted equation \eqref{eq:mu_wave}, there is already an inhomogeneity in the application of Lemma~\ref{lem:waveupper}. Fortunately, this is easily dealt with, either using Lemma~\ref{lem:expdecay}, or using $\partial_x \phi \leq C_1 r^{- \upgamma}$. Commuting the equation \eqref{eq:mu_wave} once with $\partial_x$ and applying similar methods to before, we obtain the desired estimates at orders $j = 0, 1$ for $\mu$.

    Next, for $\lambda$, we actually already have estimates for $\mu - \lambda$, see for instance \eqref{eq:mulambda_evol_x}, \eqref{eq:mulambda_x} and \eqref{eq:mulambda_xx}. These can easily be used to get the desired estimates for $\lambda$ at orders $j = 0, 1$. Finally, for the higher derivatives one simply has to commute the equations \eqref{eq:phi_wave}, \eqref{eq:mu_wave} and \eqref{eq:mulambda_evol} with further $\partial_x$-derivatives, and carry out the same approach. We remark that the commuted equations already appear in Section~\ref{sec:l2}, and that in the context of Theorem~\ref{thm:asymp_smooth} all inhomogeneous terms that appear can be easily controlled given we already have the estimates at orders $j = 0, 1$. Details are left to the reader.
\end{proof}

\begin{lemma}[Completion of Theorem~\ref{thm:asymp_smooth}]
    For $(\phi, \lambda, \mu)$ as in Theorem~\ref{thm:asymp_smooth}, there exists a function $\Psi: \mathbb{S}^1 \to \R$ such that the following convergence holds strongly as $r \to 0$, 
    \begin{equation} \label{eq:rdrf}
        r \partial_r \phi(r, \cdot) \to \Psi(\cdot), \quad
        r \partial_r \mu(r, \cdot) \to \frac{1}{2}( \Psi(\cdot)^2 + 1 ), \quad
        r \partial_r \lambda(r, \cdot) \to \frac{1}{2}( \Psi(\cdot)^2 - 1 ) \quad \text{ in } C^{k-1}.
    \end{equation}
    Moreover, there exist further functions $\Xi, M, \Lambda: \mathbb{S}^1 \to \R$ such that
    \begin{equation} \label{eq:rdrf-}
        \phi(r, \cdot) - r \partial_r \phi(r, \cdot) \log r \to \Xi(\cdot), \quad
        \mu(r, \cdot) -  r \partial_r \mu(r, \cdot) \log r \to M(\cdot), \quad
        \lambda(r, \cdot) - r \partial_r \lambda(r, \cdot) \log r \to \Lambda(\cdot) \quad \text{ in } C^{k-1}.
    \end{equation}
    Finally, the asymptotic constraint equation $d M = \Psi d \, \Xi$ holds.
\end{lemma}

\begin{proof}
    Let us first study the convergence of $r \partial_r \phi$. To do so, use the wave equation \eqref{eq:phi_wave}, which says
    \begin{equation} \label{eq:phi_wave_copy}
        (r \partial_r) (r \partial_r \phi) = r^2 e^{2(\mu - \lambda)} \partial_x^2 \phi + r^2 e^{2(\mu - \lambda)} \partial_x (\mu - \lambda) \partial_x \phi + \left( - \frac{Q^2}{r^2} + \kappa \right) e^{2\mu} r \partial_r \phi.
    \end{equation}
    By Lemma~\ref{lem:linfty_upper} for $j = 1$, noting that integrating the bound for $|r \partial_r \partial_x f|$ in \eqref{eq:linfty_smooth} yields that $|\partial_x f | \lesssim 1+ |\log r|$, and also using Lemma~\ref{lem:expdecay} as well as the bound $e^{\mu - \lambda} \lesssim r$ (found by integrating \eqref{eq:mulambda_evol} and Lemma~\ref{lem:expdecay}),
    \begin{align*}
        |(r \partial_r) (r \partial_r \phi)| 
        &\leq r^2 e^{2(\mu - \lambda)} |\partial_x^2 \phi| + r^2 e^{2(\mu - \lambda)} |\partial_x (\mu - \lambda)| |\partial_x \phi| + \left( \frac{Q^2}{r^2} + 1 \right) e^{2\mu} |r \partial_r \phi| \\
        &\lesssim r^2 + r^4 (1 + |\log r|)^2 + r^{\alpha}.
    \end{align*}
    Crucially, the right hand side is integrable with respect to the differential $\frac{dr}{r}$ as $r \to 0$, thus $r \partial_r \phi(r, \cdot)$ converges to some function $\Psi(\cdot)$ as $r \to 0$, at least in the $C^0$ norm.

    To upgrade $C^0$ convergence to $C^{k-1}$ convergence, one differentiates \eqref{eq:phi_wave_copy} $k-1$ times with $\partial_x$. Using Lemma~\ref{lem:linfty_upper} for all orders, and noting that extra powers of $(1 + |\log r|)$ will arise when $\partial_x$ hits either $e^{2(\mu - \lambda)}$ or $e^{2\mu}$, one obtains:
    \[
        |(r \partial_r) (r \partial_r \partial_x^k \phi)| \lesssim (r^2 + r^4 (1 + |\log r|)^2 + r^{\alpha}) (1 + |\log r|)^{k-1}.
    \]
    The right hand side remains integrable with respect to $\frac{dr}{r}$, hence the convergence of $r \partial_r \phi(r, \cdot)$ to $\Psi(r, \cdot)$ is actually in $C^k$, as required.

    Next, we move to the convergence of $r \partial_r \mu$ and $r \partial_r \lambda$. From \eqref{eq:mu_evol}
    \[
        r \partial_r \mu = \frac{1}{2} \left( (r \partial_r \phi)^2 + 1 \right) + \frac{1}{2} r^2 e^{2(\mu - \lambda)} (\partial_x \phi)^2 + \frac{1}{2} \left( - \frac{Q^2}{r^2} + \kappa \right) e^{2\mu}.
    \]
    By Lemma~\ref{lem:linfty_upper} for $j=1$ and Lemma~\ref{lem:expdecay} respectively, the latter two terms on the right hand side decay to $0$ as $r^4 (1 + |\log r|)^2$ and $r^{\alpha}$ as $r \to 0$. This remains true even after one takes $j$ $\partial_x$-derivatives, for $0 \leq j \leq k - 1$ (though as before we introduce an additional factor of $(1 + |\log r|)^j$. Since we already have convergence of $r \partial_r \phi$ to $\Psi$, the convergence for $r \partial_r \mu$ follows immediately. The convergence of $r \partial_r \lambda$ is analogous, using \eqref{eq:lambda_evol}.

    To get \eqref{eq:rdrf-}, we use that
    \[
        (r \partial_r) (\phi - (r \partial_r \phi) \log r) = (r \partial_r)(r \partial_r \phi) \cdot \log r.
    \]
    But we can compute the right hand side from \eqref{eq:phi_wave_copy}, and therefore we can bound $|(r \partial_r) [\partial_x^j(\phi - (r \partial_r \phi))]|$ by $(r^2 + r^4(1 + |\log r |)^2 + r^{\alpha}) (1 + |\log r|)^j \log r$ for all $0 \leq j \leq k-1$. This remains integrable with respect to $\frac{dr}{r}$ towards $r = 0$, hence there is some $\Xi(\cdot)$ so that $\phi(r, \cdot) - r \partial_r \phi(r, \cdot) \log r \to \Xi(\cdot)$ in $C^{k-1}$.

    The analogous convergence for $\mu - r \partial_r \mu \log r$ and $\lambda - r \partial_r \lambda \log r$ is similar, though perhaps a little more involved, since it involves taking an $r \partial_r$-derivative of \eqref{eq:lambda_evol}--\eqref{eq:mu_evol} in order to bound $(r \partial_r)(r \partial_r \mu) \cdot \log r$ and $(r \partial_r) (r \partial_r \lambda) \cdot \log r$. We leave the details to the reader.

    Finally, we explain how to get the asymptotic momentum constraint $d M = \Psi d \, \Xi$. Using the constraint equation \eqref{eq:mu_constraint} for $\partial_x \mu$, we have:
    \begin{align*}
        \partial_x \left( \mu - (r \partial_r \mu) \log r \right) 
        &= (r \partial_r \phi) (\partial_x \phi) - (r \partial_r) [ (r \partial_r \phi) (\partial_x \phi) ] \cdot \log r \\
        &= (r \partial_r \phi) \cdot \partial_x \left( \phi - (r \partial_r \phi) \log r \right) - (r \partial_r)(r \partial_r \phi) \cdot (\partial_x \phi) \log r.
    \end{align*}
    By the $C^1$ convergence in \eqref{eq:rdrf}, the left hand side and the first term on the right hand side converge to $\partial_x M$ and $\Psi \partial_x \Xi$ respectively as $r \to 0$. Finally, by \eqref{eq:phi_wave_copy} and its bounds the final term is bounded by $(r^2 + r^4(1 + |\log r|)^2 + r^{\alpha}) (\log r)^2$, which converges to $0$ as $r \to 0$. This completes the proof of the asymptotic momentum constraint.
\end{proof}

\section{Proof of the weighted product estimate} \label{app:weightedl2}

In Appendix~\ref{app:weightedl2} we provide a proof of the weighted $L^2$ product estimate Lemma~\ref{lem:weightedl2}. By a simple density argument, it suffices to assume that $f$ and $g$ are smooth, and for convenience we work in $L^p$--spaces with respect to the measure $w^2 dx$ i.e.~with
\begin{equation*}
    \| f \|_{L^p(w^2 dx)} = \left( \int_{\mathbb{S}^1} |f(x)|^p w^2(x) \, dx \right)^{\frac{1}{p}}
\end{equation*}

For $0 \leq k \leq K$, we shall moreover define $p_k = \frac{2K}{k}$ (where $p_0 = \infty$) and the quantities $F_k$ and $G_k$ as:
\begin{gather*}
    F_0 = \| f \|_{L^{\infty}(w^2 dx)}, \qquad F_k = \sum_{i=0}^{k-1} \| \partial_x^{k-i} f \|_{L^{p_k}(w^2 dx)} \| \partial_x W \|_{L^{\infty}(w^2 dx)}^i \text{ for } 1 \leq k \leq K, \\[0.3em]
    G_0 = \| g \|_{L^{\infty}(w^2 dx)}, \qquad G_k = \sum_{i=0}^{k-1} \| \partial_x^{k-i} f \|_{L^{p_k}(w^2 dx)} \| \partial_x W \|_{L^{\infty}(w^2 dx)}^i \text{ for } 1 \leq k \leq K.
\end{gather*}
With this notation, Lemma~\ref{lem:weightedl2} is reduced to proving that
\begin{equation} \label{eq:weightedl2_0}
    \| \partial_x^M f \, \partial_x^N g \|_{L^2(w^2 dx)} \lesssim_{M, N} F_0 G_k + G_0 F_k.
\end{equation}

In order to prove such an estimate, our first goal will be to show that for $0 \leq k \leq K$, one has
\begin{equation} \label{eq:weightedl2_interp}
    F_k \lesssim F_0^{\frac{K-k}{K}} F_K^{\frac{k}{K}}.
\end{equation}
We do this in several steps. Firstly, for $0 < k < K$ one may use the identity 
\begin{multline*}
    \partial_x (w^2 |\partial_x^k f|^{p_k - 2} \, \partial_x^k f \cdot \partial_x^{k-1} f) = w^2 |\partial_x^k f|^{p_k} 
    \\[0.2em] + w^2 (p_k-1) |\partial_x^k f|^{p_k-3} \, \partial_x^{k+1} f \cdot \partial_x^k f \cdot \partial_x^{k-1} f + 2 w^2 \partial_x W \cdot |\partial_x^k f|^{p_k-2} \partial_x^k f \cdot \partial_x^{k-1} f.
\end{multline*}

Integrating this identity over $x \in \mathbb{S}^1$, one thereby deduces that
\[
    \int_{\mathbb{S}^1} |\partial_x^k f|^{p_k} w^2 \,dx \lesssim
    \int_{\mathbb{S}^1} |\partial_x^k f|^{p_k-2} \, |\partial_x^{k+1} f| \, |\partial_x^{k-1} f| w^2 \, dx + \int_{\mathbb{S}^1} |\partial_x^k f|^{p_k - 1} \, |\partial_x^{k-1} f| \, |\partial_x W| w^2\,dx,
\]
and an application of H\"olders inequality to the two integrals on the right hand side yields (all $L^p$ spaces are with respect to the measure $w^2 dx$):
\[
    \| \partial_x^k f\|_{L^{p_k}}^{p_k} \lesssim \| \partial_x^k f \|_{L^{p_k}}^{p_k-2} \| \partial_x^{k-1} f \|_{L^{p_{k-1}}} \left( \| \partial_x^{k+1} f \|_{L^{p_{k+1}}} + \| \partial_x^k f \|_{L^{p_{k+1}}} \| \partial_x W \|_{L^{\infty}} \right).
\]
Thus by the definition of $F_k$ one concludes that
\[
    \| \partial_x^k f \|_{L^{p_k}}^2 \lesssim F_{k-1} F_{k+1}.
\]

On the other hand, a standard $L^p$ interpolation estimate (where the implied constants are importantly independent of the measure associated to the $L^p$ space), one has that for $1 \leq j \leq k - 1$,
\[
    \| \partial_x^{k-j} f \|_{L^{p_k}}^2 \| \partial_x W \|_{L^{\infty}}^{2j} \lesssim \| \partial_x^{k-j} f \|_{L^{p_{k-1}}} \| \partial_x W \|_{L^{\infty}}^{j-1} \cdot \|\partial_x^{k-j} f \|_{L^{p_{k+1}}} \| \partial_x W \|_{L^{\infty}}^{j+1} \leq F_{k-1} F_{k+1}.
\]
Combining the above two equations, one therefore has that $F_{k}^2 \lesssim F_{k-1} F_{k+1}$. 

It is straightforward to go from this to \eqref{eq:weightedl2_interp}. For instance, assuming that all the $F_j$ are nonzero the inequality $F_j^2 \lesssim F_{j-1} F_{j+1}$ implies there is some constant $C \in \R$ such that
\begin{align*}
    - \frac{K - k}{K} \log F_0 - \frac{k}{K} \log F_K + F_k 
    &= \sum_{j=1}^{K-1} \left( - \log \frac{1}{2} F_{j-1} - \frac{1}{2} \log F_{j+1} + \log F_j \right) \cdot \left( \frac{\min\{ j(K-k), (K-j)k\}}{K}\right) \\
    &\leq C,
\end{align*}
which immediately yields \eqref{eq:weightedl2_interp}.

Finally, in order to deduce \eqref{eq:weightedl2_0} from \eqref{eq:weightedl2_interp} (and the analogous inequality for the $G_k$s) we simply apply H\"older's inequality followed by \eqref{eq:weightedl2_interp}, as follows:
\[
    \| \partial^M_x f \, \partial^N_x g \|_{L^2} \leq \| \partial^M_x f \|_{L^{p_M}} \| \partial^N_x g \|_{L^{p_N}} \leq F_M G_N \leq (F_0 G_K)^{\frac{N}{K}} (F_K G_0)^{\frac{M}{K}}.
\]
The conclusion thus follows from Young's inequality. \qed

\bibliography{bibliography_master.bib}
\bibliographystyle{abbrvnat_mod}

\end{document}